\documentclass[sigconf, nonacm]{acmart}
\usepackage[show]{macros}
\usepackage{tcolorbox}
\SetKw{Break}{Break}

\usepackage{hyperref} 
\newcommand\vldbdoi{XX.XX/XXX.XX}
\newcommand\vldbpages{XXX-XXX}
\newcommand\vldbvolume{16}
\newcommand\vldbissue{3}
\newcommand\vldbyear{2022}
\newcommand\vldbauthors{\authors}
\newcommand\vldbtitle{\shorttitle} 
\newcommand\vldbavailabilityurl{https://github.com/joint-em/ftcs}
\newcommand\vldbpagestyle{empty}
\usepackage{graphicx}
\usepackage{tikz}
\usepackage{tikz-network}
\usepackage{graphics}
\usepackage{tikz,pgfplots,pgfplotstable}
\usepackage{mathtools}
\pgfplotsset{compat=newest}

\newcommand{\sft}[1]{\mathsf{SFT}(#1)}

\newcommand{\attr}[0]{\mathcal{A}}
\newcommand{\new}[1]{\textcolor{black}{#1}}

\definecolor{mycolor}{rgb}{0,0.5,0.8}
\definecolor{mycolor2}{rgb}{0.9,0.0,0.0}
\definecolor{mygreen}{rgb}{0.0,0.55,0.0}
\definecolor{mynewgreen}{rgb}{0.76,0.7,0.5}
\definecolor{myred}{rgb}{0.55,0.0,0.0}
\definecolor{mygrey}{rgb}{0.52,0.52,0.51}
\definecolor{plotcolor}{rgb}{0.95,0.66,0.34}
\definecolor{textcolor1}{rgb}{0.0, 0.43, 0.64} 

\newcommand{\squishlist}{
 \begin{list}{$\bullet$}
  { \setlength{\itemsep}{0pt}
     \setlength{\parsep}{3pt}
     \setlength{\topsep}{3pt}
     \setlength{\partopsep}{0pt}
     \setlength{\leftmargin}{1.5em}
     \setlength{\labelwidth}{1em}
     \setlength{\labelsep}{0.5em} } }

\newcommand{\squishlisttwo}{
 \begin{list}{$\bullet$}
  { \setlength{\itemsep}{0pt}
    \setlength{\parsep}{0pt}
    \setlength{	opsep}{0pt}
    \setlength{\partopsep}{0pt}
    \setlength{\leftmargin}{2em}
    \setlength{\labelwidth}{1.5em}
    \setlength{\labelsep}{0.5em} } }

\newcommand{\squishend}{
  \end{list}  }
\begin{document}
\setlength{\textfloatsep}{2pt}

%

\author{Ali Behrouz}
\authornote{These authors contributed equally.}
\affiliation{%
  \institution{University of British Columbia}
}
\email{alibez@cs.ubc.ca}

\author{Farnoosh Hashemi}  
\authornotemark[2]
\affiliation{%
  \institution{University of British Columbia}
}
\email{farsh@cs.ubc.ca}

\author{Laks V.S. Lakshmanan}
\affiliation{
  \institution{University of British Columbia}
}
\email{laks@cs.ubc.ca}




\title{FirmTruss Community Search in Multilayer Networks}

\begin{abstract}
In applications such as biological, social, and transportation networks, interactions between objects span multiple aspects. For accurately modeling such applications, multilayer networks have been proposed. Community search allows for personalized community discovery and has a wide range of applications in large real-world networks. While community search has been widely explored for single-layer graphs, the problem for  multilayer graphs has just recently attracted attention.  Existing community models in multilayer graphs have several limitations, including disconnectivity, free-rider effect, resolution limits, and inefficiency. To address these limitations, we study the problem of community search over large multilayer graphs.  We first introduce  \textit{FirmTruss}, a novel dense structure in multilayer networks, which extends the notion of truss to multilayer graphs. We show that FirmTrusses possess nice structural and computational properties and bring many advantages compared to the existing models. Building on this, we present a new community model based on FirmTruss, called \textit{FTCS}, and show that finding an FTCS community is NP-hard. We propose two efficient 2-approximation algorithms, and show that no polynomial-time algorithm can have a better approximation guarantee unless P = NP. We propose an index-based method to further improve the efficiency of the algorithms. We then consider attributed multilayer networks and propose a new community model based on network homophily. We show that community search in attributed multilayer graphs is NP-hard and present an effective and efficient approximation algorithm. Experimental studies on real-world graphs with ground-truth communities validate the quality of the solutions we obtain and the efficiency of the proposed~algorithms.
\end{abstract}

\maketitle
\pagestyle{\vldbpagestyle}
\begingroup\small\noindent\raggedright\textbf{PVLDB Reference Format:}\\
\vldbauthors. \vldbtitle. PVLDB, \vldbvolume(\vldbissue): \vldbpages, \vldbyear.\\
\href{https://doi.org/\vldbdoi}{}
\endgroup
\begingroup
\renewcommand\thefootnote{}\footnote{\noindent
This work is licensed under the Creative Commons BY-NC-ND 4.0 International License. Visit \url{https://creativecommons.org/licenses/by-nc-nd/4.0/} to view a copy of this license. For any use beyond those covered by this license, obtain permission by emailing \href{mailto:info@vldb.org}{info@vldb.org}. Copyright is held by the owner/author(s). Publication rights licensed to the VLDB Endowment. \\
\raggedright Proceedings of the VLDB Endowment, Vol. \vldbvolume, No. \vldbissue\ %
ISSN 2150-8097. \\
\href{https://doi.org/\vldbdoi}{doi:\vldbdoi} \\
}\addtocounter{footnote}{-1}\endgroup

\ifdefempty{\vldbavailabilityurl}{}{
\begingroup\small\noindent\raggedright\textbf{PVLDB Artifact Availability:}\\
The source code, data, and/or other artifacts have been made available at \url{\vldbavailabilityurl}.
\endgroup
}

\section{Introduction}
\label{sec:introduction}

Community detection is a fundamental problem in network science and has been traditionally addressed with the aim of determining an organization of a given network into subgraphs that express dense groups of nodes well connected to each other~\cite{Community_Detection_main}. Recently, a query-dependent community discovery problem, called community search (CS)~\cite{community2}, has attracted much attention due to its ability to discover personalized communities. It has several applications like social contagion modeling~\cite{social_contagion}, content recommendation~\cite{Personal_content}, and team formation~\cite{TeamFormation}. The CS problem seeks a cohesive subgraph containing the query nodes given a graph and a set~of~query~nodes.  

Significant research effort has been devoted to the study of CS over single-layer graphs, which have a single type of connection. However, in applications featuring complex networks such as social, biological, and transportation networks, the interactions between objects tend to span multiple aspects. \textit{Multilayer} (ML) \textit{networks}~\cite{main-ML}, where nodes can have interactions in multiple layers, have been proposed for accurately modeling such applications. Recently, ML networks have gained popularity in an array of applications in social and biological networks and in opinion dynamics~\cite{Ml-bio, ML-Covid, PP-ML, anomuly}, due to their more informative representation than single-layer graphs. 

\begin{example}
Figure \ref{fig:example}(a) is an ML network showing a group of researchers collaborating in various topics, \new{where each layer represents collaborations in an individual topic.} 
\end{example}
To find cohesive communities in single-layer graphs, many models have been proposed, e.g., $k$-core~\cite{k-core-community, community2}, $k$-truss~\cite{closest}, $k$-plex~\cite{k-plex}, and $k$-clique~\cite{clique-community}. Existing  methods for finding  cohesive structures in ML networks are inefficient. As a result, there is a lack of practical density-based community models in ML graphs. Indeed, there have been a number of studies on cohesive structures in ML networks~\cite{MLcore, CoreCube, Truss_cube, Coherehnt-core-ML}. However, they suffer from two main limitations. \textbf{(1)} The decomposition algorithms~\cite{MLcore, CoreCube, Truss_cube} based on these models have an \textit{exponential running time complexity in the number of layers}, making them prohibitive for CS. \textbf{(2)} These models have a hard constraint that nodes/edges need to satisfy in \textit{all} layers. It has been noted that ML networks may contain noisy/insignificant layers~\cite{FirmCore, MLcore}. These noisy/insignificant layers may be different for each node/edge. Therefore, this hard constraint could result in missing some dense structures~\cite{FirmCore}. Recently, FirmCore structure~\cite{FirmCore} in ML graphs has been proposed to address these limitations. However, a connected FirmCore can be disconnected by just removing  one edge, and it might have an arbitrarily large diameter. Both of these properties are undesirable for community models.


In addition to the above drawbacks of cohesive structures in ML networks, existing CS methods in ML graphs (e.g., \cite{ml-core-journal, ML-LCD, ML-random-walk}) suffer from some important limitations. \textbf{(1)} \textit{Free-rider effect} ~\cite{Free-rider}: some cohesive structure, irrelevant to the query vertices, could be included in the answer community. \textbf{(2)} \textit{Lack of connectivity}: a community, at a minimum, needs  to be a connected subgraph~\cite{ground_truth_community, closest}, but existing community models in ML graphs are not guaranteed to be connected. Natural attempts to enforce connectivity in these models lead to additional complications (see \S~\ref{sec:compare_community_models} for a detailed comparison with previous community models).  \textbf{(3)} \textit{Resolution Limit}~\cite{Resolution_limit}: in a large network, communities smaller than a certain size may not be detected. \textbf{(4)} \textit{Failure to scale}: to be applicable to large networks, a community model must admit scalable algorithms. To the best of our knowledge, all existing models suffer from these  limitations.


\begin{figure}[t]
    \centering
    \subfloat[\centering Multilayer network $G$]{{\includegraphics[width=0.5\linewidth]{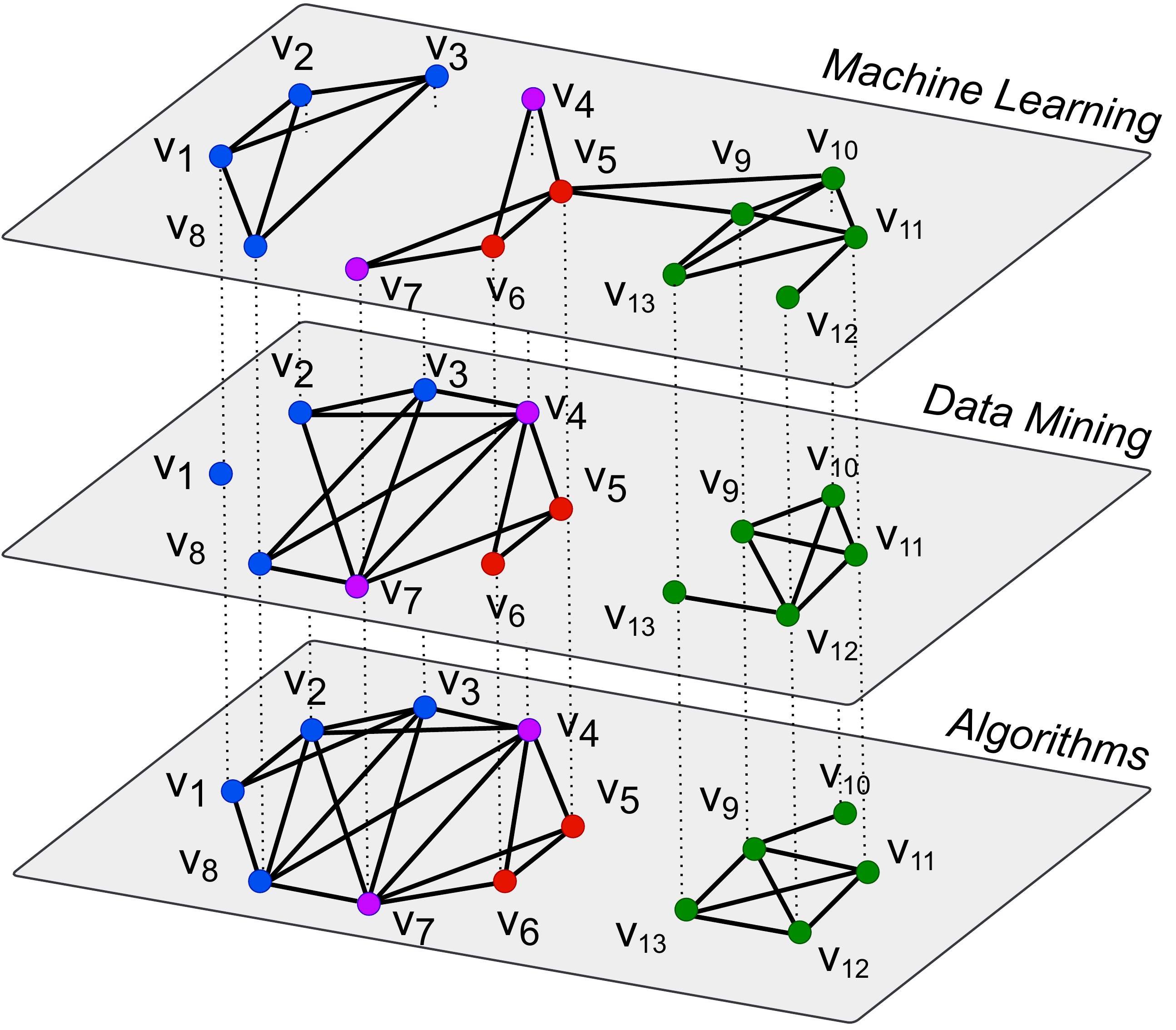} }}
    \subfloat[\centering Diameter (its~path~schema)~of~$G$]{{\includegraphics[width=0.5\linewidth]{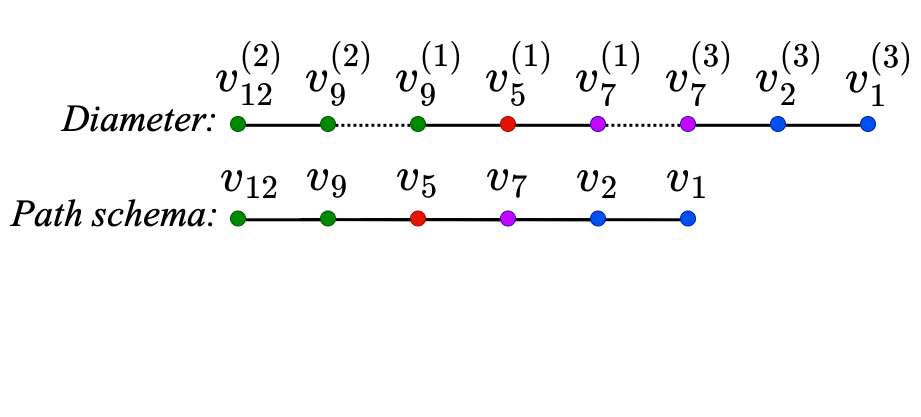} }}
    \caption{An example of a multilayer collaboration network.}
    \label{fig:example}
\end{figure}

To address the above limitations of existing studies, we study the problem of CS over  multilayer networks. First of all, we propose the notion of \textit{$(k, \lambda)$-FirmTruss}, based on the truss structure in simple graphs, as a subgraph (not necessarily induced) in which every two adjacent nodes in at least $\lambda$ individual layers are in at least $k - 2$ common triangles within the subgraph. We show that it inherits the  nice properties of trusses in simple graphs, viz., uniqueness, hierarchical structure, bounded diameter, edge-connectivity, and high density. Based on FirmTruss, we formally define our problem of \textit{FirmTruss Community Search} (FTCS). Specifically, given a set of query nodes, FTCS aims to find a connected subgraph which \textbf{(1)} contains  the query nodes; \textbf{(2)} is a FirmTruss; and \textbf{(3)} has the minimum diameter. We formally show that the diameter constraint in FTCS definition avoids the so-called "free-rider~effect".

In real-world networks, nodes are often associated with attributes. For example, they could represent a summary of a user's profile in social networks, or the molecular functions, or cellular components of a protein in protein-protein interaction networks. This rich information can help us find communities of superior quality. While there are several studies on single-layer attributed graphs, to the best of our knowledge, the problem of CS in multilayer attributed networks has not been studied. Unfortunately, even  existing CS methods in single-layer attributed graphs suffer from significant limitations. They \textit{require} users to input query attributes; however, users not familiar with the attribute distribution in the entire network, are limited in their ability to specify proper query attributes. Moreover, these studies only focus on one particular type of attribute (e.g., keyword), while most real-world graphs involve more complex  attributes. E.g., attributes of proteins can be \new{multidimensional vectors} \cite{PPI-vector-attribute}. The recently proposed VAC model~\cite{VAC} for single-layer graphs  does not require users to input query attributes, but is limited to metric similarity measures. To mitigate these limitations, we extend our FTCS model to attributed ML graphs, call it AFTCS, and present a novel community model leveraging the well-known  phenomenon of network homophily. This approach is based on maximizing the $p$-mean of similarities between users in a  community and does not require users to input query attributes. However, should a user wish to specify query attributes (say for exploration), AFTCS can easily support them. Moreover, it naturally handles a vector of attributes, handling complex features.

Since ML graphs provide more complex and richer information than single-layer graphs, they can benefit typical applications of single-layer CS~\cite{community_search_survey} (e.g., event organization, friend recommendation, advertisement, etc.), delivering better solutions. Below we  illustrate an exclusive application for multilayer CS.

\head{Brain Networks}
Detecting and monitoring functional systems in the human brain is an important and fundamental task in neuroscience \cite{functional_system_brain, functional_system_brain2}. A brain network (BN) is a graph in which nodes represent the brain regions and edges represent co-activation between regions. A BN generated from an  individual subject can be noisy and incomplete, however using BNs from many subjects helps us identify important structures more accurately~\cite{ML-random-walk, Brain_network_fmri}. A multilayer BN is a multilayer graph in which each layer represents the BN of a different person. A community search method in multilayer graphs can be used to \textbf{(1)} identify functional systems of each brain region; \textbf{(2)} identify common patterns between people's brains affected by diseases or under the influence of drugs.

We make the following contributions: \textbf{(1)} We introduce a novel dense subgraph model for ML graphs, \textit{FirmTruss}, and show that it retains the nice structural properties of Trusses 
($\S$~\ref{sec:Truss}). \textbf{(2)} We  formulate the problem of FirmTruss-based Community Search (FTCS) in ML graphs, and show the  FTCS problem is NP-hard and cannot be approximated in PTIME within a factor better than 2 of the optimal diameter, unless P = NP ($\S$~\ref{sec:undirected-firmtruss-based-community-search}). \textbf{(3)} We develop two efficient 2-approximation algorithms ($\S$~\ref{sec:Online_search}), and  propose an indexing method to further improve  efficiency ($\S$~\ref{sec:index_method}). \textbf{(4)} We extend FTCS to attributed networks and propose a novel homophily-based community model. 
We propose an exact algorithm for a special case of the problem and an approximation algorithm for the general case ($\S$~\ref{sec:AttributedFirmCommunities}). \textbf{(5)} Our extensive experiments on real-world ML graphs with ground-truth communities  show that our  algorithms can efficiently and effectively discover communities, significantly~outperforming~baselines~($\S$~\ref{sec:experiments}). For lack of space, some proofs are sketched. Complete details of all proofs and additional details can be found in Appendix.

\section{Related Work}
\label{sec:RelatedWork}
\head{Community Search}
Community search, which aims to find  query-dependent communities in a graph, was introduced by Sozio and Gionis \cite{k-core-community}. Since then, various community models have been proposed,  based on different dense subgraphs~\cite{community_search_survey}, including $k$-core~\cite{k-core-community, community2},  $k$-truss~\cite{closest, TrussEquivalence, truss_dynamic}, quasi-clique \cite{clique-community}, $k$-plex~\cite{k-plex}, and densest subgraph~\cite{densest_community}. Wu et al. \cite{Free-rider} identified an undesirable phenomenon, called free-rider effect, and propose query-biased density to reduce the free-rider effect for the returned community. More recently, CS has also been investigated for  directed~\cite{community-directed, D-core_community_search}, weighted~\cite{weighted-truss}, geo-social~\cite{geo-social-community, road_social}, temporal~\cite{temporal-community}, multi-valued~\cite{multi-valued}, and labeled~\cite{butterfly_core} graphs. Recently, learning-based CS is studied~\cite{GNN_CS, GNN_CS1, CS-MLGCN}, which needs a time-consuming step for training. 
These models are different from our work as they focus on a single type of interaction.

\head{Attributed Community Search}
Given a set of query nodes, attributed CS finds the query-dependent communities in which nodes share attributes~\cite{attribute-community, attribute-community2}. Most existing works on attributed single-layer graphs can be classified into two categories. The first category takes both nodes and attributes as query input~\cite{core_attribute, colocated_community}. The second category takes only attributes as  input, and returns the community related to the query attributes~\cite{query_keyword, contex_community_search}. All these  studies \textbf{(1)} require users to specify attributes as  input, and \textbf{(2)} consider only simple attributes (e.g.,  keywords), limiting their applications. Most recently, Liu et al.~\cite{VAC} introduced VAC in single-layer graphs, which does not require input query attributes. However, they are restricted to  metric similarity between users, which can 
limit applications. 
All these models are limited to single-layer graphs. 

\head{Community Search and Detection in ML Networks} \new{Several methods have been proposed for community detection in ML networks~\cite{ml-community-detection-survey, Community_Detection_ML1, Community_Detection_ML2}. However, they focus on detecting all communities, which is time-consuming and  independent of query nodes. 
Surprisingly, the problem of CS in ML networks is relatively less explored. Interdonato et al. \cite{ML-LCD} design a greedy search strategy by  maximizing the ratio of similarity between nodes inside and outside of the local community, over all layers. Galimberti et al. \cite{ml-core-journal} adopt a community search model based on the ML $\mathbf{k}$-core~\cite{azimi-etal}. Finally, Luo et al. \cite{ML-random-walk} design a random walk strategy to search local communities in multi-domain networks. 
}

\head{Dense Structures in ML Graphs}
Jethava et al. \cite{densest-common-subgraph} formulate the densest common subgraph problem. Azimi et al. \cite{azimi-etal} propose a new definition of core, \textbf{k}-core, over ML graphs. Galimberti et al. \cite{MLcore} propose algorithms to find all possible \textbf{k}-cores, and define the densest subgraph problem in ML graphs. Zhu~et~al.~\cite{Coherehnt-core-ML} introduce the problem of diversified coherent $d$-core search.~Liu~et~al.~\cite{CoreCube} propose the CoreCube problem for computing ML $d$-core decomposition on all subsets of layers.  Hashemi et al.~\cite{FirmCore} propose a new dense structure, FirmCore, and develop a FirmCore-based approximation algorithm for the problem of ML densest subgraph. Huang et al.~\cite{Truss_cube} define TrussCube in ML graphs, which aims to find a subgraph in which each edge has support $\geq k - 2$  in \textit{all selected layers}, which is different from the concept of FirmTruss. 


\section{Preliminaries}
\label{sec:preliminaries}
We let $G = (V, E, L)$ denote an ML graph, where $V$ is the set of nodes, $L$  the set of layers, and $E \subseteq V \times V \times L$  the set of intra-layer edges. We follow the common definition of ML networks~\cite{main-ML}, and consider inter-layer edges between two instances of identical vertices in different  layers. The set of neighbors of node $v \in V$ in layer $\ell \in L$ is denoted $N_\ell(v)$ and the degree of $v$ in layer $\ell$ is $\text{deg}_\ell (v) = |N_\ell(v)|$. For a set of nodes $H \subseteq V$, $G[H] = (H, E[H], L)$ denotes the subgraph of $G$ induced by $H$,  $G_\ell[H] = (H, E_\ell[H])$ denotes this subgraph in layer $\ell$, and $\text{deg}^H_{\ell}(v)$ denotes the degree of $v$ in this subgraph. Abusing notation, we write  $G_\ell[V]$ and $E_\ell[V]$ as $G_\ell$ and $E_\ell$, respectively. We use the  following notions in this paper. 



\head{Edge Schema} Connections (i.e., relationships) between  objects in ML networks can have multiple types;  by the \textit{edge schema} of a connection, we mean the connection ignoring its type.
\
\begin{dfn}[Edge Schema]
Given an ML network $G = (V, E, L)$ and an intra-layer edge $e = (v, u, \ell) \in E$, the edge schema of $e$ is the pair $\varphi = (v, u)$, which represents the relationship between two nodes, $v$ and $u$, ignoring its type. We denote by $\mathcal{E}$ the set of all edge schemas in $G$, $\mathcal{E} = \{ (v, u) | \exists \ell \in L: (v, u, \ell) \in E\}$.
\end{dfn}

Given an edge schema $\varphi = (v, u)$, we abuse the notation and use $\varphi_\ell$ to refer to the relationship between $v$ and $u$ in layer $\ell$, i.e., $\varphi_\ell = (v, u, \ell)$, whenever $(v,u,\ell)\in E$.

\head{Distance in ML Networks}
 For consistency, we use the common definition of ML distance~~\cite{ML-path} in the literature. However, our algorithms are valid for any definition of distance that is a metric.

\begin{dfn}[Path in Multilayer Networks]
Let $G = (V, E, L)$ be an ML graph and $v_\ell$ represent a node $v$ in layer $\ell \in L$. A path in $G$ is a sequence of nodes $\mathcal{P} : v^{1}_{\ell_1} \rightarrow v^{2}_{\ell_2} \rightarrow \dots \rightarrow v^{k}_{\ell_k}$ such that 
every consecutive pair of nodes is connected by an  inter-layer or intra-layer edge, 
\eat{
in the network, e.g.:
\begin{equation*}
    \mathcal{P} : v^{1}_{\ell_1} \rightarrow v^{2}_{\ell_2} \rightarrow \dots \rightarrow v^{k}_{\ell_k},
\end{equation*}
where 
for every two consecutive vertices in $\mathcal{P}$,
} 
i.e., $v^i = v^{i + 1}$ or $[\ell_i = \ell_{i + 1} \,\&\, (v^i, v^{i+1}, \ell_i)\in E]$. The path schema $\mathfrak{P}$ of $\mathcal{P}$ is obtained by removing inter-layer edges from path $\mathcal{P}$.
\end{dfn}

\new{Note that inter-layer edges between identical nodes are used as a penalty for changing edge types in a path}. We define the distance of two nodes $v$ and $u$,   $dist(v, u)$, as the length of the shortest path between them. The diameter of a subgraph $G[H]$, $diam(G[H])$, is the maximum distance between any pair of nodes in $G[H]$.\footnote{For convenience, we refer to both the longest shortest path distance as well as any path with that length as diameter.} 

\begin{example}
In Figure~\ref{fig:example}(a), the diameter of ML graph $G$ is 7, corresponding to the path (path schema) in Figure~\ref{fig:example}(b). 
\end{example}

\head{Density in ML Networks}
\eat{Jethava et al.~\cite{densest-common-subgraph} define the density function in multilayer networks as the minimum average degree over all layers. However, in this formulation, since we consider all layers, a noisy/insignificant layer may \hl{result in the density value}. Accordingly, it prevents us from distinguishing subgraphs that are dense in a large subset of layers. To address this issue,} 
In this study, we use a common definition of density in multilayer graphs proposed in ~\cite{ml-core-journal}.

\begin{dfn} [Density]\label{dfn:density} \cite{ml-core-journal} 
Given an ML graph $G=(V,E,L)$, a non-negative real number $\beta$, the density function is a real-valued function $\rho_\beta : 2^V \rightarrow \mathbb{R}^+$,  defined as:
\begin{equation*} \label{eq1}
\rho_\beta(S) = \max_{\hat{L} \subseteq L} \min_{\ell \in \hat{L}} \frac{|E_\ell[S]|}{|S|} |\hat{L}|^\beta.
\end{equation*}
\end{dfn}


\head{Free-Rider Effect} Prior work has identified  an undesirable phenomenon known as the "free-rider effect"~\cite{Free-rider}. Intuitively, if a community definition admits irrelevant subgraphs in the discovered community, we refer to the irrelevant subgraphs as free riders. Typically, a community definition is based on a goodness metric $f(S)$ for a subgraph $S$: subgraphs with the highest (lowest) $f(.)$ value are identified as communities. 

\begin{dfn}[Free-Rider Effect] 
Given an ML graph $G = (V, E, L)$, a non-empty  set of query vertices $Q$, let $H$ be a solution to a community definition that maximizes (resp. minimizes) goodness metric $f(.)$, and $H^*$ be a (global or local) optimum solution when our query set is empty. If $f(H^* \cup H) \geq f(H)$ (resp. $f(H^* \cup H) \leq f(H)$), we say that the  community definition suffers from free rider effect.
\end{dfn}


\eat{\head{\hl{Submodularity and Supermodularity}}
In the following, we review a few useful preliminaries on submodular and supermodular functions that will be useful for our algorithmic results.

\begin{dfn}[Supermodularity]
For a discrete set $\Omega$, a function $f: \Omega \rightarrow \mathbb{R}$ is supermodular if for any subsets $A, B \subseteq \Omega$ it satisfies:
\begin{equation}\label{eq:supermodular}
    f(A) + f(B) \leq f(A \cup B) + f(A \cap B).
\end{equation}
\end{dfn}

If the inequality $(\ref{eq:supermodular}$) is reversed, $f$ is referred to as a submodular function, and if we replace it with equality, the function is called modular.}


\head{Generalized Means} 
Given a finite set of positive real numbers $S = \{ a_1, a_2, \dots, a_n \}$, and a parameter $p \in \mathbb{R} \cup \{-\infty, +\infty\}$, the generalized mean (\textit{$p$-mean}) of $S$ is defined as
\begin{equation*}
    M_p(S) = \left( \frac{1}{|S|} \sum_{i = 1}^{|S|} (a_i)^p \right)^{1/p}.
\end{equation*}

For $p \in \{-\infty, 0, +\infty\}$, the mean can be defined by taking limits,~so that $M_{+\infty}(S) = \max a_i$, $M_0(S)$= ~$(\prod_{i = 1}^{|S|} a_i)^{1/|S|}$ , and $M_{-\infty}(S)$=$\min a_i$.

\section{FirmTruss Structure}
\label{sec:Truss}
In this section, we first recall the notion of $k$-truss in single-layer networks and then present FirmTruss structure in ML networks. 

\begin{dfn}[Support]
Given a single-layer graph $G = (V, E)$, the support of an edge $e = (u, v) \in E$, denoted  $sup(e, G)$, is defined as $| \{ \triangle_{u, v, w}: u, v, w \in V \} |$, where $\triangle_{u, v, w}$, called triangle of $u, v,$ and $w$, is a cycle of length three containing nodes $u, v,$ and $w$.
\end{dfn}

The $k$-truss of a single-layer graph $G$ is the maximal subgraph $H \subseteq G$, such that $\forall e \in H$, $sup(e, H) \geq (k - 2)$. Since each layer of an ML network can be counted as a single-layer network, one possible extension of truss structure is to consider different truss numbers for each layer, separately. However, this approach forces all edges to satisfy a constraint in all layers, including noisy/insignificant layers. This hard constraint would result in missing some dense structures~\cite{FirmCore}. Next, we suggest FirmTruss, a new family of cohesive structures based on the $k$-truss of  single-layer networks. 

\eat{It can be characterized by an $|L|$-dimensional integer vector $\mathbf{k} = [k_\ell]_{\ell \in L}$, \hl{whose} $(k_\ell - 2)$ denotes the minimum support of each edge in layer $\ell \in L$.

\begin{dfn}[Multilayer Truss]\label{dfn:MLTruss}
Given an ML network $G = (V, E, L)$, its edge schema set $\mathcal{E}$, and an $|L|$-dimensional integer vector
$\mathbf{k} = [k_\ell]_{\ell \in L}$, the multilayer $\mathbf{k}$-truss
of $G$ is a maximal subgraph $G[H] = (H \subseteq V, E[H], L)$ such that $\forall \: \ell \in L, \varphi \in \mathcal{E}[H]:$ $\varphi_\ell \in E$ and $\text{sup}(\varphi_\ell, G_\ell[H]) \geq (k_\ell - 2)$. 
\end{dfn}

There are four main drawbacks in this definition. First, a set of vertices may correspond to multiple trusses, so we do not have a unique trussness vector for each subgraph. The second challenge is the lack of nested property. Opposed to single-layer networks, ML trusses are not necessarily all nested into each other. Third, the problem of finding the set of all non-empty and distinct trusses of a multilayer graph cannot be solved in a polynomial time due to the exponential number of outputs. In fact, there are \bigo{\prod_{\ell \in L} k_\ell} trusses in the graph, where $k_\ell$ is the maximum \hl{number of supports} for an edge in layer $\ell$, which is exponential in the number of layers. Accordingly, this extension of truss decomposition is infeasible in large networks even with a small number of layers. \hl{Moreover}, as discussed in \cite{FirmCore}, in mining dense strucutres in multilayer graphs, forcing nodes/edges to satisfy a constraints in all layers, including noisy/insignificant layers would result in missing some dense structures. This approach forces all edge schemas to satisfy truss constraint in all layers, so it might cause missing some dense structures.}

\begin{dfn}[FirmTruss]\label{dfn:FirmTruss}
Given an ML graph $G = (V, E, L)$, its edge schema set $\mathcal{E}$, an integer threshold $1 \leq  \lambda \leq |L|$, and an integer $k \geq 2$, the $(k, \lambda)$-FirmTruss of $G$ ($(k, \lambda)$-FT for short) is a maximal subgraph $G[J^\lambda_{k}] = (J^\lambda_{k}, E[J^\lambda_{k}], L)$ such that for each edge schema $\varphi \in \mathcal{E}[J^\lambda_k]$ there are at least $\lambda$ layers $\{\ell_1, ..., \ell_\lambda\} \subseteq L$ such that $\varphi_{\ell_i} \in E_{\ell_i}[J^\lambda_k]$ and $\text{sup}(\varphi_{\ell_i}, G_{\ell_i}[J^\lambda_k]) \geq (k - 2)$.
\end{dfn}

\begin{example}
In Figure~\ref{fig:example}(a), let $k=4, \lambda=2$. The union of blue and purple nodes is a $(4, 2)$-FirmTruss, as every pair of  adjacent nodes in at least $2$ layers are in at least $2$ common triangles~within~the~subgraph.
\end{example}

For each edge schema $\varphi = (u,v) \in \mathcal{E}$, we consider an $|L|$-dimensional support vector, denoted  $\mathbf{S}_\varphi$, in which $i$-th element, $\mathbf{S}^i_\varphi$, denotes the support of the corresponding edge of $\varphi$ in $i$-th layer. We define the \textit{Top-$\lambda$ support} of $\varphi$ as the $\lambda$-th largest value in $\varphi$'s support vector. Next, we show that not only is the maximal $(k, \lambda)$-FirmTruss unique, it also has the nested property.


\begin{property}[Uniqueness]
The $(k, \lambda)$-FirmTruss of~$G$ is unique. 
\label{prop:Unique_FirmTruss}
\end{property}

\begin{property}[Hierarchical Structure]\label{prop:hierarchical}
Given a positive integer threshold $\lambda~\in~\mathbb{N}^+$, and an integer $k \geq 0$, the $(k + 1, \lambda)$-FT and $(k, \lambda + 1)$-FT of $G$ are subgraphs of its $(k, \lambda)$-FT.
\end{property}



\begin{property}[Minimum Degree]\label{prop:truss-degree}
Let $G = (V, E, L)$ be an ML graph, and $H = G[J^\lambda_k]$ be its $(k, \lambda)$-FT. Then $\forall$  node $u \in J^\lambda_k$, there are at least $\lambda$ layers $\{\ell_1, ..., \ell_\lambda\} \subseteq L$ such that $\text{deg}^{H}_{\ell_i}(u) \geq k - 1$, $1\leq i\leq \lambda$.
\end{property}

In ML networks, the degree of a node $v$ is an $|L|$-dimensional vector whose $i$-th element is the degree of node $v$ in $i$-th layer. Let Top-$\lambda$ degree of $v$ be the $\lambda$-th largest value in the degree vector of $v$. By Property~\ref{prop:truss-degree},   each node in a $(k, \lambda)$-FirmTruss has a Top-$\lambda$ degree of at least $k - 1$. That means, each $(k, \lambda)$-FirmTruss is a $(k - 1, \lambda)$-FirmCore~\cite{FirmCore}. Like trusses, a FirmTruss  may be disconnected, and we refer to its connected components as \textit{connected FirmTrusses}. 

Trusses are known to be dense, cohesive, and stable structures. These important characteristics of  trusses make them popular for modeling  communities~\cite{closest}. Next,  we discuss the density, closeness, and edge connectivity of  FirmTrusses. Detailed proofs of the results and tightness examples can be found in~Appendix~\ref{app:proof_theorems}.

\begin{theorem}[Density Lower Bound]
\label{theorem:FirmTruss_density}
Given an ML graph $G = (V, E, L)$, the density of a $(k, \lambda)$-FirmTruss, $G[J^\lambda_{k}] \subseteq G$, satisfies:
\begin{equation*}
    \rho_\beta(J^\lambda_k) \geq \frac{(k - 1)}{2 |L|} \underset{\xi \in \mathbb{Z}, 0\leq \xi < \lambda}{\max} (\lambda - \xi) (\xi + 1)^{\beta}.
\end{equation*}
\end{theorem}


\begin{theorem}[Diameter Upper Bound]
\label{theorem:FirmTruss_diam}
Given an ML graph $G = (V, E, L)$, the diameter of a connected $(k, \lambda)$-FirmTruss, $G[J^\lambda_{k}] \subseteq G$, is no more than $T \times\floor{\frac{2 |J^\lambda_k| - 2}{k}}$, where $T = 1 + \frac{1}{\floor{\frac{|L|}{|L| - \lambda}}}$.
\end{theorem}
\begin{proofS}
 We show that if $\mathcal{P}$ is the diameter of the $(k, \lambda)$-FT, and $\frac{t}{t+1} |L| > \lambda \geq \frac{t-1}{t} |L|$, then its path schema, $\mathfrak{P}$, has a length at least $\frac{t}{t+1} \times |\mathcal{P}|$. Then we consider every $t$ consecutive edges in the diameter as a block and construct a path, with the same path schema as $\mathcal{P}$ such that edges in each block are in the same layer. Next, we use edge schema supports to bound its length~in~each~block.
\end{proofS}

\begin{example}
In Figure~\ref{fig:example}(a), the union of blue and purple nodes is a connected $(4, 2)$-FirmTruss with diameter 2. Theorem~\ref{theorem:FirmTruss_diam} provides the upper bound of $\floor{\frac{4}{3} \times \floor{\frac{2\times 6 - 2}{4}}}~=~\floor{\frac{8}{3}} = 2$ on its diameter. 
\end{example}

\begin{theorem}[Edge Connectivity]
\label{theorem:FirmTruss_connectivity}
For an ML graph $G = (V, E, L)$, any connected $(k, \lambda)$-FirmTruss $G[J^\lambda_{k}] \subseteq G$ remains connected whenever fewer than $\lambda \times (k - 1)$ intra-layer edges are removed.
\end{theorem}
\eat{
Note that all provided upper and lower bounds in Theorems~\ref{theorem:FirmTruss_density}, \ref{theorem:FirmTruss_diam}, \ref{theorem:FirmTruss_connectivity} are tight (in Appendix~\ref{app:tightness_examples}). }


\eat{
\subsection{Algorithm}\label{sec:algorithm_FirmTruss}
In this section, we use the nested property of FirmTrusses to design an efficient algorithm for the FirmTruss decomposition problem. Given an ML network $G = (V, E, L)$, where $L = \{\ell_1, \dots, \ell_{|L|} \}$, for each edge schema $\varphi \in \mathcal{E}$, we consider an $|L|$-dimensional vector, support vector denoted by $\mathbf{S}_\varphi$, in which $i$-th element, $\mathbf{S}^i_\varphi$, represents the number of supports of $\varphi_{\ell_i}$ in layer $\ell_i$. If $sup(\varphi_{\ell_i}, G_{\ell_i}) = 0$, then $\mathbf{S}^i_{\varphi} = 0$.

The advantage of FirmTruss is that we only need to look at the $\lambda$-th largest value in the support vector, which we call Top-$\lambda$ support of $\varphi$, denoted by Top$-\lambda(\mathbf{S}_\varphi)$.

\begin{fact}
For each $\varphi \in \mathcal{E}$, Top$-\lambda(\mathbf{S}_\varphi) \geq FTI_\lambda(\varphi)$.
\end{fact}

For a given edge schema $\varphi$, if Top$-\lambda(\mathbf{S}(\varphi)) = k - 2$, then it cannot be a part of $(k', \lambda)$-FirmTruss, for $k' > k + 2$. Therefore, given $\lambda$, we can consider Top$-\lambda(\mathbf{S}(\varphi))$ as an upper bound on the FirmTruss index of $\varphi$. In the FirmTruss decomposition, we recursively pick an edge schema $\varphi$ with the lowest Top$-\lambda(\mathbf{S}(\varphi))$, assign its FirmTruss index as its Top$-\lambda(\mathbf{S}(\varphi))$, and then remove it from~the~graph.

Algorithm~\ref{alg:FirmTruss_decomposition} processes the edge schemas in increasing order of Top$-\lambda(\mathbf{S})$, by using a vector $B$ of lists such that each element $k\geq 2$ contains all edge schemas with Top$-\lambda(\mathbf{S})$ equal to $k - 2$. Based on this technique, we keep edge schemas sorted throughout the algorithm and can update each element in $\mathcal{O}(1)$ time. First, we start with computing the supports for each edge in each layer, which can be done by the state-of-the-art algorithms to count the supports in single-layer graphs~\cite{truss-algorithm}. After initializing the vector $\mathbf{S}$ and $B$, Algorithm \ref{alg:FirmTruss_decomposition} starts processing of $B$'s elements in increasing order, and if an edge schema $\varphi$ is processed at iteration $k$, its $FTI_\lambda$ is assigned to $k$ and removed all its corresponding edges in all layers from the graph. In order to remove an edge schema from the graph, we need to update the supports of the adjacent edges in each layer, which leads to changing the Top$-\lambda(\mathbf{S})$ of them, and changing their bucket accordingly. To efficiently update the support of adjacent edges, we use the following fact.

\begin{fact}\label{fact:FirmTruss_algorithm2}
If two edge schemas $\varphi$ and $\varphi'$ are adjacent in layer $\ell$, removing an edge schema $\varphi$ cannot not affect the Top$-\lambda(\mathbf{S})$ of $\varphi'$, unless Top$-\lambda(\mathbf{S}(\varphi')) = \mathbf{S}^\ell(\varphi')$.
\end{fact}

Therefore, in lines 13-16 of the algorithm we store all potential edge schemas such that their Top$-\lambda(\mathbf{S})$ can be changed, and in lines 17-20 we update their Top$-\lambda(\mathbf{S})$ and their buckets.

Algorithm~\ref{alg:FirmTruss_decomposition} finds the FirmTruss indices ($FTI_\lambda$) for a given $\lambda$, but for Firmtruss decomposition we need to find all $FTI_\lambda$s for all possible values of $\lambda$. A possible way is to re-run Algorithm \ref{alg:FirmTruss_decomposition} for each value of $\lambda$, but it is inefficient since we can avoid re-computing the supports and Top$-\lambda(\mathbf{S})$ in lines 1-4. Alternatively, we first sort the vector $\mathbf{S}$ of each edge schema in decreasing order and when threshold is $\lambda$, we assign $\lambda$-th element of sorted vector to $I[\varphi]$.

\begin{algorithm}[t] 
    \small
    \SetKwInOut{Input}{Input}
    \SetKwInOut{Output}{Output}
    \Input{An ML graph $G = (V, E, L)$}
    \Output{Skyline FirmTruss index of each edge schema}
    
    Compute $\mathbf{S}^\ell_\varphi = sup(\varphi_\ell, G_\ell)$ for each edge schema $\varphi \in \mathcal{E}$ in each layer $\ell \in L$;
    
    \ForAll{$\lambda = 1, 2, \dots, |L|$}{
        reinitialize supports, $\mathbf{S}^\ell_\varphi$;\\
        \ForAll{$\varphi \in \mathcal{E}$}{
            $I[\varphi] \leftarrow \text{Top-$\lambda$}(\mathbf{S}_\varphi) + 2$;\\
            $B[I[\varphi]] \leftarrow B[I[\varphi]] \cup \{\varphi\}$;
        }
        \ForAll{$k = 2, 3, \dots, |V|$}{
            \While{$B[k] \neq \emptyset$}{
                Pick and remove $\varphi = (v, u)$ from $B[k]$;\\
                $\sft{\varphi} \leftarrow \sft{\varphi}  \cup  (k, \lambda)$, $N \leftarrow \emptyset$;\\
                \ForAll{$(v, w, \ell) \in E$ and $I[(v, w)] > k$ and $\varphi_\ell \in E$}{
                    \If{$(u, w, \ell) \in E$ and $I[(u, w)] > k$}{
                        \If{$\mathbf{S}^\ell_{(v, w)} + 2 = I[(v, w)]$}{
                            $N \leftarrow N \cup \{(v, w)\}$;\\
                        }
                        \If{$\mathbf{S}^\ell_{(u, w)} + 2 = I[(u, w)]$}{
                            $N \leftarrow N \cup \{(u, w)\}$;\\
                        }
                        $\mathbf{S}^\ell_{(v, w)} \leftarrow \mathbf{S}^\ell_{(v, w)} - 1$;
                        $\mathbf{S}^\ell_{(u, w)} \leftarrow \mathbf{S}^\ell_{(u, w)} - 1$;\\
                    }
                }
                \ForAll{$\varphi'=(w, t) \in N$}{
                    Remove $\varphi'$ from $B[I[\varphi']]$;\\
                    Update $I[\varphi']$;\\
                    $B[I[\varphi']] \leftarrow B[I[\varphi']] \cup \{\varphi'\}$;
                }
                Remove all instance of $\varphi$ from $G$ in all layers;
            }
        }
    }
    Remove all dominated indices in $\sft{\varphi}$ for each $\varphi \in \mathcal{E}$;
    \caption{FirmTruss Decomposition}
    \label{alg:FirmTruss_decomposition}
\end{algorithm}

\head{Computational complexity} The time complexity of Algorithm~\ref{alg:FirmTruss_decomposition} is $\mathcal{O}(\sum_{\ell \in L} |E_\ell|^{1.5} + |E||L| + |E|\lambda \log |L|)$, so based on the efficient computing of Top$-\lambda$ degree for all values of $\lambda$, the time complexity of FirmCore decomposition is $\mathcal{O}(\sum_{\ell \in L} |E_\ell|^{1.5} + |E||L|^2)$. 
Using a heap, we can find the $\lambda$-th largest element of vector $\mathbf{S}$ in $\mathcal{O}(|L|+\lambda \log|L|)$ time so lines 2-4 take $\mathcal{O}(|E|( |L| + \lambda\log |L|))$ time. We remove each edge schema exactly one time from the buckets, taking time $\mathcal{O}(|E|)$. For $\ell \in L$, let $nb^\ell_{\geq}(u) = \{ v | (v, u, \ell) \in E, deg_\ell(v) \geq deg_\ell(u) \}$, lines 9-16 take $\mathcal{O}(\sum_{(v, u) \in \mathcal{E}} \sum_{\ell \in L} deg_\ell(u) |nb^\ell_{\geq}(u) |)$. For each layer $\ell$, we know that $|nb^\ell_{\geq}(u)| \leq 2\sqrt{|E_\ell|}$~\cite{truss-algorithm}, so we can conclude that lines 9-16 take $\mathcal{O}(\sum_{\ell \in L} |E_\ell|^{1.5})$. Finally, to update the buckets for a given edge schema we need $\mathcal{O}(|L|)$,~so~line~14~takes~$\mathcal{O}(|E| |L|)$~time.
}

\section{FirmTruss-Based Community Search}
\label{sec:undirected-firmtruss-based-community-search}
\subsection{\new{Problem Definition}}
In this section, we propose a community model based on FirmTruss  in ML networks. Generally, a community in a network is identified as a set of nodes that are densely connected. Thus, we use the notion of FirmTruss for modeling a densely connected community in ML graphs, which inherits several desirable structural properties, such as high density (Theorem~\ref{theorem:FirmTruss_density}), bounded diameter (Theorem~\ref{theorem:FirmTruss_diam}), edge connectivity (Theorem~\ref{theorem:FirmTruss_connectivity}), and hierarchical structure (Property~\ref{prop:hierarchical}).

\begin{problem}[FirmTruss Community Search] Given an ML network $G = (V, E, L)$, two integers $k \geq 2$ and $\lambda \geq 1$, and a set of query vertices $Q \subseteq V$, the FirmTruss community search (FTCS) is to find a connected subgraph $G[H] \subseteq G$ satisfying:
\begin{enumerate}[noitemsep,topsep=0pt,parsep=0pt,partopsep=0pt]
    \item $Q \subseteq H$,
    \item $G[H]$ is a connected $(k, \lambda)$-FirmTruss,
    \item diameter of $G[H]$ is the minimum among all subgraphs satisfying conditions (1) and (2).
\end{enumerate}
\end{problem}

Here, Condition (1) requires that the community contains the query vertex set $Q$, Condition (2) makes sure that the community is densely connected through a sufficient number of layers, and Condition (3) requires that each vertex in the community  be as~close to other vertices as possible,~which~excludes~irrelevant vertices from the community. Together, all three conditions ensure that the returned community is a cohesive subgraph with good~quality.


\begin{example}
In the graph shown in Figure~\ref{fig:example}, let $v_1$ be the query node, $k = 4$, and $\lambda = 2$. The union of purple and blue nodes is a $(4, 2)$-FirmTruss, with diameter 2. The FTCS community removes purple nodes to reduce the diameter. Let $v_6$ be the query node, $k = 4$, and $\lambda = 1$, the entire graph is a $(4, 1)$-FirmTruss, with diameter 7. The FTCS community removes blue and green nodes to reduce the diameter.
\end{example}


\noindent
\textbf{Why FirmTruss Structure?}
Triangles are  fundamental building blocks of networks, which show a strong and stable relationship among nodes~\cite{triangle_building_block}. In ML graphs, every two nodes can have different types of relations, and a connection can be counted as strong and stable if it is a part of a triangle in each type of interaction. However, forcing all edges to be a part of a triangle in every interaction type is too strong a constraint. Indeed,  TrussCube~\cite{Truss_cube}, which is a subgraph in which each edge has support $k - 2$  in all selected layers, is based on this strong constraint. In Figure~\ref{fig:example}, the green nodes are densely connected. However, while this subgraph is a $(4, 1)$-FirmTruss, due to the hard constraint of TrussCube, green nodes are a $2$-TrussCube, meaning that this model misses it. That is, even if the green subgraph were to be far less dense and have no triangles in it, it would still be regarded as $2$-TrussCube. Furthermore, in some large networks, there is no non-trivial TrussCube when the number of selected layers is more than $3$~~\cite{Truss_cube}. In addition to these limitations, the exponential-time complexity of its algorithms makes it impractical for large ML graphs. By contrast, FirmTrusses have a polynomial-time algorithm, with guaranteed high density, bounded diameter, and edge connectivity. \new{While FirmCore~\cite{FirmCore} also has a polynomial-time algorithm, a connected FirmCore can be disconnected by just removing one edge, and it might have an arbitrarily large diameter, which are~both~undesirable~for~communities.}

\subsection{Problem Analysis} \label{sec:Problem_analysis}
\new{Next we analyze the hardness of the FTCS problem and show  not only that it is NP-hard, but it cannot be approximated within a factor better than $2$. Thereto, we define the decision version of the FTCS, $d$-FTCS, to test whether $G$ contains a connected FirmTruss community with diameter $\leq d$, that contains $Q$. Given $\alpha \geq 1$ and the optimal solution to FTCS, $G[H^*]$, an algorithm achieves an $\alpha$-approximation to FTCS if it outputs a connected $(k, \lambda)$-FirmTruss, $H$, such that $Q \subseteq H$ and $diam(G[H]) \leq \alpha \times diam(G[H^*])$.}

\eat{Following results show that not only in general FTCS problem is NP-hard but also for any arbitrary value of $\lambda$ and $|L|$ it remains NP-hard.} \eat{\textcolor{red}{Following results show that not only in general FTCS problem is NP-hard but also for any arbitrary value of $\lambda$ and $|L|$ it is NP-hard.}}

\eat{\begin{theorem}[$d$-FTCS Hardness] \label{theorem:FTCS-hardness}
The $d$-FTCS problem is NP-hard.
\end{theorem}
\begin{proofS}
\new{Let $d = 1$, we reduce the well-known NP-hard problem of Maximum Clique (decision version) to $d$-FTCS problem. Given an instance of Maximum Clique problem $G$, we construct an ML graph $G'$ such that $G'_\ell = G$ for every layer $\ell$. Next, to reduce $d$-FTCS problem to the Maximum Clique problem, we show that each $(k, \lambda)$-FirmTruss contains a $k$-clique in an arbitrary layer $\ell$.}
\end{proofS}}
\eat{Given the NP-hardness of the FTCS problem, we next analyze its approximability.}

\begin{theorem} [FTCS Hardness and Non-Approximability]\label{theorem:FTCS-apx}
\new{Not only the $d$-FTCS problem is NP-hard, but also for any $\epsilon > 0$, the FTCS-problem cannot be approximated in polynomial-time within a factor $(2 - \epsilon)$ of the optimal solution, unless $P = NP$.}
\end{theorem}

In \S~\ref{sec:Online_search}, we provide a 2-approximation algorithm for FTCS, thus essentially matching this lower bound.  

\head{Avoiding Free-rider Effect}  We can show: 

\begin{theorem}[FTCS Free-Rider Effect]\label{theorem:FTCS-free-rider}
For any multilayer network $G = (V, E, L)$ and query vertices $Q \subseteq V$, there is a solution $G[H]$ to the FTCS problem such that for all query-independent optimal solutions $G[H^*]$, either $H^* = H$, or $G[H \cup H^*]$ is disconnected, or $G[H \cup H^*]$ has a strictly larger diameter than $G[H]$.
\end{theorem}

\subsection{Comparison of CS Models in ML Networks}\label{sec:compare_community_models}
We compare FirmTruss with existing CS models for ML networks. \eat{Table~\ref{tab:comparision_Table} shows the summary of theoretical comparison of existing models in terms of four main metrics, i.e., cohesiveness, connectivity, edge redundancy, scalability. Stars rank each model in terms of a metric.}

\head{Cohesiveness}
In the literature, communities are  defined as cohesive, densely connected subgraphs. Hence, cohesiveness, i.e., high density, is an  important metric to measure the quality of communities. It is shown that FirmCore can find subgraphs with higher density than the ML $\mathbf{k}$-core~\cite{FirmCore}. Since each $(k, \lambda)$-FirmTruss is a $(k - 1, \lambda)$-FirmCore (Property~\ref{prop:truss-degree}), FirmTruss  is more cohesive than ML $\mathbf{k}$-core. ML-LCD model~\cite{ML-LCD} maximizes the similarity of nodes within the subgrpah. RWM~\cite{ML-random-walk} is a random walk-based approach and minimizes the conductance. Both of these models do not control the 
density of the subgraph. Thus, one node may have degree 1 within the subgraph, allowing non-cohesive structures. 

\head{Connectivity}
A minimal requirement for a community is to be a connected subgraph. Surprisingly, ML  $\mathbf{k}$-core, ML-LCD, and RWM (with multiple query nodes) community search models do not guarantee connectivity! Natural attempts to enforce connectivity in these community models lead to additional complications and might change the hardness of the problem. Even after enforcing connectivity, these models can be disconnected by just removing one intra-layer edge, which is undesirable for community models~\cite{ECC}. Our FirmTruss community model forces the subgraph to be connected, and guarantees that after removing up to  $\lambda\times (k-1)$ intra-layer edges, the $(k, \lambda)$-FirmTruss is still connected (Theorem~\ref{theorem:FirmTruss_connectivity}).

\head{Edge Redundancy}
In ML networks, the rich information about~node connections leads to repetitions, meaning  edges between the same pair of nodes repeatedly appear in multiple layers. Nodes with repeated connections are more likely to belong to the same~community~\cite{Redundancy-ML-Community}. Also, without such  redundancy of connections, the  tight connection  between objects in ML networks may not be represented effectively and accurately. While none of the models ML $\mathbf{k}$-core, ML-LCD, and RWM guarantees edge redundancy, in a  $(k, \lambda)$-FirmTruss, each edge is required to appear in at least $\lambda$~layers.

\eat{
\head{Scalability}
 To be applicable to large networks, a community model must admit scalable algorithms. In the ML $\mathbf{k}$-core model, the community search algorithm in~\cite{ml-core-journal} checks each of the exponentially many candidate cores containing the query nodes to find the core maximizing the objective function. \eat{It is unclear if there is a  polynomial-time algorithm to find the community based on this model.} The time complexity of community search using ML-LCD in the worst case is $\mathcal{O}(|L||V|^3 \log|V|)$. The time complexity of RWM method is $\mathcal{O}(T|L|(|V| + |E|))$, where $T$ is the number of iterations. Finally, the time complexity of the FirmTruss-based community search model is $\mathcal{O}(|E|^{1.5})$ in the worst case, which is more efficient than other community search models, except RWM. While theoretically RWM algorithm seems more efficient than FirmTruss, due to the large value of needed iterations, $T$, we experimentally show that FirmTruss is faster.
}
\head{Hierarchical Structure}
The hierarchical structure is a desirable property for community search models as it represents a community at different levels of granularity, and can also avoid the Resolution Limit problem as is discussed in \cite{Resolution_limit}. While FirmTruss has a hierarchical structure, none of the existing models has~this~property.

\section{FTC Online Search}\label{sec:Online_search}
Given the hardness of the FTCS problem, we propose two online\eat{\footnote{Following \cite{butterfly_core}, we refer to an algorithm involving no offline processing such as index building as an online algorithm.}} 2-approximation algorithms in top-down and bottom-up manner.

\subsection{Global Search}
We start by defining query distance in multilayer networks.

\begin{dfn}[Query Distance]
Given a multilayer network $G = (V, E, L)$, a subgraph $G[H] \subseteq G$, a set of query vertices $Q \subseteq H$, and a vertex set $S \subseteq H$, the query distance of $S$ in $G[H]$, $dist_{G[H]}(S, Q)$, is defined as the maximum length of the shortest path from $u \in S$ to a query vertex $q \in Q$, i.e., $dist_{G[H]} (S, Q) = \max_{u \in S, q \in Q} dist(u, q)$.
\end{dfn}

For a graph $G$, we use $dist_{G}(u, Q)$ to denote the query distance for a vertex $u \in V$. Previous works (e.g., see \cite{closest, butterfly_core}) use a simple greedy algorithm which iteratively removes  the nodes with maximum distance to query nodes, in order to minimize the query distance. This approach can be inefficient, as it reduces the query distance  by just 1 in each iteration, in the worst case. We instead employ a binary search  on~the~query~distance~of~a~subgraph. \eat{Given integers $k$ and $\lambda$, it first finds the maximal connected $(k, \lambda)$-FirmTruss containing $Q \subseteq V$, denoted as $G_0$. Since the diameter of $G_0$ may be large, by conducting binary search on the value of $d$, it iteratively removes the vertices with distance more than $d$ to $Q$ and maintains the remaining graph as a $(k, \lambda)$-FirmTruss.} 

Algorithm~\ref{alg:binary_search} gives the details of the FTCS Global algorithm. It first finds a maximal connected $(k, \lambda)$-FirmTruss $G_0$ containing $Q$. We keep our best found subgraph in $\mathcal{G}$, through the algorithm. Then in each iteration, we make a copy of $\mathcal{G}$, $G'$, and for each vertex $u \in V[G']$, we compute the query distance of $u$. Then, we conduct a binary search on the value of $d_{avg}$ and delete vertices with query distance $\geq d_{avg}$ and all their incident edges, in all layers. From the resulting graph we remove edges/vertices to maintain $G'$ as a $(k, \lambda)$-FirmTruss (lines 6 and 7). We maintain the $(k, \lambda)$-FirmTruss by deleting the edge schemas whose Top-$\lambda$ support is  $< k - 2$. Finally, the algorithm returns a subgraph $\mathcal{G}$, with the smallest query distance.

\begin{algorithm}[t] 
    \small
    \SetKwInOut{Input}{Input}
    \SetKwInOut{Output}{Output}
    \Input{An ML graph $G = (V, E, L)$, a set of query vertices $Q \subseteq V$, and two integers $k\geq 2$ and $\lambda \geq 1$}
    \Output{A connected $(k, \lambda)$-FT containing $Q$ with a small diameter}
    
    $G_0 \leftarrow$ Find a maximal connected $(k, \lambda)$-FirmTruss containing $Q$; //~See Algorithm~\ref{alg:FirmTruss_G0} (or Algorithm~\ref{alg:maximal_firmtruss_index}) \\
    $i \leftarrow 0$; $d_{\min} \leftarrow 1$; $d_{\max} \leftarrow \text{dist}_{G_0}(G_0, Q)$; $\mathcal{G} \leftarrow G_0$;\\
    
    \While{$d_{\min} < d_{\max}$}{
        $d_{avg} \leftarrow \floor{\frac{d_{\min} + d_{\max}}{2}}$; $G'  \leftarrow \mathcal{G}$\\
        $S \leftarrow $ set of vertices with $d_{avg} \leq \text{dist}_{G'}(u, Q)$;\\
        Delete nodes in $S$ and their incident edges from $G'$ in all layers;\\
        Maintain $G'$ as $(k, \lambda)$-FirmTruss by removing vertices/edges; \\
        \If{$Q \not \subseteq G'$ \textbf{or} $G'$ \text{is disconnected} \textbf{or} $d_{\max} < \text{dist}_{G'}(G', Q)$}{
            $d_{\min} \leftarrow 1 + d_{avg}$;
        }
        \Else{
        $d_{\max} \leftarrow \text{dist}_{G'}(G', Q)$;\\
        Let the remaining graph $G'$ as $\mathcal{G}$;
        }
    }
    \Return $\mathcal{G}$;
    \caption{FTCS Global Search}
    \label{alg:binary_search}
\end{algorithm}

The procedure for finding the maximal FirmTruss containing $Q$ is given in Algorithm~\ref{alg:FirmTruss_G0}. Notice,  a $(k, \lambda)$-FirmTruss (see Def.~\ref{dfn:FirmTruss}) is a maximal subgraph $G[J^\lambda_k]$ in which each edge schema $\varphi \in \mathcal{E}[J^\lambda_k]$ has Top-$\lambda$ support $\geq k - 2$. The algorithm first uses Property~\ref{prop:truss-degree}, and  removes all vertices with Top-$\lambda$ degree $< k-1$. It then iteratively deletes all instances of disqualified edge schemas in all layers from the original graph $G$, and then updates the Top-$\lambda$ support of their adjacent edges. To do this efficiently, 
we use the following fact: 

\begin{fact}\label{fact:FirmTruss_algorithm2}
If two edge schemas $\varphi$ and $\varphi'$ are adjacent in layer $\ell$, removing  edge schema $\varphi$ cannot affect Top$-\lambda(\mathbf{S}_{\varphi'})$, unless Top$-\lambda(\mathbf{S}_{\varphi'}) = \mathbf{S}^\ell_{\varphi'}$.   
\end{fact}

Thus, in lines 12-20, we update the Top-$\lambda$ support of those edge schemas whose Top-$\lambda$ support may be affected by removing $\varphi$. Finally, we use BFS traversal from a query node $q \in Q$ to find the connected component including query vertices. We omit the details of FirmTruss maintenance since it can use  operations similar to those in  lines 8-21 of Algorithm~\ref{alg:FirmTruss_G0}.

\begin{example}
\new{In Figure~\ref{fig:example}, let $k = 4$, $\lambda = 2$, and $Q = \{v_2\}$. Algorihtm~\ref{alg:FirmTruss_G0} first calculates the  support of each edge schema. Next, it removes the edge schema $\varphi = (v_{12}, v_{13})$  in all layers, as its Top-2 support is 0. Next, it updates the support of edge schema adjacent to $\varphi$, and iteratively removes all edges between green, red, and purple nodes since their edge schema has Top-2 support less than $2$. Finally, the remaining graph, the union of blue of purple nodes, is returned by the algorithm.}
\end{example}

\begin{example}
In Figure~\ref{fig:example}, let $k = 4$, $\lambda = 1$, and $Q = \{v_{1}\}$. Algorithm~\ref{alg:binary_search} starts from the entire graph as $G_0$. Since the query distance is 7, it sets $d_{avg} = \frac{7+1}{2} = 4$, removes all nodes with query distance $\geq$ 4, and maintains the remaining graph as $(4, 1)$-FirmTruss. The remaining graph includes blue, purple, and red nodes. Next, it sets $d_{avg} = \floor{\frac{3 + 1}{2}} = 2$, removes all vertices with query distance $\geq$ 2, and maintains the remaining graph as a $(4, 1)$-FirmTruss, which includes blue nodes. Algorithm~\ref{alg:binary_search} terminates and returns this subgraph~as~the~solution.
\end{example}

Next, we analyze the  approximation quality and complexity of the FTCS Global algorithm.

\begin{algorithm}[t] 
    \small
    \SetKwInOut{Input}{Input}
    \SetKwInOut{Output}{Output}
    \Input{An ML graph $G = (V, E, L)$, a set of query nodes $Q \subseteq V$, and integers $k\geq 2$ and $\lambda \geq 1$}
    \Output{A maximal connected $(k, \lambda)$-FirmTruss containing $Q$}
    
    $G' \leftarrow$ Remove all vertices with Top-$\lambda$ degree less than $k-1$; \\
    
    Compute $\mathbf{S}^\ell_\varphi = sup(\varphi_\ell, G_\ell')$ for each edge schema $\varphi \in \mathcal{E}$ and $\ell \in L$;\\
    
    $N, B \leftarrow \emptyset$;\\
    \ForAll{$\varphi \in \mathcal{E}[G']$}{
        $I[\varphi] \leftarrow \text{Top-$\lambda$}(\mathbf{S}_{\varphi}) + 2$;\\
        \If{$I[\varphi] < k$}{
            $N \leftarrow N \cup \{\varphi\}$;
        }
    }
    \While{$N \neq \emptyset$}{
        Pick and remove $\varphi = (v, u)$ from $N$;\\
        \ForAll{$(v, w, \ell) \in E[G']$ and $I[(v, w)] \geq k$ and $\varphi_\ell \in E$}{
            \If{$(u, w, \ell) \in E[G']$ and $I[(u, w)] \geq k$}{
                \If{$\mathbf{S}^{\ell}_{(v, w)} + 2 = I[(v, w)]$}{
                    $B \leftarrow B \cup \{(v, w)\}$;\\
                }
                \If{$\mathbf{S}^{\ell}_{(u, w)} + 2 = I[(u, w)]$}{
                    $B \leftarrow B \cup \{(u, w)\}$;\\
                }
                 $\mathbf{S}^{\ell}_{(v, w)} \leftarrow \mathbf{S}^\ell_{(v, w)} - 1$; $\mathbf{S}^{\ell}_{(u, w)} \leftarrow \mathbf{S}^\ell_{(u, w)} - 1$;\\
            }
        }
        \ForAll{$\varphi'=(w, t) \in B$}{
            Update $I[\varphi']$;\\
            \If{$I[\varphi'] < k$}{
                $N \leftarrow N \cup \{\varphi'\}$;
            }
        }
        Remove all instance of $\varphi$ from $G'$ in all layers;
    }
    $H \leftarrow $ The connected component of $G'$ containing $Q$;\\
    \Return $H$;
    \caption{Maximal $(k, \lambda)$-FirmTruss containing $Q$}
    \label{alg:FirmTruss_G0}
\end{algorithm}

\begin{theorem}[FTCS-Global Quality Approximation]\label{theorem:FCCS-Global-approx-quality}
Algorithm~\ref{alg:binary_search} achieves 2-approximation to an optimal solution $G[H^*]$ of the FTCS problem, that is, the obtained $(k, \lambda)$-FirmTruss, $G[H]$ satisfies  \begin{equation*}
    diam(G[H])\leq 2 \times diam(G[H^*]).
\end{equation*}
\end{theorem}

\begin{lemma}\label{lemma:time_complexity_FirmTruss}
Algorithm~\ref{alg:FirmTruss_G0} takes $\mathcal{O}(\sum_{\ell \in L} |E_\ell|^{1.5} + |E||L| + |E|\lambda \log |L|)$ time, and $\mathcal{O}(|E||L|)$ space. 
\end{lemma}

\begin{theorem}[FTCS-Global Complexity] \label{theorem:FTCS-Global-Complexity}
Algorithm~\ref{alg:binary_search} takes \\ $\mathcal{O}(\gamma (|Q| |E[G_0]| + \sum_{\ell \in L} |E_\ell|^{1.5}) + |E||L| + |E|\lambda \log |L|)$ time, and $\mathcal{O}(|E||L|)$ space, where $\gamma = \log \left(dist_{G_0}(G_0, Q) \right)$.
\end{theorem}

\subsection{Local Search}
The top-down approach of the Global algorithm may incur unnecessary computations over massive networks. The FTCS Local algorithm (Algorithm~\ref{alg:binary_local_search}), presented next, addresses this limitation using  a bottom-up approach. 

\begin{algorithm}[t] 
    \small
    \SetKwInOut{Input}{Input}
    \SetKwInOut{Output}{Output}
    \Input{An ML graph $G = (V, E, L)$, a set of query vertices $Q \subseteq V$, and two integers $k\geq 2$ and $\lambda \geq 1$}
    \Output{A connected $(k, \lambda)$-FT containing $Q$ with a small diameter}
    
    $d_{\min} \leftarrow 1$; $d_{\text{mid}} \leftarrow 1$; $G_{\text{out}} \leftarrow \emptyset$ ; $d_{\max} \leftarrow \infty$; $V'=\emptyset$; \\
    
    \While{$d_{\min} < d_{\max}$ \textbf{and} $V' \neq V$}{
        $V' \leftarrow Q \cup \{ u \in V | \text{dist}_{G}(u, Q) \leq d_{\text{mid}}\}$;\\
        $G' \leftarrow$ Induced subgraph of $G$ by vertices $V'$;\\
        $G' \leftarrow$ Find maximal $(k, \lambda)$-FirmTruss of $G'$ containing $Q$; \\
        \While{$G' \neq \emptyset$}{
            $N \leftarrow \emptyset$;\\
            \For{$u \in V[G']$}{
                \If{$\text{dist}_{G'}(u, Q) > d_{\text{mid}}$}{
                    $N \leftarrow N \cup \{ u \}$;
                }
            }
            \If{$N = \emptyset$}{
                $d_{\max} \leftarrow d_{\text{mid}}$;
                $d_{\text{mid}} \leftarrow \floor{\frac{d_{\min} + d_{\max}}{2}}$; \\
                $G_{\text{out}} \leftarrow G'$;\\
                \Break; //Break in the inner \textbf{while} loop
            }
            \Else{
                Delete~$N$~and~their~incidents~edges~in~all~layers~from~$G'$;\\
                Maintain $G'$ as $(k, \lambda)$-FirmTruss;
            }
        }
        \If{$G' = \emptyset$}{
            $d_{\min} \leftarrow d_{\text{mid}} + 1$;
            $d_{\text{mid}} \leftarrow 2 \times d_{\text{mid}}$;
        }
    }
    \Return $G_{\text{out}}$;
    \caption{FTCS Local Search}
    \label{alg:binary_local_search}
\end{algorithm}

 We can first to collect all vertices whose query distances are $\leq d$ into $V'$ (line 3) and then construct $G'$ as the induced subgraph of $G$ by $V'$ (line 4). Next, given $d$, examine whether $G'$ contains a $(k, \lambda)$-FirmTruss whose query distance is $d$. If such a FirmTruss exists,  return it as the solution, and otherwise, increment  $d$ by 1 and iterate. One drawback of this approach is that it increases the query distance only by 1 in each iteration, which is  inefficient. We instead conduct a binary search on the value of $d$. One challenge is the lack of upper bound on  $d$. A trivial upper bound, which is the query distance in the entire graph, might lead to considering almost the entire graph in the first iteration. \eat{is opposed to the motivation of a method in a bottom-up manner.} 
 \eat{Here we propose a modified version of the binary search, where we first multiply the query distance $d$ by 2 in each iteration until a solution is found. } 
 We instead use a doubling search whereby we double the query distance $d$ in every iteration until a solution is found. Then by considering the resulting query distance as an upper bound on $d$, we conduct a binary search. Algorithm~\ref{alg:binary_local_search} shows the details.



\begin{theorem}[FTCS-Local Quality Approximation]\label{theorem:FTCS_local_quality}
Algorithm~\ref{alg:binary_local_search} achieves 2-approximation to an optimal solution $G[H^*]$ of the FTCS problem, that is, the obtained $(k, \lambda)$-FirmTruss, $G[H]$  satisfies  
\begin{equation*}
    diam(G[H])\leq 2 \times diam(G[H^*]).
\end{equation*}
\end{theorem}
\begin{proofS}
We first prove that the binary search method finds a solution with a smaller query distance than the optimal diameter solution. Next, by the triangle inequality, we show that the diameter of the found solution is at most twice the  optimal. The detailed proof can be found in Appendix~\ref{app:proof_theorems}.
\end{proofS}

\begin{theorem}[FTCS-Local Complexity]\label{theorem:FTCS-local-complexity}
FTCS-Local algorithm takes $\mathcal{O}(\gamma (|Q| |E| + \sum_{\ell \in L} |E_\ell|^{1.5}) + |E||L| + |E|\lambda \log |L|)$ time, and $\mathcal{O}(|E||L|)$ space, where $\gamma = \log \left(dist_{G_0}(G_0, Q) \right)$.
\end{theorem}

\section{Index-based Algorithm}\label{sec:index_method}

Both online algorithms need to find FirmTruss from scratch. However, for each query set, computing the maximal FirmTruss from scratch can be  inefficient for large multilayer networks. \eat{While all the possible maximal FirmTrusses of a given multilayer graph are determined, we can precompute and store them in indices for speeding up our algorithms. Motivated by this,} 
In this section, we discuss how to employ FirmTruss decomposition to accelerate our algorithms, by storing maximal FirmTrusses as they are identified into an index structure. 
\eat{We first develop an index to store the results of FirmTruss decomposition i.e., to compute all the possible FirmTrusses for a given multilayer graph, and then present the details of FirmTruss decomposition algorithm.} 
We first present our FirmTruss decomposition algorithm and then  describe how the index can be used for efficient retrieval of the maximal FirmTruss given a query.

\subsection{FirmTruss Decomposition}\label{sec:FirmTruss-Decomposition}
In this section, we define the Skyline FirmTrussness index. For an edge schema $\varphi \in \mathcal{E}$, we let $FTI(\varphi)$ denote the set $\{(k, \lambda) \mid \varphi \mbox{ is in a } (k, \lambda)\mbox{-FirmTruss}\}$. We will use  the following notion of index dominance.

\begin{dfn}[Index Dominance]
Given two pairs of numbers $(k_1, \lambda_1)$ and $(k_2, \lambda_2)$, we say $(k_1, \lambda_1)$ \textit{dominates} $(k_2, \lambda_2)$, denoted  $(k_2, \lambda_2) \preceq (k_1, \lambda_1)$, provided $k_1 \geq k_2$ and $\lambda_1 \geq \lambda_2$.
\end{dfn}

Clearly, $(FTI(\varphi), \preceq)$ is a partial order.

\begin{dfn}[Skyline FirmTrussness]
Let $\varphi \in \mathcal{E}$ be an edge schema. \eat{and let $FT(\varphi)$ denote the set of all the FirmTruss indices $(k_i, \lambda_i)$ such that $\varphi$ is in a $(k_i, \lambda_i)$-FirmTruss. } The \textit{skyline FirmTrussness} of $\varphi$, denoted  $\sft{\varphi}$, contains the maximal elements of $FTI(\varphi)$. 
\eat{ 
\begin{equation*}
    \sft{\varphi} = \{ (k_i, \lambda_i) \in FT(\varphi) | \nexists (k_j, \lambda_j):  (k_i, \lambda_i) \prec (k_j, \lambda_j) \}.
\end{equation*}
} 
\end{dfn}

In order to find all possible FirmTrusses, we only need to compute the skyline FirmTrussness for every edge schema in a  multilayer graph $G$. To this end, we present the details of FirmTruss algorithm in Algorithm~\ref{alg:FirmTruss_decomposition}. For a given edge schema $\varphi$, if Top$-\lambda(\mathbf{S}_\varphi) = k - 2$, then it cannot be a part of a $(k', \lambda)$-FirmTruss, for $k' > k$. Therefore, given $\lambda$, we can consider Top$-\lambda(\mathbf{S}_\varphi) + 2$ as an upper bound on the FirmTruss index of $\varphi$ (line 5). In the FirmTruss decomposition, we recursively pick an edge schema $\varphi$ with the lowest Top$-\lambda(\mathbf{S}_\varphi)$, assign its FirmTruss index as Top$-\lambda(\mathbf{S}_\varphi) + 2$, and then remove it from~the~graph. After that, to efficiently update the Top-$\lambda$ support of its adjacent edges, we use Fact~\ref{fact:FirmTruss_algorithm2} (lines 13-16). At the end of the algorithm, we remove all dominated indices in $\sft{\varphi}$ for each $\varphi \in \mathcal{E}$ to only store skyline indices (line 23).  We can show: 

\begin{theorem}[FirmTruss Decomposition Complexity]\label{theorem:FirmTruss_Decomposition_Complexity}
Algorithm~\ref{alg:FirmTruss_decomposition} takes $\mathcal{O}(\sum_{\ell \in L} |E_\ell|^{1.5} + |E||L|^2)$ time.
\end{theorem}

\subsection{Index-based Maximal FirmTruss Search}
Using 
Algorithm~\ref{alg:FirmTruss_decomposition}, we can find \new{offline} all skyline FirmTruss indices for a given edge schema and query vertex set. \new{Next, we  start from the query vertices and by using a breadth-first search, check for each neighbor whether its corresponding edge schema has a skyline FirmTruss index that dominates the input $(k, \lambda)$. Algorithm~\ref{alg:maximal_firmtruss_index} shows the procedure.} We have:

\begin{theorem}
Algorithm~\ref{alg:maximal_firmtruss_index} takes $\mathcal{O}(|E[G_0]|)$ time. 
\end{theorem}
This indexing approach can be used in Algorithm~\ref{alg:binary_search} to find the maximal $G_0$, as well as in Algorithm~\ref{alg:binary_local_search}  so that we only need to add edges whose corresponding edge schema has an index that dominates $(k, \lambda)$. We refer to these variants of Global and Local as iGlobal and iLocal, respectively. 

\begin{algorithm}[t] 
    \small
    \SetKwInOut{Input}{Input}
    \SetKwInOut{Output}{Output}
    \Input{An ML graph $G = (V, E, L)$, a set of query vertices $Q \subseteq V$, \new{$\mathsf{SFT}$ indices,} and two integers $k\geq 2$ and $\lambda \geq 1$}
    \Output{A maximal connected $(k, \lambda)$-FirmTruss containing $Q$}

    $G_0 \leftarrow \emptyset$; $N \leftarrow Q$; \\
    \While{$N \neq \emptyset$}{
        Pick and remove $u \in N$;\\
        \For{each unvisited edge schema $\varphi = (u, v)$}{
            Mark $\varphi$ as visited;\\
            \For{each skyline FirmTruss index $(k_i, \lambda_i) \in \mathsf{SFT(\varphi)}$}{
                \If{$(k, \lambda) \preceq (k_i, \lambda_i)$}{
                    add $v$ and $u$ with all their incident edges into $G_0$;\\
                    $N \leftarrow N \cup \{ v\}$;\\
                }
            }
        }
    }
    \Return $G_0$;
    \caption{Index-based Maximal FirmTruss Finding}
    \label{alg:maximal_firmtruss_index}
\end{algorithm}

\section{Attributed FirmTruss Community}
\label{sec:AttributedFirmCommunities}



Often networks come naturally endowed with attributes associated with their nodes. For example, in DBLP, authors may have areas of interest as attributes. In protein-protein interaction networks, the attributes may correspond to biological processes, molecular functions, or cellular components of a protein made available through the Gene Ontology (GO) project \cite{gop}. It is natural to impose some level of similarity between a community's members,~based~on~their~attributes.

Network homophily is a phenomenon which states similar nodes are more likely to attach to each other than dissimilar ones. Inspired by this ``birds of a feather flock together'' phenomenon, in social networks, we argue that users remain engaged with their community if they feel enough similarity with others, while users who feel dissimilar from a community may decide to leave the community. Hence, for each node, we measure how similar it is to the community's members and use it to define the~homophily~in~the~community.

We show that surprisingly, use of homophily in a definition of attributed community offers an alternative means to avoid the free-rider effect. In this section, we extend the definition of the FirmTruss-based community to attributed ML networks, where we assume each vertex has an attribute vector. In order to capture vertex similarity, we propose a new function to measure the homophily in a subgraph. We show that this function not only guarantees a high correlation between attributes of vertices in a community but also avoids the free-rider effect. Unlike previous work~\cite{attribute-community, attribute-community2, keyword-community}, our model allows for continuous valued attributes. E.g., in a PPI network, the biological process associated with a protein may have a real value, as opposed to just a boolean or a categorical value. 

\eat{ 
As we discussed in Section~\ref{sec:introduction}, identifying communities only by their cohesiveness might cause an undesirable phenomenon of the free-rider effect. While only the structural properties of a graph are accessible, we use the diameter of a subgraph to avoid including unrelated and far nodes in the community. However, in the presence of richer information, e.g. nodes' attributes, we have more powerful choices to avoid the free-rider effect and find communities with higher quality. In this section, we adopt the definition of the FirmTruss-based community to attributed multilayer networks, where we assume each vertex has an attribute vector. In order to capture attribute correlation, we propose a new function to measure the homophily in a subgraph. We show that this function not only guarantees a high correlation between nodes but also avoids the free-rider effect. 

Network homophily is a phenomenon states similar nodes may be more likely to attach to each other than dissimilar ones. Inspired by this phenomenon, in social networks, we argue that users remain engaged with their community if they feel enough similarity with others—users who feel dissimilarity may decide to leave the community. Accordingly, for each vertex, we measure how it feels similar to the community's members and use these values to measure the homophily in the community.  

Most existing attributed community search models assume that each vertex has a specific categorical attribute (e.g., keyword). However, in many real-world networks, nodes have more complex features; e.g., users' profiles in social networks might include much information about a user. In this study, we consider a vector attribute for each node that embeds and describes all corresponding information about the node. Given these attribute vectors, we measure how similar a pair of nodes are. Based on the network homophily, similar nodes are more likely to attach and build communities. Therefore, a good community has a high average similarity between nodes. Accordingly, for each node, we measure how it is similar to other nodes within the community and consider their p-mean value as the goodness metric of a community.} 

Let $\attr = \{A_1, ..., A_d\}$ be a set of attributes. An \textit{attributed multilayer network} $G = (V, E, L, \Psi)$, where $(V,E,L)$ is a multilayer network and $\Psi : V \rightarrow \mathbb{R}_{\geq 0}^{d}$ is a non-negative function that assigns a $d$-dimensional vector to each vertex, with $\Psi(v)[i]$ representing the strength of attribute $A_i$ in vertex $v$. Let $h(v, u)$ be a symmetric and non-negative similarity measure based on attribute vectors of $u$ and $v$. E.g., $h(v,u)$ can be the cosine similarity between $\Psi(u)$ and $\Psi(v)$. Let $S$ be a community containing $v$. We define $h_S(v)$, capturing the aggregate similarity between $v$ and members of~$S$:
\begin{equation*}
    h_S(v) = \sum_{\substack{u \in S \\ u \neq v}} h(v, u).
\end{equation*}

The higher the value $h_S(v)$ the more similar user $v$ ``feels" they are with the community $S$. While cosine similarity of attribute vectors is a natural way to compute the similarity $h(v, u)$, any symmetric and non-negative measure can be used in its place. 

\eat{In other words, the value of $h_S(v)$ represents the welfare of $v$ in community $S$. Also, while our formulation works for any symmetric and non-negative similarity measure, for consistency, in the continue, we use cosine similarity.
\eat{\begin{equation*}
    h(u, v) = \frac{\Psi(u).\Psi(v)}{\|\Psi(u)\| \|\Psi(v)\|}.
\end{equation*}}} 

Based on $h_S(v)$, we define the \textit{homophily score} of community $S$ as follows. Let $p \in \mathbb{R} \cup \{+\infty, -\infty\}$ be any number. Then the homophily score of $S$ is defined as: 
\begin{equation*}
    \Gamma_p(S) = \left( \frac{1}{|S|} \sum_{v \in S} h_S(v)^p \right)^{1/p}.
\end{equation*}

\eat{ 
Finally, given $p \in \mathbb{R} \cup \{+\infty, -\infty\}$, we define the homophily score of a subgraph $S$ as the power mean of the $h_S(v)$ values, where $v \in S$: 
\begin{equation*}
    \Gamma_p(S) = \left( \frac{1}{|S|} \sum_{v \in S} h_S(v)^p \right)^{1/p}.
\end{equation*}
}

The parameter $p$ gives flexibility for controlling the emphasis on  similarity at different ends of the spectrum. When $p \rightarrow +\infty$ (resp. $p \rightarrow -\infty$) , we have higher emphasis on large (resp. small) similarities. This flexibility allows us to tailor the homophily score to the application at hand.

\eat{ is desirable for the homophily score function since \textbf{(1)} the level emphasis on low/high similarities of a community is application-dependent, \textbf{(2)} it is able to uncover different meaningful notions of subgraphs with high homophily score in practice as we change our parameter $p$.}




\eat{ 
\subsection{Properties of Homophily Score Function}

\head{Positive Influence of Similar Vertices}
The more similar vertices a community has, the higher the homophily score.

\begin{fact}\label{fact:positive_influence}
Let $\Delta_{u}(S \cup \{ u \})$ be the exact value that adding vertex $u$ to a subgraph $S$ increases the $|S|\Gamma_p^p(S)$ (we formally define it in Equation~\ref{eq:Xi}), if $\Delta_{u}(S \cup \{ u \}) > \Gamma_p^p(S)$, adding $u$ increases the homophily score of subgraph $S$.
\end{fact}

\begin{example}
Let $p = 1$, Fact~\ref{fact:positive_influence} states that adding a node $u$ with node-homophily score $h_{S}(u) > \frac{1}{2|S|} \sum_{v \in S} h_S(v)$ to subgraph $S$, increases the homophily score of the subgraph.
\end{example}


\head{Negative Influence of Dissimilar Vertices}
Adding dissimilar vertices to the community will decrease the homophily score.

\begin{fact}\label{fact:negative_influence}
If $\:\Delta_{u}(S \cup \{ u \}) < \Gamma_p^p(S)$, adding $u$ decreases the homophily score of subgraph $S$.
\end{fact}

\begin{example}
Let $p = 1$, Fact~\ref{fact:negative_influence} states that adding a node $u$ with node-homophily score $h_{S}(u) < \frac{1}{2|S|} \sum_{v \in S} h_S(v)$ to subgraph $S$, decreases the homophily score of the subgraph.
\end{example}

\head{Non-monotone Property}
Non-monotonicity is a desirable property for the community model. Assume that if the attribute correlation measure for the community model is monotone and increasing, then always the maximal FirmTruss is the optimal solution. At the same time, we know that maximal FirmTruss may include some free riders. The introduced homophily score function is non-monotone: as we discussed in Facts~\ref{fact:positive_influence} and \ref{fact:negative_influence}, for a node $u$ based on the value of $h_S(u)$, adding $u$ may increase or decrease the homophily score.

\head{Non-submodularity and Non-supermodularity}
Optimization problems over submodular or supermodular functions lend themselves to efficient approximation. We thus study whether our homophily score function is submodular with respect to the set of vertices. By a contradictory example, we show that our homophily score function is neither submodular nor supermodular.

\begin{example}
Consider a graph $G$ with $V = \{u_1, u_2, v_1, v_2\}$. Let $p=1$. Assume that $h(v_1,u_2) = h (v_2,u_1) = 0.1$, $h(v_2,u_2) = 0.5$, $h(v_1,u_1) =0.2$, $h(u_1,u_2) = 0.3$ and $h(v_1,v_2) = 0$. Now consider the sets $S = \{u_1\}$ and $T =  \{u_1,u_2\}$. Let us compare the marginal gains $\Gamma_p(S \cup \{v^*\}) - \Gamma_p(S)$ and  $\Gamma_p(T \cup \{v^*\}) - \Gamma_p(T)$ , from adding the new vertex $v^* \notin T$ to $S$ and $T$. Suppose $v^*=v_1$, then we have $\Gamma_p(S \cup \{v_1\}) - \Gamma_p(S) = (2 \times 0.2)/2- 0 = 0.2 > \Gamma_p(T \cup \{v_1\}) - \Gamma_p(T) = 2(0.1+0.2+0.3)/3 - (2\times0.3)/2= 0.1$, violating supermodularity. On the other hand, suppose $v^*=v_2$. Then we have $\Gamma_p(S \cup \{v_2\}) - \Gamma_p(S) = (2 \times 0.1)/2- 0 = 0.1 < \Gamma_p(T \cup \{v_2\}) - \Gamma_p(T) = 2(0.1+0.5+0.3)/3 - (2\times0.3)/2= 0.3$, which violates submodularity.  Thus, this function is non-submodular and non-supermodular.
\end{example}
} 

\subsection{Attributed FirmCommunity Model}

\begin{problem}[Attributed FirmTruss Community Search]\label{prob:AFTCS} Given an attributed ML network $G = (V, E, L, \Psi)$, two integers $k \geq 2, \lambda \geq 1$, a parameter $p \in \mathbb{R} \cup \{+\infty, -\infty\}$, and a set of query vertices $Q \subseteq V$, the attributed FirmTruss community search (AFTCS) is to find a connected subgraph $G[H] \subseteq G$ satisfying:
\begin{enumerate}
    \item $Q \subseteq H$,
    \item $G[H]$ is a connected $(k, \lambda)$-FirmTruss,
    \item $\Gamma_p(H)$ is the maximum among all subgraphs satisfying (1) and (2).
\end{enumerate}
\end{problem}

\head{Hardness Analysis} Next we analyze the complexity of the AFTCS problem and show that when $p$ is finite, it is NP-hard.

\begin{theorem}[AFTCS Hardness]\label{theorem:AFTCS-Hardness}
The AFTCS problem is NP-hard, whenever $p$ is finite. 
\end{theorem}

\begin{proofS}
\eat{The crux of our proof idea comes from the hardness of}  Finding the densest subgraph with  $\geq k$ vertices in single-layer graphs~\cite{exact-directed} is a hard problem. Given an instance of this problem, $G = (V, E)$, we construct a complete, attributed ML  graph and provide an approach to construct an attribute vector of each node such that $\forall$ vertices $u, v$, $h(u,v) = \frac{1}{2|V|}$ if  $(u, v) \in E$, and  $h(u,v) = 0$, if $(u, v) \notin E$. So the densest subgraph with $\geq k$ vertices in $G$ is a solution for AFTCS, and vice versa.
\end{proofS}

\head{Free-rider Effect} Analogously to Theorem~\ref{theorem:FTCS-free-rider}, we can show:

\begin{theorem}[AFTCS Free-Rider Effect]\label{theorem:AFTCS-free-rider}
For any attributed ML network $G = (V, E, L, \Psi)$ and query vertices $Q \subseteq V$, there is a solution $G[H]$ to the AFTCS problem such that for all query-independent optimal solutions $G[H^*]$, either $H^* = H$, or $G[H \cup H^*]$ is disconnected, or $G[H \cup H^*]$ has a strictly smaller homophily score than $G[H]$.
\end{theorem}

\eat{Analogously to Theorem~\ref{theorem:FTCS-free-rider}, the theorem shows that use of attribute based homophily leads to avoidance of the free-rider-effect. } 

\subsection{Algorithms}
In this section, we propose an efficient approximation algorithm for the AFTCS problem. We show that when $p = +\infty$, or $-\infty$, this algorithm finds the exact solution. We can show that our objective function $\Gamma_p(.)$ is neither submodular nor supermodular (proof in Appendix~\ref{app:homophily_property}), suggesting this problem may be hard to approximate, for some values of $p$. 

\head{Peeling Approximation Algorithm}
We divide the problem into two cases: (i) $p > 0$, and (ii) $p < 0$. For finite $p > 0$, $\arg \max \Gamma_p(S) = \arg \max \Gamma^p_p(S)$, so for simplicity, we focus on maximizing $\Gamma^p_p(.)$. 
\eat{On the other hand, for finite $p < 0$, $\arg \max \Gamma_p(S) = \arg \min \Gamma^p_p(S)$, so for simplicity, in this case, we focus on minimizing $\Gamma^p_p(.)$ function.} 
Similarly, for finite $p < 0$, we focus on minimizing $\Gamma^p_p(.)$. Note that, for any finite $p$, an $\alpha$-approximate solution for optimizing  $\Gamma_p^p(.)$ provides an $\alpha^{1/p}$-approximate solution for optimizing~$\Gamma_p(.)$.

Consider a set of vertices $S \subseteq V$. Our approximation algorithm is to greedily remove nodes $u \in S$ that may improve the objective. Since removing any node $u \in S$ will change the denominator of $\Gamma^p_p(S)$ in the same way, we can choose the node that leads to the minimum (maximum) drop in the numerator. Let us examine the change to the $\Gamma_p^p(S)$ from dropping $u \in S$:

\begin{equation*}
    \Gamma_p^p(S\setminus\{u\}) = \frac{\sum_{v \in S\setminus\{u\}} h_{S\setminus\{u\}}(v)^p}{|S| - 1} = \frac{1}{|S| - 1} \left( |S|\cdot \Gamma^p_p(S) - \Delta_u(S) \right),
\end{equation*}

\noindent
where 
\begin{equation*}\label{eq:Xi}
    \Delta_u(S) = h_S(u)^p + \left( \sum_{v \in S\setminus\{u\}} h_S(v)^p - \left[ h_S(v) - h(v, u) \right]^p \right).
\end{equation*}

\noindent
Notice that $\Delta_u(S)$ represents the exact decrease in the numerator of $\Gamma_p^p(S)$ resulting from removing $u$. Based on this observation, in Algorithm~\ref{alg:approx-AFTCS}, we recursively remove a vertex with a minimum (maximum) $\Delta$ value, and maintain the remaining subgraph as a $(k, \lambda)$-FirmTruss. We have the following result: 

\begin{algorithm}[t]
    \small
    \SetKwInOut{Input}{Input}
    \SetKwInOut{Output}{Output}
    \Input{An attributed ML graph $G = (V, E, L, \Psi)$, a set of query vertices $Q \subseteq V$, and two integers $k\geq 2$ and $\lambda \geq 1$}
    \Output{A connected $(k, \lambda)$-FT containing $Q$ with a large $\Gamma_{p}(.)$}

    $G_0 \leftarrow$ Find a maximal connected $(k, \lambda)$-FirmTruss containing $Q$; \eat{//See Algorithm~\ref{alg:maximal_firmtruss_index}} \\
    Calculate $h_{V[G_0]}(u)$ for all $u \in V[G_0]$;  $i \leftarrow 0$;\\
    
    \While{$Q \subseteq V[G_i]$}{
        \If{$p > 0$}{
        $u \leftarrow \arg \min_{u \in V[G_i]} \Delta_u(V[G_i])$;
        }
        \Else{
        $u \leftarrow \arg \max_{u \in V[G_i]} \Delta_u(V[G_i])$;
        }
        Delete vertex $u$ and its incident edges from $G_i$ in all layers;\\
        Maintain $G_i$ as $(k, \lambda)$-FirmTruss by removing vertices/edges; \\
        Let the remaining graph as $G_{i+1}$; $i \leftarrow i + 1$;
    }
    \Return $\arg \max_{H \in \{G_0, \dots, G_{i-1}\}} \Gamma_p(H)$;
    \caption{AFTCS-Approx}
    \label{alg:approx-AFTCS}
\end{algorithm}

\begin{theorem}[AFTCS-Approx Complexity]\label{theorem:AFTCS-approx-complexity}
Algorithm~\ref{alg:approx-AFTCS} takes $\mathcal{O}(d|V_0|^2 + t(|V_0| + |E_0|)+\sum_{\ell \in L} |E_\ell|^{1.5} + |E||L| + |E|\lambda \log |L|)$ time, and $\mathcal{O}(|E||L| + |V_0|^2)$ space, where $t$ is the number of iterations, $V_0$ and $E_0$ are the vertex set and edge set of maximal $(k, \lambda)$-FirmTruss. 
\end{theorem}

\eat{To examine the quality guarantee of this algorithm, we need to divide the problem into two sub-problems, (1) $p \geq 1$, and (2) $p < 1$. The reason is the behaviour of function $\Gamma_p(.)$ is completely different in each of these intervals. As we discussed, removing any node will change the denominator of $\Gamma_p^p(S)$ in exactly the same way, so we need to focus on the behaviour of its numerator. It is not hard to show that the set of all connected subgraphs that are FirmTruss is closed under union and intersection. Therefore, the set of FirmTrusses in a graph is a lattice family. While $p < 1$, our problem is closely related to the problem of (approximately) maximizing submodular functions over a lattice family, which is still an open problem. Accordingly, we provide an approximation guarantee only for the case of $p \geq 1$. Although our method provides no formal guarantees when $p < 1$, it achieves the results of good quality in practice, as validated in our experiments.} 

As for the approximation quality, we can show the following  when $p\geq 1$. The detailed proof and tightness example can be found in Appendix~\ref{app:proof_theorems} and \ref{app:tightness_examples}.


\begin{theorem}[AFTCS-Approx Quality]\label{theorem:AFTCS-approx-quality}
Let $p~\geq~1$, Algorithm~\ref{alg:approx-AFTCS} returns a $(p+1)^{1/p}$-approximation solution of AFTCS problem.
\end{theorem}

\begin{proofS}
Let $H^*$ be the optimal solution. Since removing a node $u^* \in H^*$ will produce a subgraph with homophily score at most $\Gamma_p^p(H^*)$, we have $\Gamma^p_p(H^*) \leq \Delta_{u^*}(H^*)$. Next, we show that the first removed node $u^* \in H^*$ by the algorithm cannot be removed by maintaining FirmTruss, so it was a node with a minimum $\Delta$. Then, we use the fact that the minimum value of $\Delta$ is less than the average of $\Delta$ over all nodes and provide an upper bound of $(p + 1) \Gamma^p_p(S)$ for the average of $\Delta$ over $S$. Finally, we show that function $|S|\Gamma^p_p(S)$ is supermodular for $p \geq 1$, and based on its increasing differences property, we~conclude~the~approximation~guarantee.
\end{proofS}

\begin{remark}
 How much good can Algorithm~\ref{alg:approx-AFTCS} work? As $p$ increases, Algorithm~\ref{alg:approx-AFTCS} has a better approximation factor. In the worst case, ($p = 1$), we get approximation factor = $2$, and when $p \rightarrow \infty$, our approximation factor has a limit of 1. This limit of the approximation factor intuitively matches the fact that when $p=+\infty$ the optimal solution is trivial to obtain by the maximal FirmTruss. 
\end{remark}

\head{Exact Algorithm when $p = +\infty$, or $-\infty$} 
The case $p = +\infty$ is straightforward, where we just want to maximize $\Gamma_{+\infty}(S) = \max_{v \in S} h_S(v)$. The solution of this case is the maximal subgraph that satisfies the conditions (1) and (2) in Problem~\ref{prob:AFTCS}. In the $p = -\infty$ case, we want to maximize $\Gamma_{-\infty}(S) = \min_{v \in S} h_S(v)$. We can recursively remove a vertex with minimum value of $h_S$ and maintain the remaining subgraph such that satisfies  conditions (1) and (2) in Problem~\ref{prob:AFTCS}. The pseudocode is identical to Algorithm~\ref{alg:approx-AFTCS}, except in  lines 5-8, we recursively remove a vertex with a minimum value of $h_S$. We refer to this modified peeling algorithm as Exact-MaxMin.

\begin{theorem}[Correctness of Exact-MaxMin]\label{theorem:Correctness_maxmin_AFTCS}
Exact-MaxMin returns the exact solution to the AFTCS problem with $p = -\infty$. 
\end{theorem}


\begin{center}
\begin{table} [tpb!]
 \caption{Network Statistics}
 \vspace{-2ex}
    \resizebox{0.38\textwidth}{!} {
\begin{tabular}{l | c c c c c c c c c}
 \toprule
  {Dataset} & {$|V|$} &  {$|E|$} & {$|L|$} & \new{Size} & \new{\#FT} & {Attribute} & {GT}\\
 \midrule 
 \midrule
    Terrorist     &  79     &  2.2K   &   14  & \new{17 KB}  & \new{48} & \checkmark    & \checkmark  \\
    RM            &  91     &  14K    &   10  & \new{112 KB} & \new{113} & \checkmark    & \checkmark \\
    FAO           & 214     &  319K   &  364  &  \new{3 MB}   &   \new{2397}  &        &     \\
    Brain         & 190     &  934K   &  520  & \new{10 MB}  &    \new{1493}    &       & \checkmark \\ 
    DBLP          &  513K   &  1.0M   &  10   & \new{16 MB}  & \new{66} & \checkmark    & \checkmark   \\
    Obama         & 2.2M    &  3.8M   &   3   &  \new{60 MB}  & \new{20}   &        &     \\
    YouTube       & 15K     &  5.6M   &   4   & \new{106 MB}  & \new{372} & \checkmark    & \checkmark \\
    Amazon        & 410K    &  8.1M   &   4   &  \new{123 MB}   &   \new{23}    &     & \checkmark    \\
    YEAST         & 4.5K    &  8.5M   &   4   &  \new{97 MB}   &  \new{542}     &     &              \\
    Higgs         & 456K    &  13M    &   4   &  \new{205 MB}    &  \new{94}   &      &     \\
    Friendfeed    & 510K    &  18M    &   3   &  \new{291 MB}    &   \new{320}  &      &     \\
    StackOverflow &  2.6M   &  47.9M  &   24  &  \new{825 MB}    &  \new{1098}   &      &     \\
    Google+       &  28.9M  &  1.19B  &   4   &  \new{20 GB}    &   \new{-}   &     &     \\
 \bottomrule
 \multicolumn{2}{l}{\new{Size: graph size}} & \multicolumn{4}{l}{\new{$\#$FT: number of FirmTrusses}} & \multicolumn{2}{l}{\new{GT: ground truth}}
\end{tabular}
}
 \label{tab:datastat}
 \vspace{0.5ex}
\end{table}
\end{center}
\vspace{-2ex}
\section{Experiments}
\label{sec:experiments}

We conduct experiments to evaluate the proposed CS  models and algorithms. Additional experiments on efficiency and parameter sensitivity can be found in Appendix~\ref{app:additional_experiments}.

\head{Setup}
All algorithms are implemented in Python and compiled by Cython. The experiments are performed on a Linux machine with Intel Xeon 2.6 GHz CPU and 128 GB RAM.

\head{Baseline Methods}
 We compare our FTCS with the state-of-the-art CS methods in ML networks. ML $\mathbf{k}$-core~\cite{ml-core-journal} uses an objective function to automatically choose a subset of layers and finds a subgraph such that the minimum of per-layer minimum degrees, across selected layers, is maximized. ML-LCD~\cite{ML-LCD} maximizes the ratio of Jaccard similarity between nodes inside and outside of the local community. RWM~\cite{ML-random-walk} sends random walkers in each layer to obtain the local proximity w.r.t. the query nodes and returns a subgraph with the smallest conductance. We implemented a baseline based on TrussCube~\cite{Truss_cube}, which finds a maximal connected TrussCube containing query nodes. We compare our approach with CTC~\cite{closest}, which finds the closest truss community in single-layer graphs, and VAC~\cite{VAC}, an attributed variant of CTC, also on single-layer graphs.

\head{Datasets}
We perform extensive experiments on thirteen real networks ~\cite{KONECT, MLcore, Friendfeed, Higgs, amazon_datset, Google+, FAO, homo, Twitter_datasets, Noordin-dataset, RM, brain_dataset} covering social, genetic, co-authorship, financial, brain, and co-purchasing networks, whose main characteristics are summarized in Table~\ref{tab:datastat}. While Terrorist and DBLP datasets naturally have attributes, for RM and YouTube, we chose one of the layers, embedded it  using node2vec \cite{node2vec}, and used the vector representation of each node as its attribute vector. \eat{More details about datasets can be found in Appendix~\ref{app:datasets}.}

\head{Queries and Evaluation Metrics}
We evaluate the performance of all algorithms using different queries by varying the number of query nodes, and the parameters $k$, $\lambda$, and $p$. To evaluate the quality of found communities $C$, we measure their F1-score to grade their  alignment with the ground truth $\tilde{C}$. \new{Here, $F1(C, \tilde{C}) = \frac{2pre(C, \tilde{C})rec(C, \tilde{C})}{pre(C, \tilde{C}) + rec(C, \tilde{C})}$, where $pre(C, \tilde{C}) = \frac{|C \cap \tilde{C}|}{|C|}$ and $rec(C, \tilde{C}) = \frac{|C \cap \tilde{C}|}{|\tilde{C}|}$.} To evaluate the efficiency, we report the running time. In reporting results, we cap the running time at 5 hours and memory footprint at 100 GB.  For index-based methods, we cap the construction time at 24 hours. Unless stated otherwise, we run our algorithms over 100 random query sets with a random size between 1 and 10, and report the average results. We randomly set $k$ and $\lambda$ to one of the common skyline indices of edge schemas incident to query nodes. 

\begin{figure}
    \centering
    \includegraphics[width=0.47\textwidth]{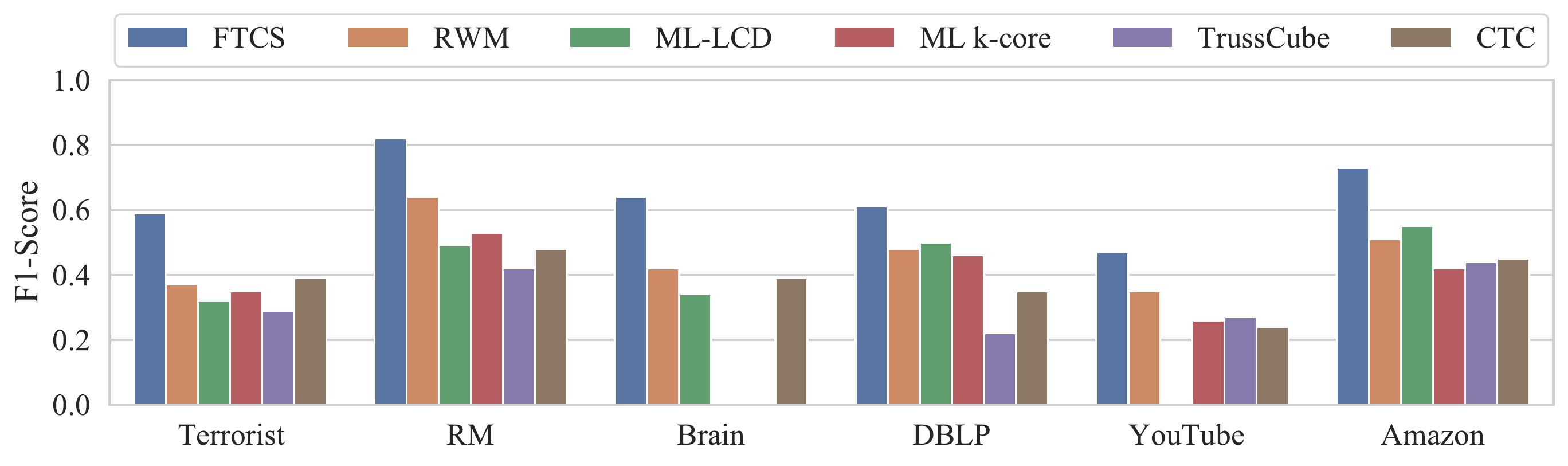}
    \vspace{-2ex}
    \caption{Quality evaluation on ground-truth networks.}
    \label{fig:F1-score}
\end{figure}
\begin{center}
\begin{table} [tpb!]
    \caption{Evaluation of FTCS with the state-of-the-art methods on datasets without ground truth.}
    \resizebox{0.48\textwidth}{!} {
\begin{tabular}{l | c c |c c |c c| c c}
 \toprule
  \multirow{2}{*}{CS Model} & \multicolumn{2}{c|}{FAO} &  \multicolumn{2}{c|}{Obama} & \multicolumn{2}{c|}{YEAST} & \multicolumn{2}{c}{Higgs} \\
  \cline{2-9}
  & Density & Diameter  & Density & Diameter  & Density & Diameter & Density & Diameter \\
 \midrule 
 \midrule
    FTCS                                & \textbf{979.71} & \textbf{1} & \textbf{9.81} & \textbf{1.84}        & \textbf{177.27} &        \textbf{1.52} & \textbf{65.14} &  \textbf{1.93}   \\
    ML $\mathbf{k}$-core                & - & - & 8.13 &      $\infty$  & 159.94 &   $\infty$      & 59.41 & $\infty$     \\
    ML-LCD                              & 952.88 & 1.09 & 4.87 &    2.46    & - &         - & - &   -   \\
    RWM                                 & 911.94 & 1.12 & 4.62 &    3.07    & 25.45 &        1.84 & 24.99 &   3.16     \\
    \new{TrussCube}    & \new{-} & \new{-} & \new{4.71} &    \new{2.03}    & \new{147.33} &   \new{1.87} & \new{26.89} &   \new{2.14}     \\
    CTC                                 & 733.85 & \textbf{1} & 5.35 &   1.99     & 139.03 &        1.92 & 35.18 &  2.05      \\
  \toprule
\end{tabular}
}
 \label{tab:comparison_effectiveness_FTCS}
\end{table}
\end{center}

\head{Quality}
We evaluate the effectiveness of different community search models over multilayer networks. Figure~\ref{fig:F1-score} reports the average F1-scores of all methods on datasets with the ground-truth community. We observe that our approach achieves the highest F1-score on all networks against baselines. The reason  is two-fold. First, in our problem definition, we enforce the minimum-diameter restriction,  effectively removing the irrelevant vertices from the result. Second, FirmTruss requires each edge schema to have enough support in a sufficient number of layers, ensuring that the found subgraphs are cohesive and densely connected. While CTC also minimizes the diameter, it is a single-layer approach and misses some structure due to ignoring the type of connections.

We also evaluate all algorithms in terms of other goodness metrics -- density ($\beta = 1$), and diameter. Table~\ref{tab:comparison_effectiveness_FTCS} reports the results on FAO, Obama, YEAST, and Higgs datasets. The results on other datasets are similar, and are omitted for lack of space. We observe that our approach achieves the highest density, and lowest diameter on all networks against baselines.

Since there are no prior models for  attributed community search in ML networks, we compare the quality of AFTCS with our ML unattributed baselines as well as  VAC. Table~\ref{tab:comparison_effectiveness_AFTCS} reports the F1-score and density of communities found, over four datasets with ground-truth communities.  AFTCS consistently beats the baselines. Notice that AFTCS has a higher F1-score than FTCS in all but one case, as the existence of both attributes and structure is richer information than only structure. Accordingly, AFTCS is better able to distinguish members from non-members of a ground-truth community. 

\begin{center}
\begin{table} [tpb!]
    \caption{Evaluation of AFTCS with the state-of-the-art methods on attributed datasets with ground-truth.}
    \resizebox{0.48\textwidth}{!} {
\begin{tabular}{l l| c c |c c |c c | cc }
 \toprule
  \multicolumn{2}{c|}{\multirow{2}{*}{CS Model}} & \multicolumn{2}{c|}{Terrorist} &  \multicolumn{2}{c|}{RM} & \multicolumn{2}{c}{DBLP} & \multicolumn{2}{c}{Youtube} \\
  \cline{3-10}
  & & F1 & Density &  F1 & Density & F1 & Density & F1 & Density \\
 \midrule 
 \midrule
    \multirow{5}{*}{AFTCS}  & \new{$p = +\infty$} & \new{0.52} & \new{\textbf{15.29}} &   \new{0.77}  & \new{62.35}                        &  \new{0.62}     & \new{8.29} & \new{0.45} &  \new{\textbf{11.64}}\\
                            & $p = 2$ & 0.52 & \textbf{15.29} &      0.79  & 61.24  &    0.61     & 8.22 & 0.45 &  \textbf{11.64} \\
                           & $p = 1$ & \textbf{0.61} & 15.22 &       0.83 & \textbf{64.31}  &        0.60 & 7.91 & 0.45 &  11.59 \\
                           & $p = 0$ & \textbf{0.61} & 15.18 &       0.82 & 63.98 &        \textbf{0.64} & 8.11 & 0.43 & 10.88  \\
                           & $p = -1$ & 0.59 & 13.76 &      0.81  & 63.19 &        0.61 & 8.19 & 0.44 &  11.24 \\
                            & \new{$p = -2$} & \new{0.56} & \new{13.94} &   \new{0.81}  & \new{63.19}                        &  \new{0.60}     & \new{8.03} & \new{0.46} &  \new{11.49}\\
                           & $p = -\infty$ & 0.57 & 14.08 &     \textbf{0.85}   & 62.46 &    0.62 & 7.97 & 0.46  & 11.49 \\
    \midrule
     \multicolumn{2}{l|}{FTCS} & 0.59 & 10.23 &    0.84 &   60.52     & 0.61 &    \textbf{8.69}    & \textbf{0.47} &  10.36  \\
     \multicolumn{2}{l|}{ML $\mathbf{k}$-core}                & 0.35 & 8.43 &    0.53 &    55.98    & 0.46 &  5.53       & 0.26 &  8.78   \\
     \multicolumn{2}{l|}{ML-LCD}                              & 0.32 & 7.82 &      0.49 &  47.26  & 0.50 &       6.49  & - & -    \\
     \multicolumn{2}{l|}{RWM}                                 & 0.37 & 5.45 &        0.65  & 39.81   & 0.48 &       5.12  & 0.35 &   7.46   \\
     \multicolumn{2}{l|}{VAC}                                 & 0.41 & 7.51 &        0.48  & 52.50 & 0.35 &       5.27  & 0.24 & 4.34 \\
  \toprule
\end{tabular}
}
 \label{tab:comparison_effectiveness_AFTCS}
 \vspace{-3ex}
\end{table}
\end{center}

\head{Efficiency}
We evaluate the efficiency of different community search models on multilayer graphs. Figure~\ref{fig:Runing_time} shows the query processing time of all methods. All of our methods terminate within 1 hour, except Global on the two largest datasets, as it generates a large candidate graph $G_0$. Our algorithms Local and iLocal run much faster than Online-Global. Overall, iLocal achieves the best efficiency, and it can deal with a search query within a second on most datasets.  Local is the only algorithm that scales to graphs containing billions of edges. Bars for  iGlobal and iLocal are missing for Google+, as  index construction time exceeds our threshold.  

\begin{figure}
    \centering
    \includegraphics[width=0.49\textwidth]{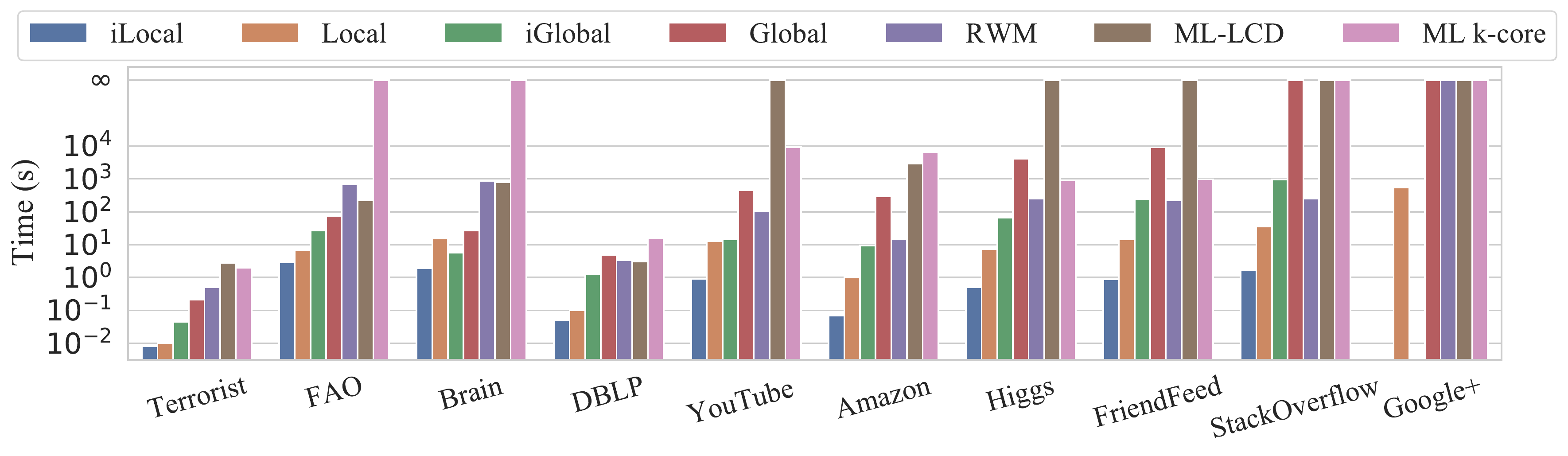}
    \vspace*{-5ex}
    \caption{Efficiency Evaluation.}
    \label{fig:Runing_time}
    \vspace{-2ex}
\end{figure}

\head{Parameter Sensitivity}
We evaluate the sensitivity of algorithm efficiency to the parameters $k, \lambda,$ and $p$, varying one parameter at a time. Figures~\ref{fig:param_sensitivity}(a) and (b) show the running time as a function  of $k$ and $\lambda$ on DBLP. The larger $k$ and $\lambda$ for Global and iGlobal  result in lower running time since the algorithms generate a smaller $G_0$. However, the larger $k$ and $\lambda$ increase the running time of Local and iLocal  since they need to count more nodes in the neighborhood of query nodes to find a $(k, \lambda)$-FirmTruss. \new{This also is the reason for the sharp increase of time in both plots. With large $k$ and $\lambda$, Local and iLocal need to count nodes farther away, and there is a significant increase in the number of nodes that they need to explore.} Figure~\ref{fig:param_sensitivity}(c) shows the running time as a function of $p$. We observe that AFTCS-Approx achieves a stable efficiency  on different finite values of $p$. Notice, when $p = -\infty$, this algorithm takes less time as it does not need to calculate $\Delta_u(S)$ for each node and can simply remove the vertex with minimum $h_S(u)$ in each iteration. 

\begin{figure}
    \centering
    \subfloat[\centering Effect of $k$ (DBLP)]{{\includegraphics[height=0.12\textwidth]{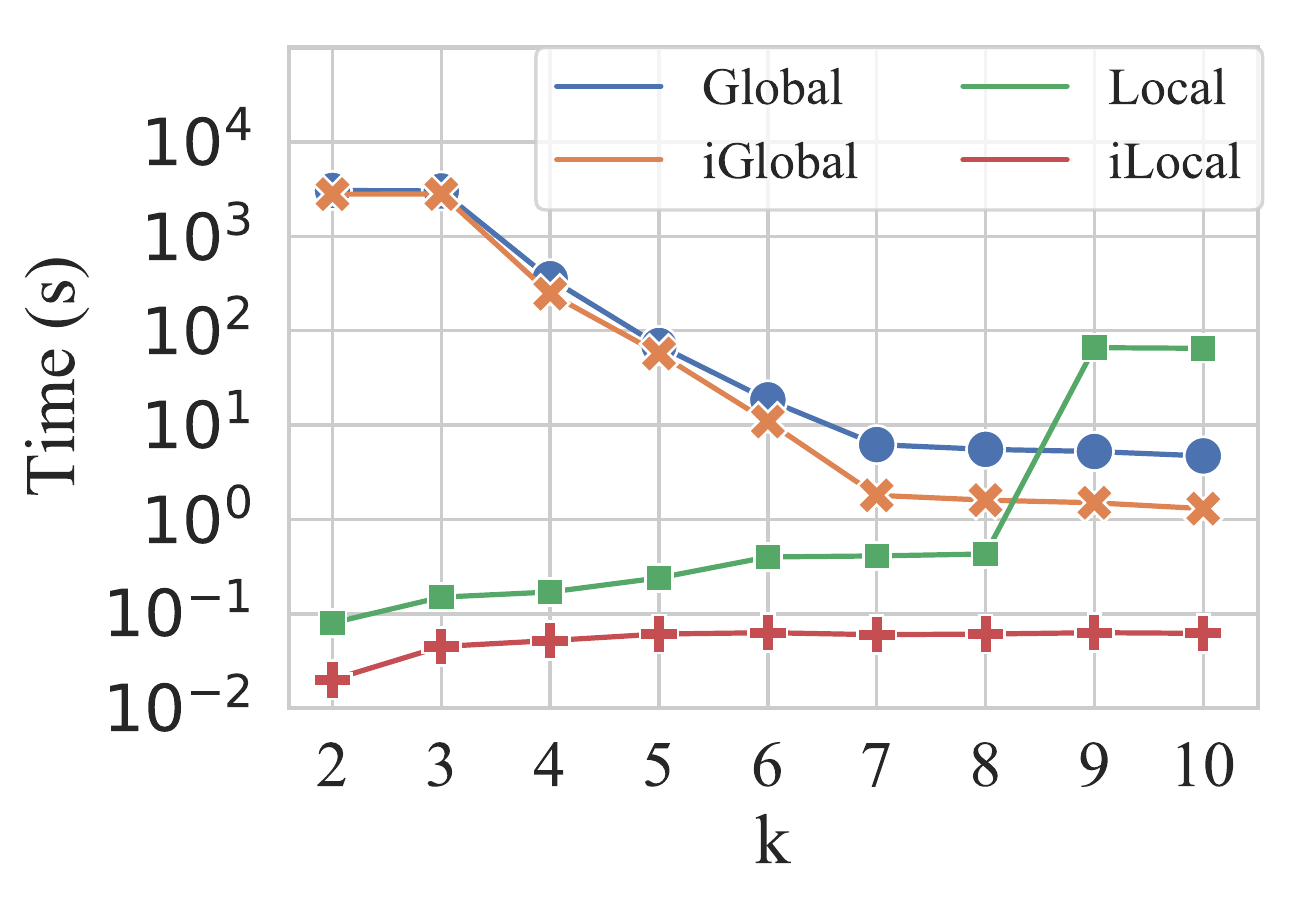} }}
    \subfloat[\centering Effect of $\lambda$ (DBLP)]{{\includegraphics[height=0.117\textwidth]{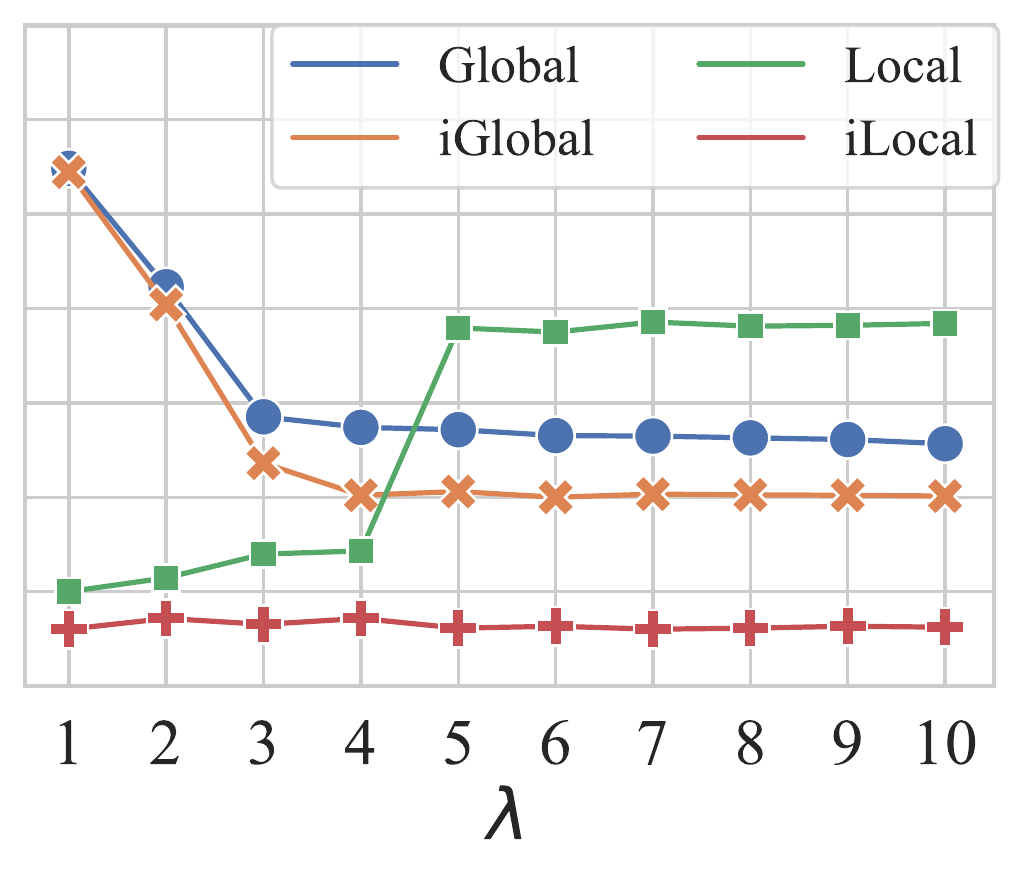} }}
    \subfloat[\centering Effect of $p$]{{\includegraphics[height=0.12\textwidth]{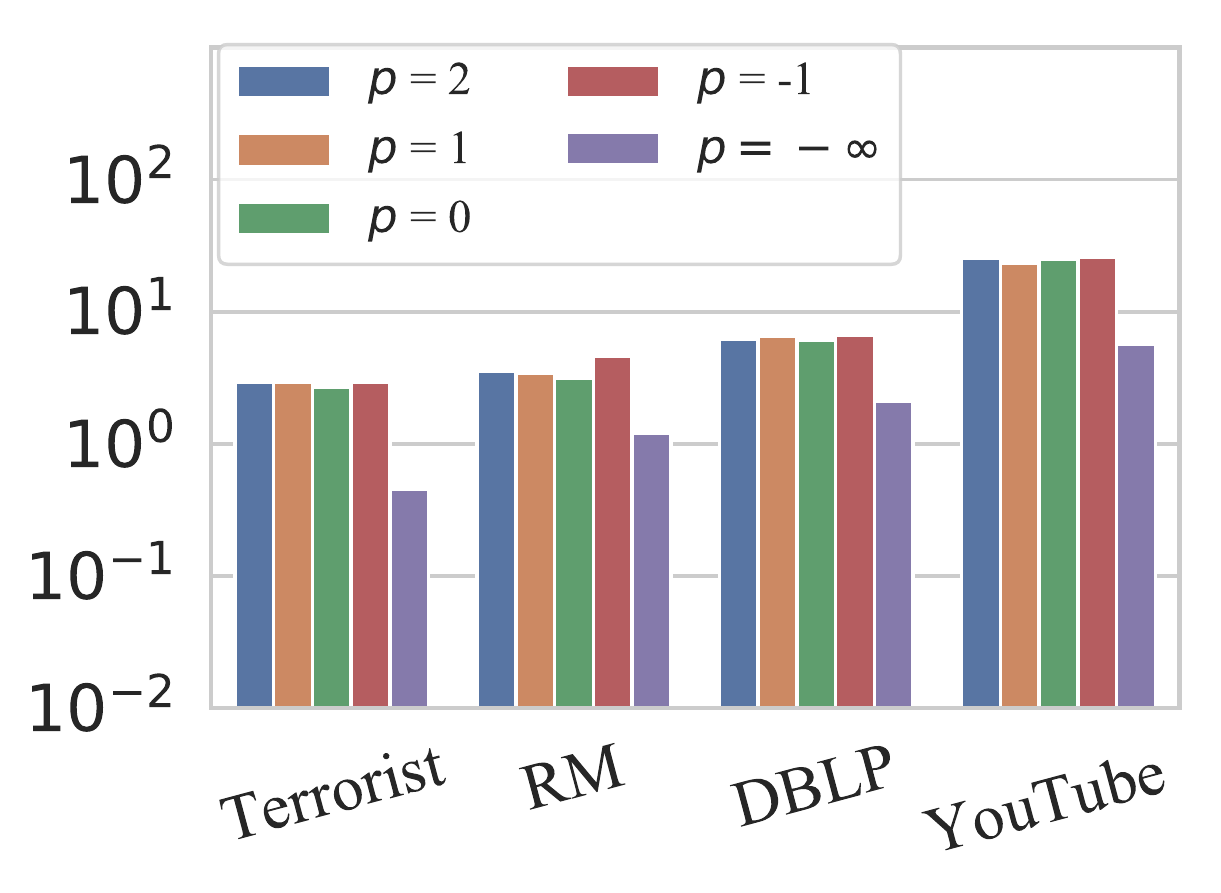} }}
    \caption{Parameter Sensitivity Evaluation.}
    \label{fig:param_sensitivity}
\end{figure}

\head{Scalability}
We test our algorithms using different versions of StackOverflow obtained by selecting a variable $\#$layers from 1 to 24 and also with different subsets of edges. Figure~\ref{fig:scalability} shows the results of the index-based Global, Local Search, and AFTCS-Approx algorithms. The results Global and iLocal are similar, and are omitted for lack of space (see Appendix~\ref{app:additional_experiments}). The running time of all approaches scales linearly in $\#$layers. By varying $\#$edges, all algorithms scale gracefully. As expected, the Local algorithm is less sensitive to varying $\#$edges than $\#$layers.

\begin{figure}[t]
    \centering
    \subfloat[\centering iGlobal]{{\centering \includegraphics[width=0.33\linewidth]{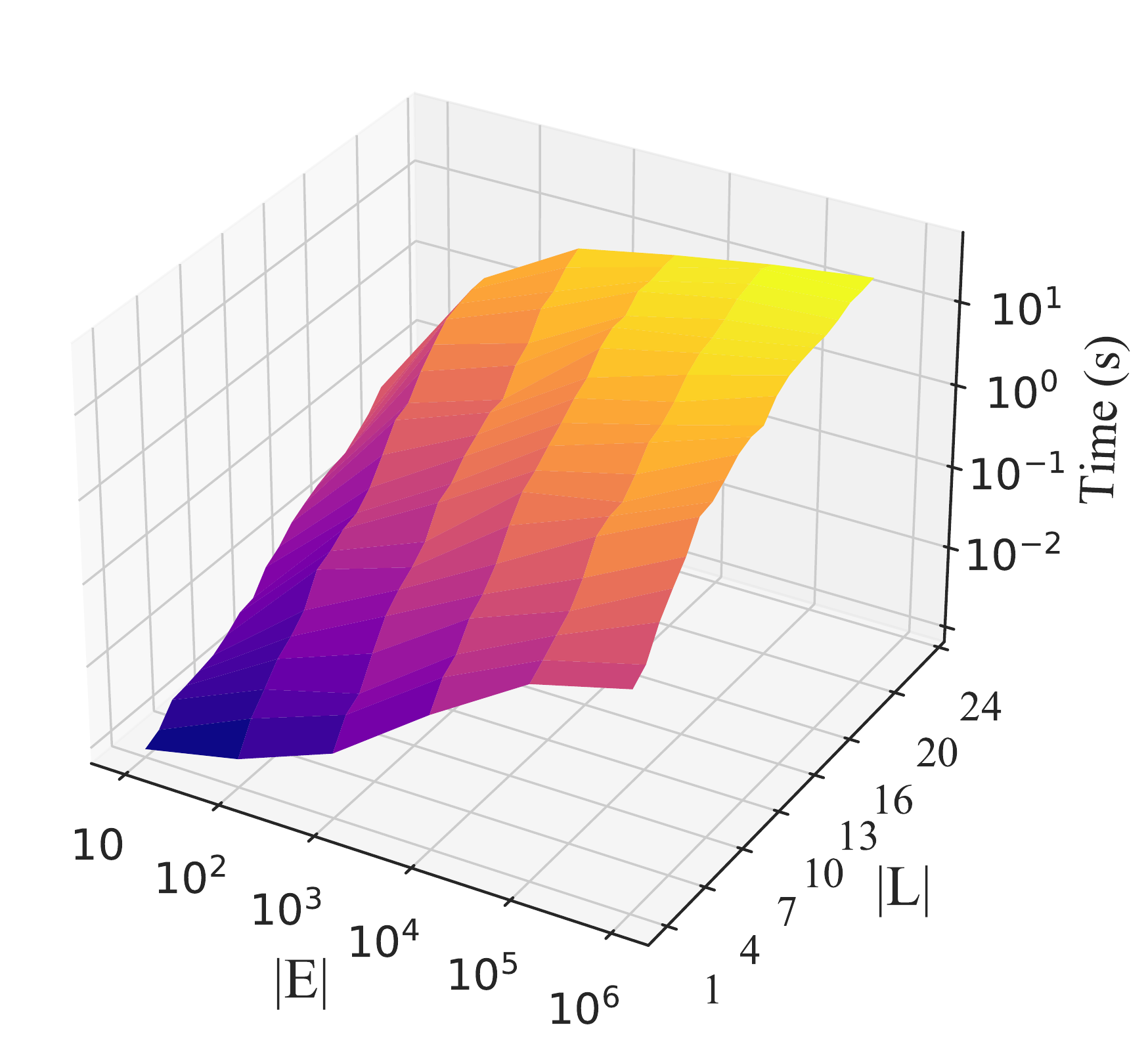}}}
    \subfloat[\centering Local Search]{{\centering \includegraphics[width=0.33\linewidth]{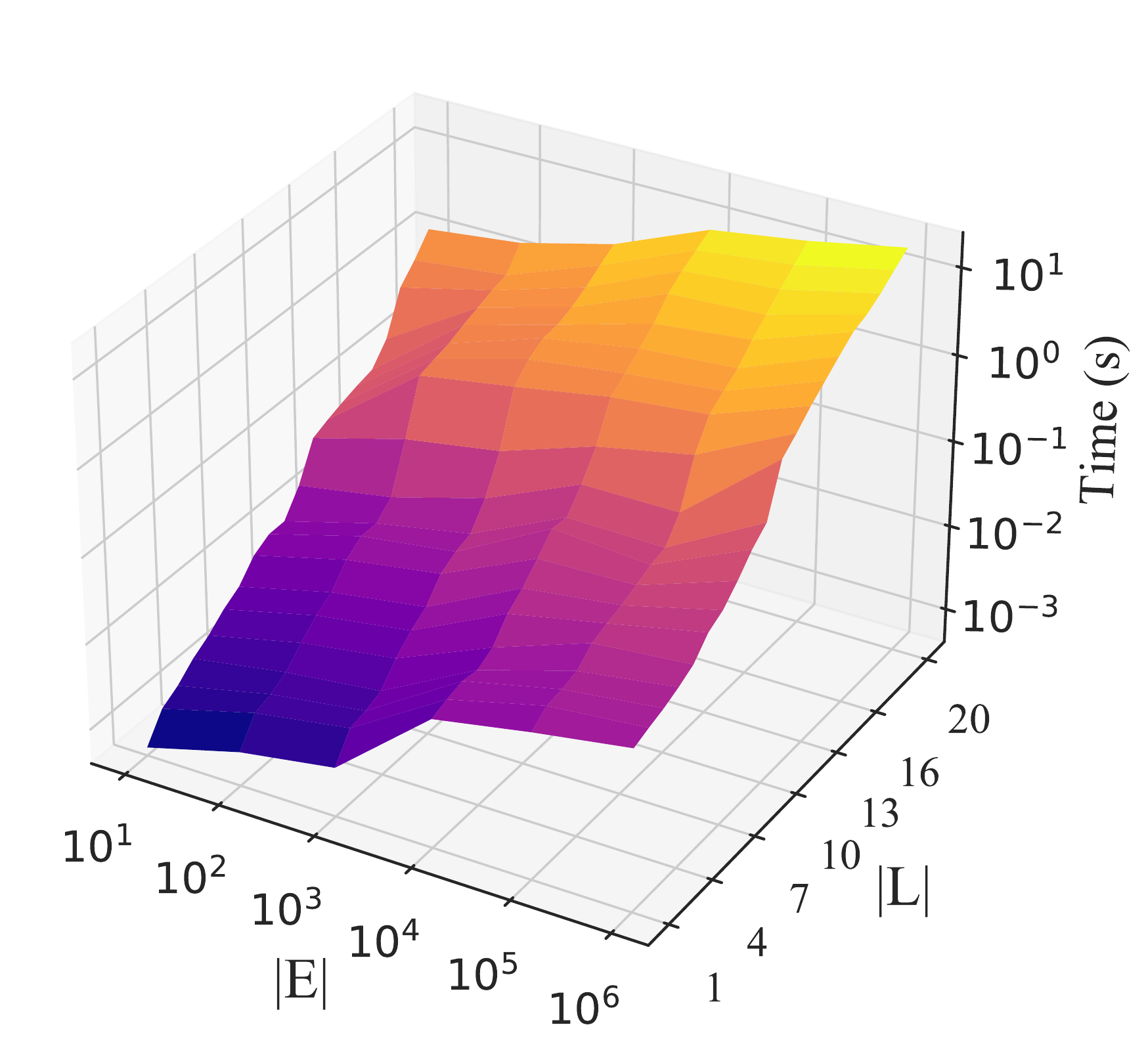}}}
     \subfloat[\centering AFTCS-Approx]{{\centering \includegraphics[width=0.33\linewidth]{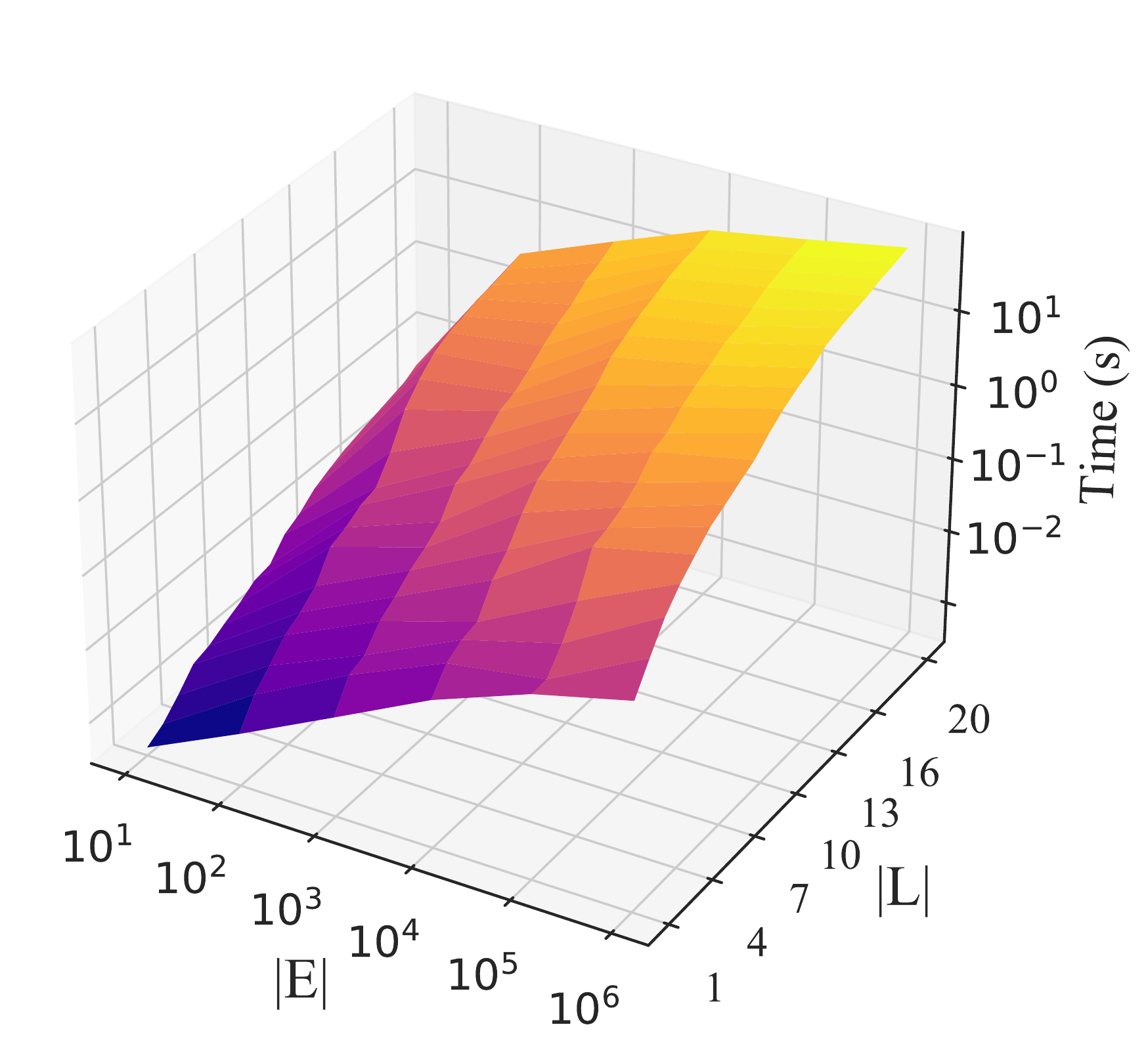}}}
    \caption{Scalability of proposed algorithms with varying the number of layers and the number of edges.}
    \label{fig:scalability}
\end{figure}

\head{Index Construction}
Figure~\ref{fig:index_construction} reports the $\mathsf{SFT}$ index construction time and size. \new{The size of indices is more dependent on the structure of a graph than its size. That is, since we store the SFT indices for each edge schema, the size of indices depends on the number of FirmTrusses in the network}. For all datasets, the $\mathsf{SFT}$ index can be built within 24 hours, and its size is within $2.6\times$ of the original graph size. The result shows the~efficiency~of~$\mathsf{SFT}$~index~construction. 

\begin{figure}
    \centering
    \includegraphics[width=0.47\textwidth]{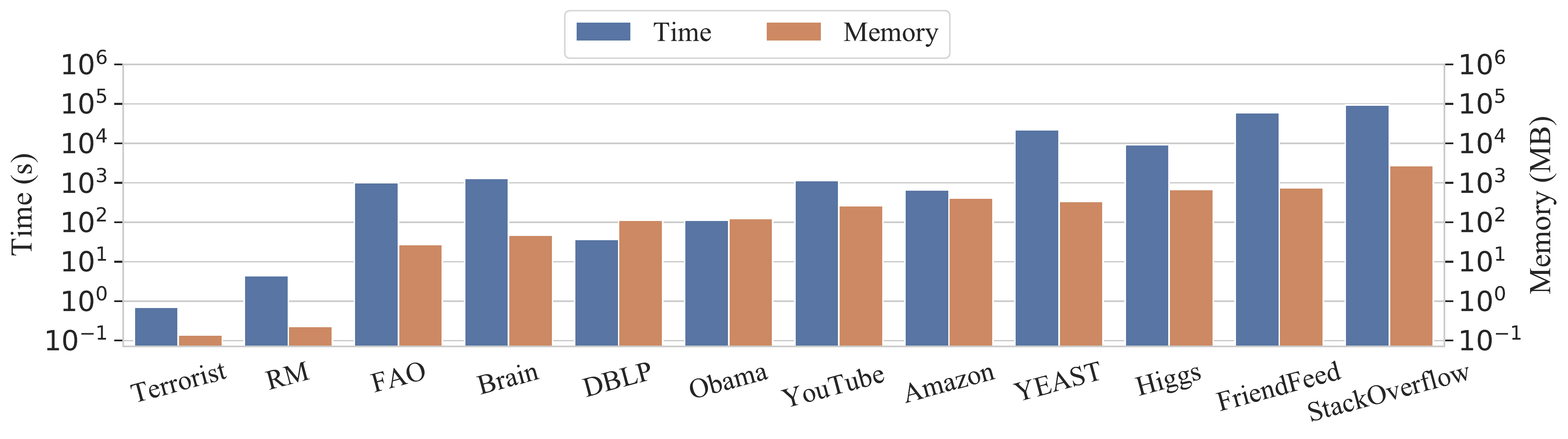}
    \vspace*{-3ex}
    \caption{Index Construction Costs.}
    \label{fig:index_construction}
\end{figure}

\head{Case Studies: Identify Functional Systems in Brain Networks}
Detecting and monitoring functional systems in the human brain is a primary task in neuroscience. However, the brain network generated from an individual can be noisy and incomplete. Using brain networks from many individuals can help to identify functional systems more accurately. A community in a multilayer brain network, where each layer is the brain network of an individual, can be interpreted as a functional system in the brain. In this case study, to show the effectiveness of the FTCS, we compare its detected functional system with ground truth. Here, we focus on the ``visual processing'' task in the brain. As the ``Occipital Pole'' is primarily responsible for visual processing~\cite{Occipital_lobe}, we use one of its representing nodes as the query node. Figure~\ref{fig:functional_systems} reports the found communities by FTCS and baselines. The identified communities are highlighted in red, and the query node is green. Results show the effectiveness of FTCS as the community detected by our method is very similar to the ground truth with F1-score of $0.75$. RWM, which is a random walk-based community model, includes many false-positive nodes that cause F1-score of $0.495$. On the other hand, some nodes in the boundary region are missed by ML-LCD that caused low F1-score of $0.4$. The result of the ML \textbf{k}-core is omitted as it does not terminate even before one week. 

\begin{figure}
\centering
    \hspace{-3ex}
    \subfloat[\centering Ground truth]{{\hspace{2ex}\includegraphics[width=0.11\textwidth]{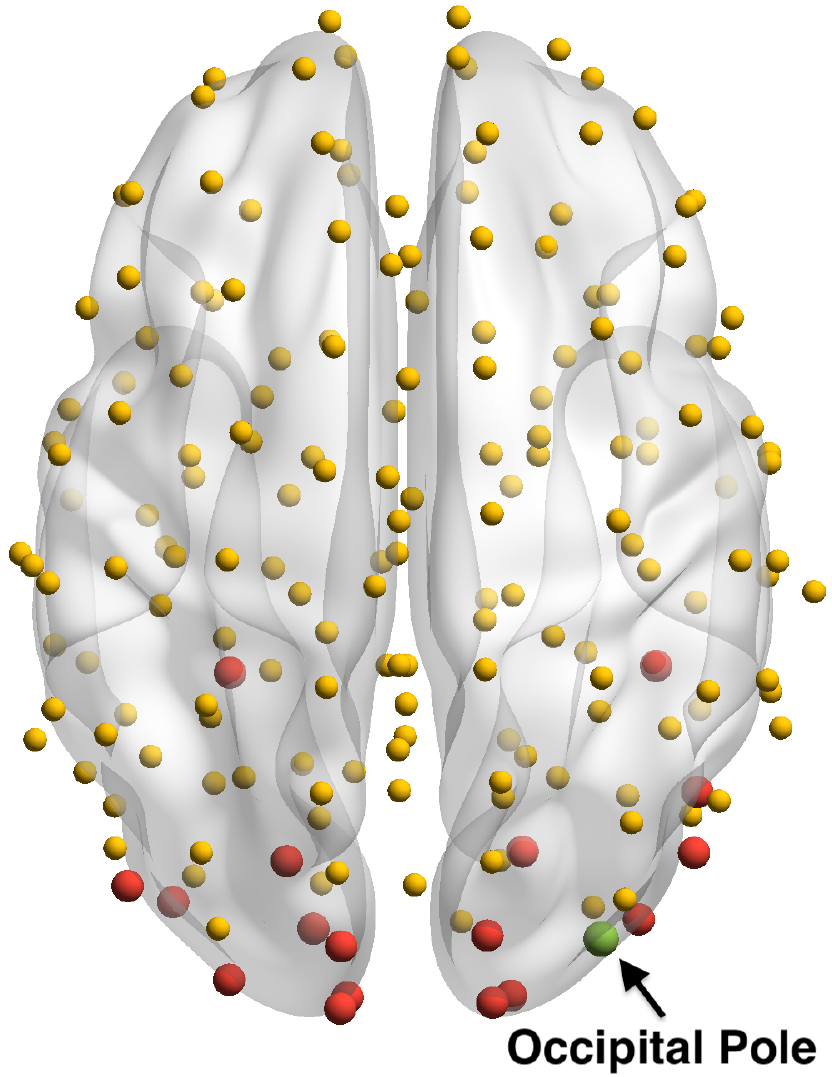} }}~
    \hspace{-2ex}
    \subfloat[\centering FirmTruss]{{\hspace{2ex}\includegraphics[width=0.11\textwidth]{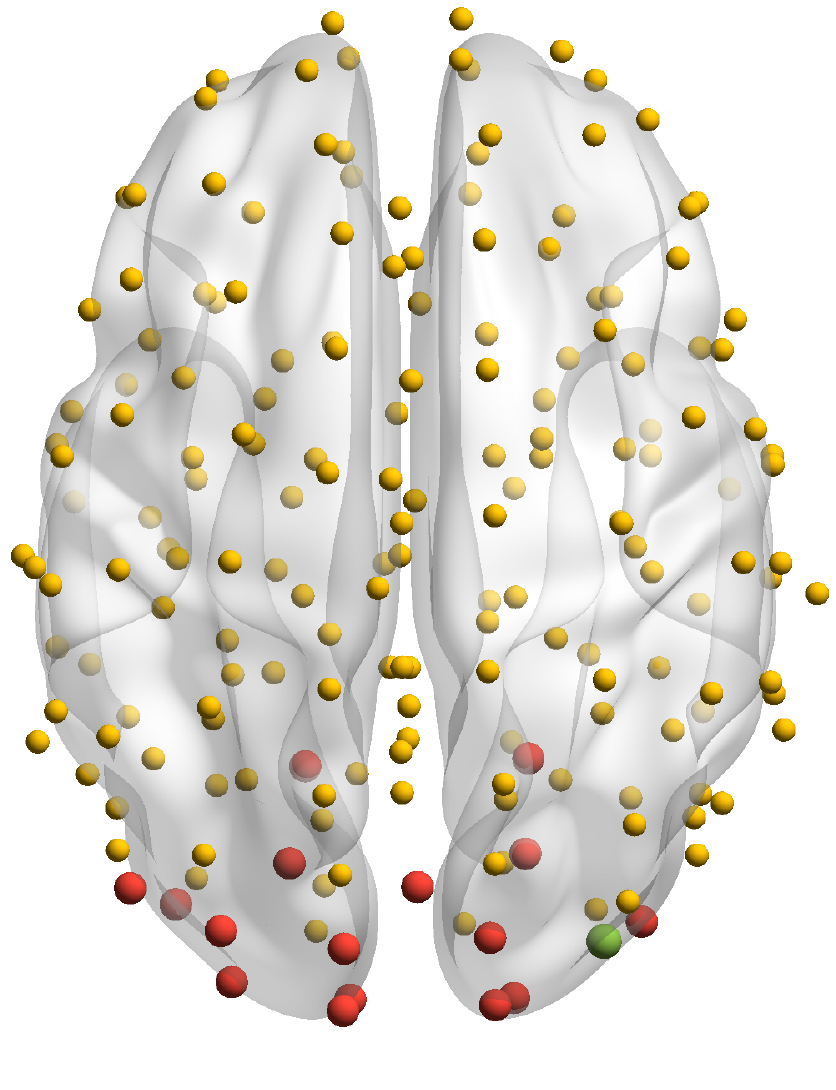} }}
    \hspace{-2ex}
    \subfloat[\centering ML-LCD]{{\hspace{2ex}\includegraphics[width=0.11\textwidth]{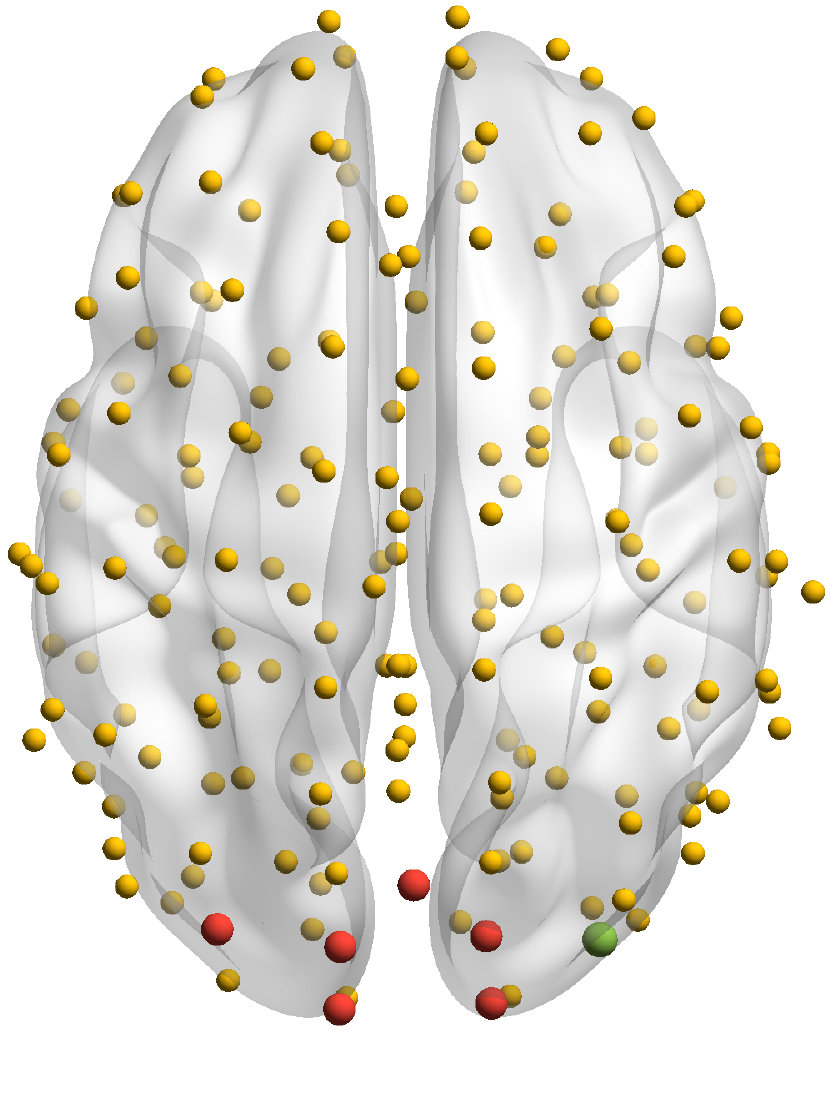} }}
    \hspace{-2ex}
    \subfloat[\centering RWM]{{\hspace{2ex}\includegraphics[width=0.11\textwidth]{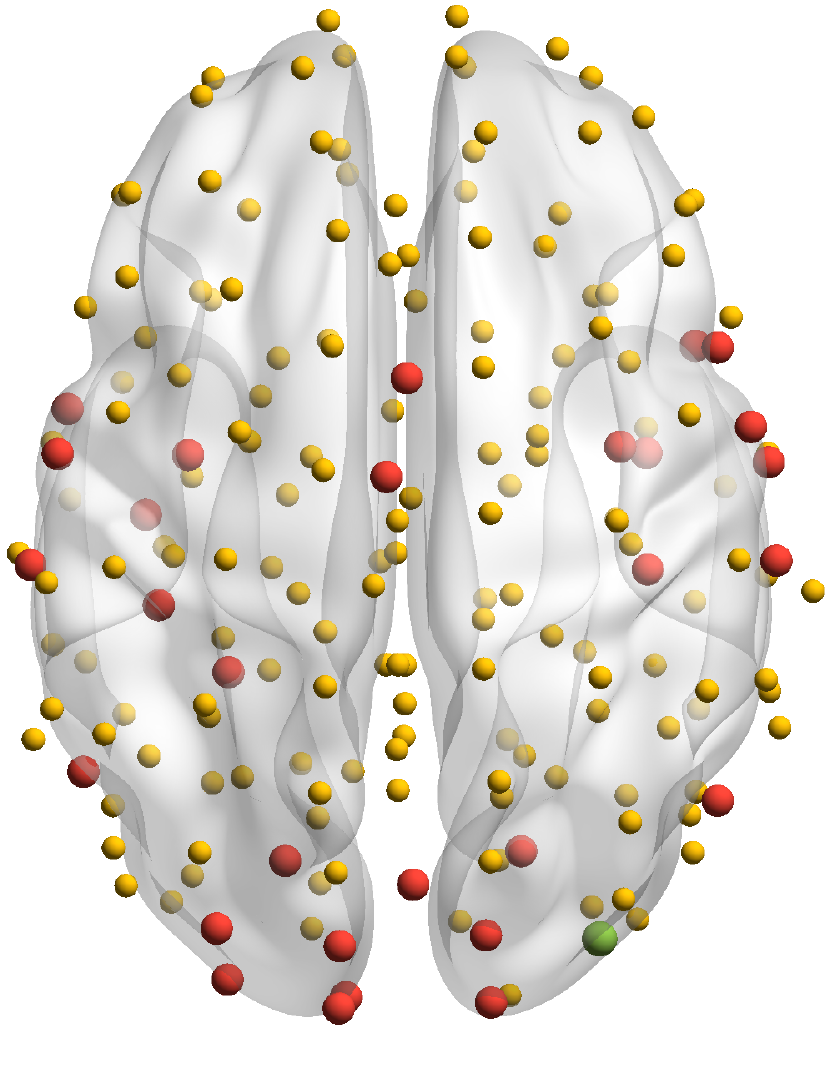} }}%
\caption{Detected functional systems.}
\label{fig:functional_systems}
\end{figure}

\begin{figure}
    \centering
    \subfloat[\centering $C_{\text{ADHD}}$]{{\hspace{6ex}\includegraphics[width=0.14\textwidth]{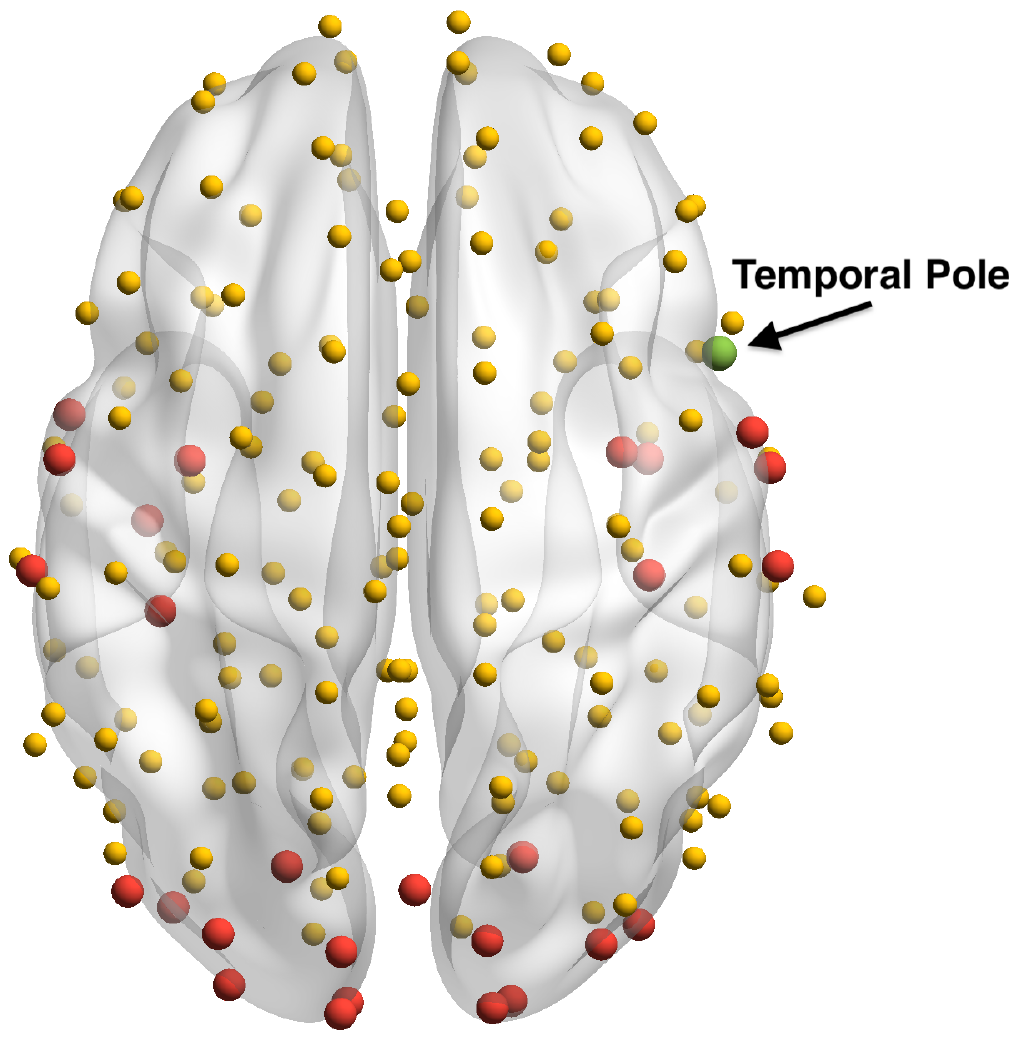} }}
     \hspace{3mm}
    \subfloat[\centering $C_{\text{TD}}$]{{\hspace{6ex}\includegraphics[width=0.14\textwidth]{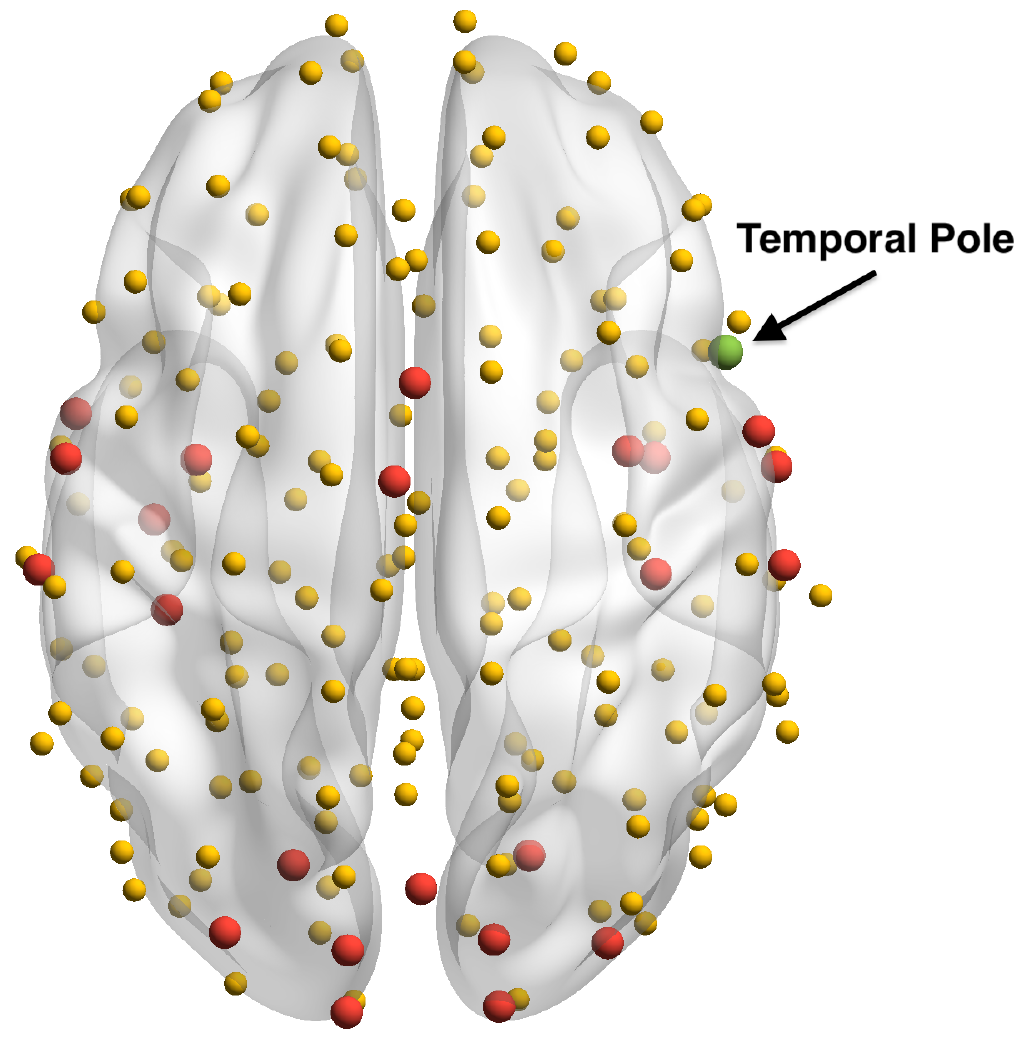} }}
    \caption{FirmTruss community in TD and ADHD groups.}
    \label{fig:adhd_prediction}
\end{figure}

\begin{center}
\begin{table} [tpb!]
    \caption{Results of the ADHD classification task.}
    \vspace{-2.5ex}
    \resizebox{0.4\textwidth}{!} {
\begin{tabular}{l | c |c |c | c}
 \toprule
  CS Model & Accuracy &  Precision & Recall & F1-score \\
 \midrule 
 \midrule
    FTCS                                & \textbf{76.56 $\pm$ 0.72}  & \textbf{75.73 $\pm$ 1.00} & \textbf{83.77 $\pm$ 1.21} &  \textbf{77.54 $\pm$ 0.66}   \\
    ML-LCD                              & 55.70 $\pm$ 1.25 & 55.43 $\pm$ 1.15 & 78.91 $\pm$ 1.64 & 64.13 $\pm$ 1.06 \\
    RWM                                 & 50.47 $\pm$ 0.18 & 53.03 $\pm$ 2.08& 55.09 $\pm$ 0.41 &   45.59 $\pm$ 1.18  \\
  \toprule
\end{tabular}
}
 \label{tab:ADHD_classification}
\end{table}
\end{center}

\head{Case Studies: Classification on Brain Networks}
Behavioral disturbances in attention deficit hyperactivity disorder (ADHD) are considered to be caused by the dysfunction of spatially distributed, interconnected neural systems~\cite{SVM_ADHD}. In this section, we employ our FTCS to detect common structures in the brain functional connectivity network of ADHD individuals and typically developed (TD) people. Our dataset is derived from the functional magnetic resonance imaging (fMRI) of 520 individuals with the same methodology used in~\cite{Brain_network_fmri}. It contains 190 individuals in the condition group, labeled  ADHD, and 330 individuals in the control group, labeled  TD. Here, each layer is the brain network of an individual person, where nodes are brain regions, and each edge measures the statistical association between the functionality of its endpoints. Since ``Temporal Pole'' is known as the part of the brain that plays an important role in ADHD symptoms~\cite{Brain_network_temporal_pole, Brain_network_temporal_pole2}, we use a subset of its representing nodes as the query nodes. 

Next, we randomly chose 230 individuals labeled  TD and 90 individuals labeled  ADHD to construct two multilayer brain networks and then found the  FirmTruss communities associated with ``Temporal Pole'' in each group separately, referred to as  $C_{\text{TD}}$ and $C_{\text{ADHD}}$ in Figure~\ref{fig:adhd_prediction}. In the second step, for each individual unseen brain network, we find the associated communities to the query nodes using the FTCS model, setting $|L| = 1$. In order to classify an unseen brain network, we calculate the similarity of its found communities with $C_{\text{TD}}$ and $C_{\text{ADHD}}$ and then predict its label as the label of the community with maximum similarity. Here, we use the overlap coefficient~\cite{overlap_coef} as the similarity measure~between~two~communities.

To ensure that the result is statistically significant, we repeat this process for 1000 trials and report the mean, and its relative standard deviation of accuracy, precision, recall and F1-score in Table~\ref{tab:ADHD_classification}. Not only does our FTCS outperform baseline community search  models, but it also achieves results comparable with the state-of-the-art ADHD classification model~\cite{SVM_ADHD},  based on SVM, which reports an accuracy of 76\%. This comparable result is achieved by the FTCS method, which is a white-box and explainable model.

\begin{figure}
    \centering
    \subfloat[\centering AFTCS Community]{{\includegraphics[width=0.32\textwidth]{botpng.pdf} }}
    \hspace*{-3ex}
    \subfloat[\centering Average attribute]{{\includegraphics[width=0.18\textwidth]{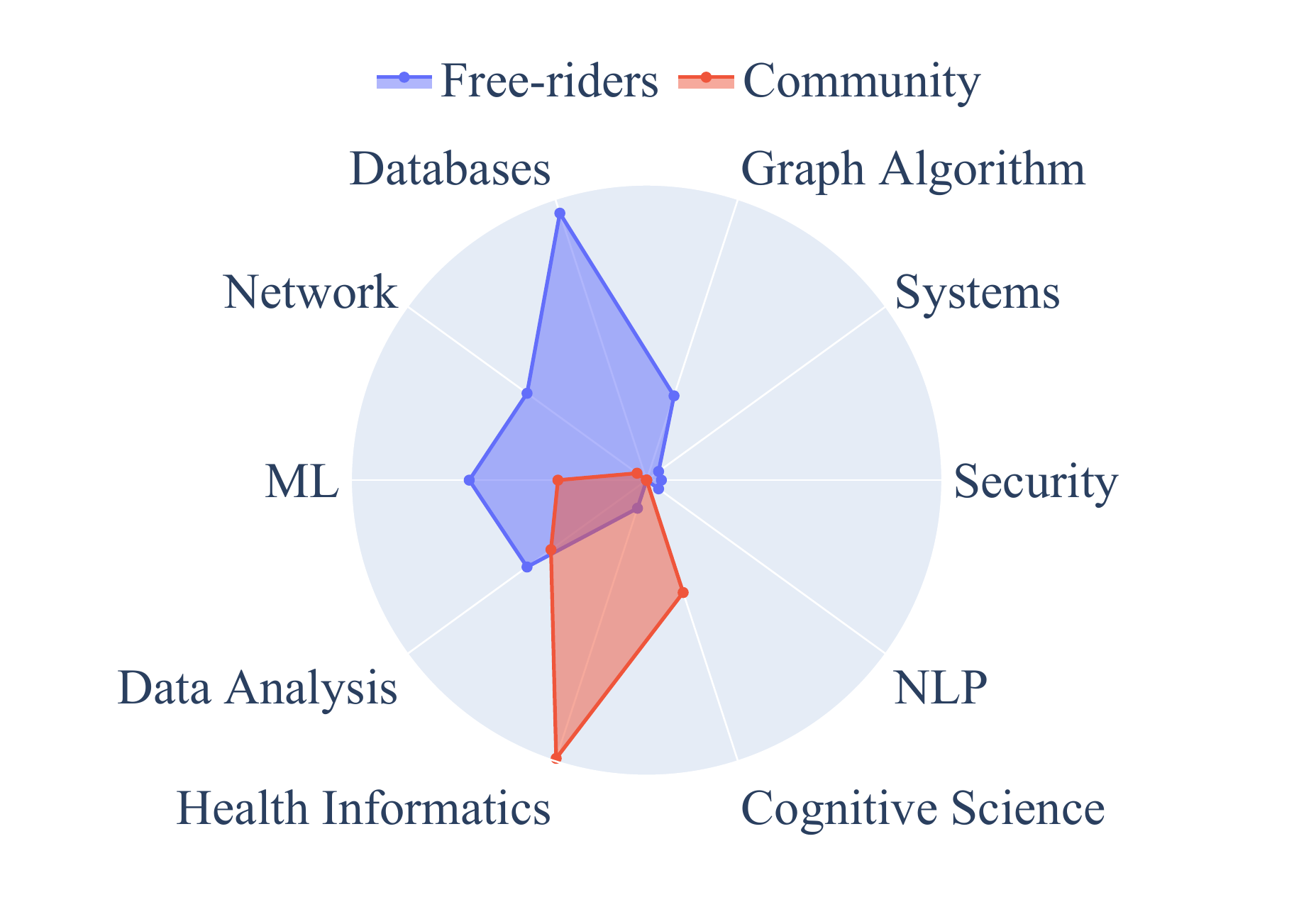} }}
    \caption{Case study of DBLP.}
    \label{fig:DBLP_case_study}
\end{figure}

\head{Case Studies: DBLP}
We conduct a case study on the DBLP dataset to judge the quality of the AFTCS model and to show the effectiveness of the homophily score in removing free riders. The multilayer DBLP dataset is a collaboration network derived following the methodology in~\cite{DBLP_methodology}. In this dataset, each node is a researcher, an edge shows collaboration, and each layer is a topic of research. For each author, we consider the bag of words drawn from the titles of all their papers and apply LDA topic modeling~\cite{LDA} to automatically identify 240 topics. The attribute of each author is the vector that describes the distribution of their papers in these 240 topics. We use "Brian D. Athey" as the query node.  The maximal $(8, 2)$-FirmTruss, including the query node, has 44 nodes with a minimum homophily score of $0.08$, shown in Figure~\ref{fig:DBLP_case_study}(a). \eat{While generally, researchers in this FirmTruss work on "Data," the aspect of their research is different.} The community found by AFTCS ($p = -\infty$) is an $(8, 2)$-FirmTruss with a minimum homophily score of $0.28$, which resulted from  removing 28 nodes as free-riders. The found community is shown in the larger circle, while the  smaller circle shows free riders. We compute the average attributes of community members and free-riders and then cluster their non-zero elements into ten known research topics. Results are shown in Figure~\ref{fig:DBLP_case_study}(b). While researchers in the found community have  focused more on "Health Informatics," removed researchers (free-riders) have  focused more on ``Databases.'' The connection between these two communities, which results in their union being an $(8, 2)$-FirmTruss, is the collaboration of ``Brian D. Athey,'' from the ``Health Informatics'' community with some researchers in ``Databases'' community. AFTCS divides the maximal FirmTruss into two communities with more correlations inside each of them.

\section{Conclusions}
\label{sec:conclusions}
We propose and study a novel extended notion of truss decomposition in ML networks, FirmTruss, and establish its nice properties. We then study a new problem of FirmTruss-based community search over ML graphs. We show that the problem is NP-hard. To tackle it efficiently, we propose  two 2-approximation algorithms and prove that our approximations are tight. To further improve their efficiency, we propose an index and develop fast index-based variants of our approximation algorithms. We extend the FirmTruss-based community model to attributed ML networks and propose a homophily-based model making use of generalized $p$-mean. We prove that this problem is also NP-hard for finite value of $p$ and to solve it efficiently, we develop a fast greedy algorithm which has a quality guarantee for $p \geq 1$. Our extensive experimental results on large real-world networks with ground-truth communities confirm the effectiveness and efficiency of our proposed models and algorithms, while  our case studies on brain networks and DBLP illustrate their practical utility.


\bibliography{main}
\bibliographystyle{ACM-Reference-Format}

\newpage
\appendix

\section{Proofs}

\subsection{Properties}

\head{Property~\ref{prop:Unique_FirmTruss} (Uniqueness)}
\begin{proof}
 Suppose $J^\lambda_{k}$ and $J^{'\lambda}_{k}$ are two distinct $(k, \lambda)$-FirmTrusses of $G$. By Definition~\ref{dfn:FirmTruss}, $J^\lambda_{k}$ is a maximal subgraph such that $ \forall \varphi \in \mathcal{E}[J^\lambda_{k}]$ there are $\geq \lambda$ layers where the supports of $\varphi$ in each of these layers is $\geq k - 2$. Similarly, $J^{'\lambda}_{k}$ is a maximal subgraph with the same property. Then the subgraph $J^\lambda_{k} \cup J^{'\lambda}_{k}$ trivially satisfies the FirmTruss conditions, contradicting~the~maximality~of~$J^\lambda_{k}$~and~$J^{'\lambda}_{k}$.
\end{proof}

\head{Property~\ref{prop:hierarchical} (Hierarchical Structure)}
\begin{proof}
(a) The property follows from the definition of FirmTruss and this fact that in a subgraph, every edge schema $\varphi \in \mathcal{E}$ that has at least $k - 1$ supports, also has at least $k - 2$ supports. (b) Similarly, if edge schema $\varphi$ in at least $\lambda + 1$ layers has no less than $k - 2$ supports, in at least $\lambda$ layers it has at least $k - 2$ supports as well.
\end{proof}

\head{Property~\ref{prop:truss-degree} (Minimum Degree)}
\begin{proof}
Given an edge schema $(v, u) \in \mathcal{E}$, where $u, v \in V$, since $(v, u)$ in at least $\lambda$ layers has at least $k - 2$ supports, $u$ and $v$ has at least $k - 2$ common neighbours in at least $\lambda$ layers. Therefore, each has degree at least $k - 1$ in at least $\lambda$ layers.
\end{proof}

\subsection{Theorems and Lemmas}\label{app:proof_theorems}

\head{Theorem~\ref{theorem:FirmTruss_density} (Density Lower Bound)}

\begin{proof}
First, we provide a lower bound for the mean of average degree over all layers: $\frac{\sum_{\ell = 1}^{|L|} |E_\ell[J^\lambda_k]|}{|L| |J^\lambda_k|}$. We can rewrite the numerator as half of the sum of degrees of each node in each layer. In fact, 
\begin{equation}\label{eq:density_lowerbound_1}
    \sum_{\ell = 1}^{|L|} |E_\ell[J^\lambda_k]| = \frac{1}{2}\left(\sum_{\ell = 1}^{|L|} \sum_{u \in J^\lambda_k}  \text{deg}_{\ell}^{J^\lambda_k}(u)\right).
\end{equation}
Since $J^\lambda_k$ is a FirmTruss, each edge schema $\varphi = (v, u)$ in at least $\lambda$ layers has at least $k - 2$ supports. Accordingly, $v$ and $u$ each in at least $\lambda$ layers have at least $k-1$ neighbours. Therefore, 
\begin{equation}\label{eq:density_lowerbound_2}
    \sum_{\ell = 1}^{|L|} \sum_{u \in J^\lambda_k}  \text{deg}_{\ell}^{J^\lambda_k}(u) \geq \lambda (k - 1) |J^\lambda_k|.
\end{equation}
Based on Equations (\ref{eq:density_lowerbound_1}) and (\ref{eq:density_lowerbound_2}), we can conclude that the mean of average degree over all layers is at least $\frac{\lambda (k - 1)}{2|L|}$. Since it is the average over all layers, there exist a layer $\ell_1 \in L$ such that:
\begin{equation*}
    \frac{|E_{\ell_1}[J^\lambda_k]|}{|J^\lambda_k|}  \geq \frac{\lambda (k - 1)}{2|L|}.
\end{equation*}
Now, ignoring this layer, and exploiting the definition of $J^\lambda_k$, each edge schema $\varphi = (v, u)$ in at least $\lambda - 1$ layers has at least $k - 2$ supports. Accordingly, $v$ and $u$ each in at least $\lambda - 1$ layers have at least $k-1$ neighbours. Therefore, the mean of average degree over remaining layers, $L \setminus \{\ell_1\}$, is at least $\frac{(\lambda - 1) (k - 1)}{2|L|}$. Again, we can conclude that there is a layer $\ell_2 \neq \ell_1$ such that $\frac{|E_{\ell_2}[J^\lambda_k]|}{|J^\lambda_k|}  \geq \frac{(\lambda - 1) (k - 1)}{2|L|}$. So there is a subset of $L$, $\hat{L} = \{\ell_1, \ell_2\}$, with two layers such that each layer has density at least $\frac{(\lambda - 1) (k - 1)}{2|L|}$. By iterating this process, we can conclude that $\exists \Tilde{L} \subseteq L$, such that each layer $\ell \in \Tilde{L}$ has density at least $\frac{(\lambda - |\Tilde{L}| + 1) (k - 1)}{2|L|}$. Based on the definition of the density in ML graphs, we can conclude that:
\begin{equation*}
    \rho_\beta(J^\lambda_k) \geq \frac{(\lambda - |\Tilde{L}| + 1)(k - 1)}{2|L|} |\Tilde{L}|^\beta.
\end{equation*}
The right hand side of the above inequality is a function of $|\Tilde{L}|$, and since the above inequality is valid for each arbitrary integer $1 \leq |\Tilde{L}| \leq \lambda$, $\rho_\beta(J^\lambda_k)$ is not less than the maximum of this function exhibited by these values.
\end{proof}

\head{Theorem~\ref{theorem:FirmTruss_diam} (Diameter Upper Bound)}

\begin{proof}
Let path 
$$\mathcal{P} : v^{1}_{\ell_1} \rightarrow v^{2}_{\ell_2} \rightarrow \dots \rightarrow v^{\alpha}_{\ell_\alpha},$$ 
be the diameter of $G[J^\lambda_{k}]$, and for each edge schema $\varphi \in \mathcal{E}[J^\lambda_{k}]$, we define $L_{\varphi} = \{\ell'_1, \dots, \ell'_{\lambda'} \}$ as the set of layers such that in each layer $\ell'_i$, $\varphi_{\ell'_i}$ has at least $k-2$ supports. Note that $\lambda' \geq \lambda$. Let $t \in \mathbb{N}$ such that $\frac{t}{t+1} |L| > \lambda \geq \frac{t-1}{t} |L|$, then we show that there is a path 
$$\mathcal{P}': w^{1}_{\ell''_1} \rightarrow w^{2}_{\ell''_2} \rightarrow \dots \rightarrow w^{\alpha'}_{\ell''_{\alpha'}},$$
such that its path schema is the same as the path schema of $\mathcal{P}$, and for each $0 \leq r \leq \floor{\frac{|\mathcal{P}'|}{t}} - 1$ all $\{w^{rt+1}, w^{rt+2}, \dots, w^{rt + t}\}$ are in the same layer $\ell_r^*$ such that: 
$$\ell_r^* \in \bigcap_{i = 1}^{t - 1} L_{(v^{rt+i}, v^{rt+i+1})}.$$ 
To show that, let $\mathfrak{P}$ be the path schema of $\mathcal{P}$, and $u^1, u^2, \dots, u^t$ be the $t$ consecutive vertices in the path schema. For each edge schema $\varphi^i = (u^i, u^{i+1})$, we know that there are at least $\lambda$ layers $\ell \in L_{\varphi^i}$ in which $\varphi^{i}_{\ell}$ has at least $k-2$ support. Thus, since $\lambda \geq \frac{t-1}{t} |L|$, based on the pigeonhole principle there is a layer $\ell^*$ such that all $\varphi^i$ in $\ell^*$ has at least $k-2$ supports. Therefore, we have constructed path $\mathcal{P}'$ such that for each $0 \leq r \leq \floor{\frac{|\mathfrak{P}|+1}{t}} - 1$:
\begin{enumerate}
    \item vertices $\{u^{rt+1}, u^{rt+2}, \dots, u^{rt + t}\}$ of the path are in the same layer $\ell_r^*$,
    \item we have an inter-layer edge in the path that $u^{rt+t}_{\ell^*_r} \rightarrow u^{rt+t}_{\ell^*_{r+1}}$.
\end{enumerate}

Due to the number of inter-layer edges, which is at most $\frac{|\mathfrak{P}|}{t}$, the length of the $\mathcal{P}'$ is at most $(1+1/t)|\mathfrak{P}|$.

Next we show that the length of $\mathfrak{P}$ is bounded by $\floor{\frac{2|J_k^\lambda| - 2}{k}}$. Given the part of $\mathfrak{P}$, $\mathfrak{P}_r: u^{rt + 1} \rightarrow u^{rt + 2} \rightarrow \dots \rightarrow u^{rt + t}$, that all edges are in layer $\ell^*_r$, for any $0 \leq r \leq \floor{\frac{|\mathfrak{P}| + 1}{t}}$. We know that each edge $(u^{rt + i}, u^{rt + i+1}, \ell_r^*)$ has at least $k-2$ supports. Thus, each pair $(u^{rt + i}, u^{rt + i+1})$ has at least $k-2$ common neighbours in layer $\ell^*_r$. We divide the set of common neighbours into two sets $A^r_i$ and $B^r_i$: each $A^r_i$ holds those vertices not in $\mathfrak{P}_r$ that are adjacent to both $u^{rt + i}$ and $u^{rt + i+1}$, and to no other vertices in $\mathcal{P}'_r$; each $B^r_i$ holds those vertices that are adjacent to all three of $u^{rt + i-1}, u^{rt + i}$ and $u^{rt + i+1}$. Note that the $B^r_i$ sets are disjoint, else we can find a path with smallest length. Let
\begin{equation*}
    |N_{\ell^*_r}| = |\mathfrak{P}_r| + 1 + \sum_{i=0}^{|\mathfrak{P}_r| - 1} |A^r_i| + \sum_{i=1}^{|\mathfrak{P}_r| - 1} |B^r_i|.
\end{equation*}
  For each $i$, $|A^r_i| = k - 2 - |B^r_i|$ and also $|B^r_i| + |B^r_{i+1}|\leq k-2$. Thus:
\begin{align*}
    & |N_{\ell^*_r}| \geq 1 + |\mathfrak{P}_r|(k-1) - \floor{\frac{|\mathfrak{P}_r|}{2}} (k-2)\\
    \Rightarrow & \sum_{r} |N_{\ell^*_r}| \geq \floor{\frac{|\mathfrak{P}| + 1}{t}} + (k-1)\sum_{r}|\mathfrak{P}_r| - (k-2)\sum_{r}\floor{\frac{|\mathfrak{P}_r|}{2}}\\
    &\geq \floor{\frac{|\mathfrak{P}| + 1}{t}} + |\mathfrak{P}|(k-1) - \floor{\frac{|\mathfrak{P}|}{2}} (k-2)
\end{align*}
Note that each node cannot be counted in two $N_{\ell^*_r}$, else we can find a path with smallest length. Therefore, 
\begin{align*}
    & |J^\lambda_k| \geq \sum_{r} |N_{\ell^*_r}| \geq \floor{\frac{|\mathfrak{P}| + 1}{t}} + |\mathfrak{P}|(k-1) - \floor{\frac{|\mathfrak{P}|}{2}} (k-2)\\
    \Rightarrow & \: |\mathfrak{P}| \leq \floor{\frac{2|J^\lambda_k| - 2}{k}}
\end{align*}
Overall, we have:
\begin{equation*}
    (1+1/t) \floor{\frac{2|J^\lambda_k| - 2}{k}} \geq (1+1/t) |\mathfrak{P}| \geq |\mathcal{P}'| \geq |\mathcal{P}|.
\end{equation*}
The last (right) inequality comes from the fact that the start and end of paths $\mathcal{P}$ and $\mathcal{P}'$ are the same. Accordingly, if the length of $\mathcal{P}'$ is less than $\mathcal{P}$, then $\mathcal{P}'$ is the shortest path between the start and end nodes, which contradicts the assumption that $\mathcal{P}$ is the diameter. 
\end{proof}

\head{Theorem~\ref{theorem:FirmTruss_connectivity} (Edge Connectivity)}

\begin{proof}
Let $C$ be a minimal set of intra-layer edges whose removal will divide the $G[J^\lambda_{k}]$ into separate components, and $e = (u, v, \ell) \in C$. When all of the edges in $C$ are removed, the vertices $v$ and $u$ will be in separate components, else the removal of $e$ was unnecessary, violating the minimality of $C$. Since $(v, u) \in \mathcal{E}(J^\lambda_{k})$ there are at least $\lambda$ layers $\hat{L} = \{\ell_1, \dots, \ell_\lambda \}$ in which $v$ and $u$ are joined by at least $k - 2$ independent paths of length 2. Therefore, in each layer $\ell_i \in \hat{L}$, at least $k - 1$ edges must be removed to place $v$ and $u$ in separate components. Therefore, $|C| \geq \lambda(k-1)$.
\end{proof}

\head{\new{Theorem~\ref{theorem:FTCS-apx} (FTCS Hardness and Non-Approximability)}}

\begin{proof} 
\eat{The NP-hardness of FTCS problem is concluded by its non-approximibility. That is, while there is no polynomial algorithm that can find an approximation solution with quality better than 2, then it also shows that there is no polynomial-time algorithm that can exactly solve the problem.} \new{We directly prove non-approximibility. The proof of NP-hardness follows a similar construction. } 

\new{Assume  there exists a polynomial time $(2-\epsilon)$-approximation algorithm for the FTCS problem, for any $\epsilon \in (0, 1)$. Let  $\mathcal{A}$ be an algorithm that provides a solution $H$ with an approximation factor $(2 - \epsilon)$ of the optimal solution $H^*$. Let the set of query nodes be $Q = \emptyset$. Let $G = (V, E)$ be an instance of Maximum Clique Decision problem. Construct $G' = (V', E', L)$, where $V' = V$, $|L| = \lambda$, and for each $\ell \in L : E'_\ell = E$. \eat{We claim that $G$ contains a clique of size $k$ iff $\mathcal{A}$ outputs a solution $H$ that is $(k, \lambda)$-FirmTruss with diameter 1.} 
} 

\new{
$(\Rightarrow) \: :$ Clearly, if $\mathcal{A}$ outputs a solution $G'[H]$ that is a $(k, \lambda)$-FirmTruss with diameter 1, then $G_\ell[H]$ is clearly a clique of size at least $k$ for each $\ell \in L$, which implies $G$ contains a clique of size $\geq k$. } 

\new{
$(\Leftarrow) \: :$ Suppose $diam(G'[H]) \geq 2$, so:
\begin{equation*}
    2 diam(G'[H^*]) > (2 - \epsilon) diam(G'[H^*]) \geq diam(G'[H]) \geq 2. 
\end{equation*}
Since diameter is an integer, we deduce that $diam(G'[H^*]) \geq 2$. In this case, $G'_\ell[H^*]$ cannot possibly contain a clique of size $k$ in any layer, for if it did, that clique would be the optimal solution to the FTCS problem with diameter  1. Therefore, since $G'_\ell[H^*]=G[H^*]$, we can distinguish between the YES and NO instances of the Maximum Clique Decision problem in polynomial-time, which is a contradiction. } 

\new{To show the NP-hardness of $d$-FTCS problem, we reduce the Maximum Clique (decision version) to $d$-FTCS problem. By setting $d = 1$, the instance construction and proof is similar to the proof of non-approximability.
}
\end{proof}

\eat{ 
\head{Proof for FTCS Hardness}
\begin{proof}
We reduce the well-known NP-hard problem of Maximum Clique (decision version) to $d$-FTCS problem. Given a single-layer graph $G = (V, E)$ and integer $k$, the Maximum Clique Decision problem is to check whether $G$ contains a clique of size $k$. From this, construct an instance $G' = (V', E', L)$, where $V' = V$, $|L|$ is an arbitrary number $\geq \lambda$, $E'_\ell = E$ for each $\ell \in \hat{L} = \{\ell_1, \dots, \ell_\lambda \}$ and $E'_\ell = \emptyset$ for $\ell \in L/\hat{L}$, parameters $k$ and $d = 1$, and the empty set of query nodes $Q = \emptyset$. We show that the instance of the Maximum Clique (decision version) problem is a YES-instance iff the corresponding instance of $d$-FTCS problem is a YES-instance.

$(\Rightarrow) \: :$ A set of vertices $H \subseteq V$ that in any layer $\ell \in \hat{L}$ is a clique with at least $k$ vertices is a connected $(k, \lambda)$-FirmTruss with diameter 1 since each two node have at least $k - 2$ common neighbours in $\lambda$ layers. 

$(\Leftarrow) \: :$ Given a solution $H$ for $d$-FTCS problem, so each edge schema in at least $\lambda$ layers has at least $k - 2$ supports, but since for $\ell \in L/\hat{L}$ the set of edges in layer $\ell$ is empty, all edge schema each in layers $\ell \in \hat{L}$ has supports of at least $k - 2$. Given a layer $\ell \in \hat{L}$, the distance of two arbitrary nodes $v \neq u \in H$, in $G'_\ell[H]$ is 1 since if their shortest path contains an inter-layer edge, then its length is at least 2. Therefore, $G'_\ell[H]$ contains a clique of size at least $k$.
\end{proof}
} 

\head{Theorem~\ref{theorem:FTCS-free-rider} (FTCS Free-Rider Effect)}

\begin{proof}
Let $G[H]$ be the maximal subgraph among all optimal solutions to the FTCS problem, and $G[H^*]$ be any query-independent optimal solution. We prove this theorem by contradiction. If $G[H^*] \:\char`\\ \:G[H]\ \neq \emptyset$, $G[H \cup H^*]$ is a connected $(k, \lambda)$-FirmTruss containing $Q$, and $diam(G[H \cup H^*]) \leq diam(G[H])$, then $G[H \cup H^*]$ is an optimal solution to the FTCS problem while $G[H] \subseteq G[H \cup H^*]$, violating the maximality of $G[H]$.
\end{proof}

\head{Theorem~\ref{theorem:FCCS-Global-approx-quality} (FTCS-Global Quality Approximation)}

\begin{proof}
Motivate by~\cite{closest}, we first prove: 
\begin{equation*}
dist_{G[H]}(H, Q) \leq dist_{G[H^*]}(H^*, Q).
\end{equation*}
Algorithm~\ref{alg:binary_search} outputs a sequence of intermediate graphs (in line 12) $\hat{G} = \{G_0, G_1, \dots, G_{i - 1}\}$, which are $(k, \lambda)$-FirmTruss containing query vertices $Q$. Therefore, we have $G_{i - 1} \subseteq \dots \subseteq G_{1} \subseteq G_0 \subseteq G$ and also since the optimal solution is also a $(k, \lambda)$-FirmTruss, $G[H^*] \subseteq G_0$. We consider two cases: 

\textbf{(I)} $G[H^*] \nsubseteq G_{i - 1}$. Suppose the first deleted vertex $u^* \in H^*$ happens in graph $G_s$, this vertex must be deleted because of the distance constraint but not the FirmTruss maintenance. Accordingly, $dist_{G_s}(G_s, Q) = dist_{G_s}(u^*, Q) = dist_{G_s}(H^*, Q)$. Since $H$ has the smallest query distance in $\hat{G}$ and $G[H^*] \subseteq G_s$, we have:
\begin{equation*}
    dist_{G[H]}(H, Q) \leq dist_{G_s}(G_s, Q) = dist_{G_s}(H^*, Q) \leq dist_{G[H^*]}(H^*, Q)
\end{equation*}

\textbf{(II)} $G[H^*] \subseteq G_{i - 1}$. We prove this case by contradiction. Assume that $dist_{G_{i - 1}}(G_{i - 1}, Q) > dist_{G_{i - 1}}(H^*, Q)$, then there exist a vertex $u^* \in G_{i - 1}\char`\\H^*$ with the largest query distance $dist_{G_{i-1}}(u^*, Q) = dist_{G_{i-1}}(G_{i-1}, Q) > dist_{G_{i - 1}}(H^*, Q)$. In the next iteration, Algorithm~\ref{alg:binary_search} will delete $u^*$ from $G_{i-1}$, and maintain the FirmTruss structure of $G_{i-1}$. Since $G[H^*] \subseteq G_{i - 1}$, the output of this iteration is a feasible FirmTruss. This contradicts that $G_{i-1}$ is the last feasible FirmTruss. Therefore, $dist_{G_{i-1}}(G_{i-1}, Q)\leq dist_{G_{i-1}}(H^*, Q)$, so we have:
\begin{align*}
     & \:\: \: dist_{G[H]}(H, Q) \leq dist_{G_{i-1}}(G_{i-1}, Q)\leq dist_{G_{i-1}}(H^*, Q)\\
     & \leq dist_{G[H^*]}(H^*, Q).
\end{align*}
From above we have proved that $dist_{G[H]}(H, Q) \leq dist_{G[H^*]}(H^*, Q)$. Not that based on the triangle inequality  
\begin{equation*}
    diam(G[H]) \leq 2dist_{G[H]}(H, Q).
\end{equation*}
Thus, we can conclude that:
\begin{align*}
    diam(G[H]) &\leq 2dist_{G[H]}(H, Q) \leq 2 dist_{G[H^*]}(H^*, Q) \\
    & \leq 2diam(G[H^*]).
\end{align*}
\end{proof}

\vspace{-2ex}
\head{Lemma~\ref{lemma:time_complexity_FirmTruss} (Time Complexity of Maximal $(k, \lambda)$-FirmTruss)}

\begin{proof}
Using a heap, we can find the $\lambda$-th largest element of vector $\mathbf{S}$ in $\mathcal{O}(|L|+\lambda \log|L|)$ time so lines 4-7 take \\$\mathcal{O}(|E|( |L| + \lambda\log |L|))$ time. We remove each edge schema exactly one time from the buckets, taking time $\mathcal{O}(|E|)$. For $\ell \in L$, let $nb^\ell_{\geq}(u) = \{ v | (v, u, \ell) \in E, deg_\ell(v) \geq deg_\ell(u) \}$, lines 10-16 take $\mathcal{O}(\sum_{(v, u) \in \mathcal{E}} \sum_{\ell \in L} deg_\ell(u) |nb^\ell_{\geq}(u) |)$. For each layer $\ell$, we know that $|nb^\ell_{\geq}(u)| \leq 2\sqrt{|E_\ell|}$~\cite{truss-algorithm}, so we can conclude that lines 10-16 take $\mathcal{O}(\sum_{\ell \in L} |E_\ell|^{1.5})$. Finally, to update the Top-$\lambda$ support    for a given edge schema we need $\mathcal{O}(|L|)$, so lines 17-20 takes~$\mathcal{O}(|E| |L|)$ time. In this algorithm, we only need to store and keep the support vector of each edge schema so it takes $\mathcal{O}(|E||L|)$ space.
\end{proof}

\head{Theorem~\ref{theorem:FTCS-Global-Complexity} (FTCS-Global Complexity)}
\begin{proof}
First, Algorithm~\ref{alg:binary_search} finds the maximal $(k, \lambda)$-FirmTruss of $G$, which takes $\mathcal{O}(\sum_{\ell \in L} |E_\ell|^{1.5} + |E||L| + |E|\lambda \log |L|)$ time as analyzed in Lemma~\ref{lemma:time_complexity_FirmTruss}. Then, it
iteratively deletes the vertex with the largest query distance starting from $G_0$ with $|E_0|$ edges. In each iteration, the algorithm needs to compute the query distance using $\mathcal{O}(|Q| |E_0|)$ time and to perform FirmTruss maintenance using $\mathcal{O}(\sum_{\ell \in L} |E_\ell|^{1.5})$. Therefore, for $t$ iterations the time cost is $\mathcal{O}(t (|Q| |E_0| + \sum_{\ell \in L} |E_\ell|^{1.5}))$. In total, FTCS-Global algorithm takes $\mathcal{O}(t (|Q| |E_0| + \sum_{\ell \in L} |E_\ell|^{1.5}) + |E||L| + |E|\lambda \log |L|)$ time. The space complexity of FTCS-Global algorithm is dominated by finding $G_0$, which takes $\mathcal{O}(|E||L|)$ space (Lemma~\ref{lemma:time_complexity_FirmTruss}).
\end{proof}

\vspace{-2ex}
\head{Theorem~\ref{theorem:FTCS_local_quality} (FTCS-Local Quality Approximation)}
\begin{proof}
By contradiction, we first show that $dist_{G[H]}(H, Q) \leq dist_{G[H^*]}(H^*, Q)$. Let $dist_{G[H]}(H, Q) > dist_{G[H^*]}(H^*, Q)$, so there exists another FirmTruss with smallest diameter, so when $d = \\dist_{G[H^*]}(H^*, Q)$, the algorithm must find a FirmTruss with query distance of $d$, which is contradiction as $d = dist_{G[H]}(H, Q)$ is the smallest value of $d$ that the algorithm finds a non-empty FirmTruss. Thus, $dist_{G[H]}(H, Q) \leq dist_{G[H^*]}(H^*, Q)$. Based on triangle inequality, $diam(G[H]) \leq 2dist_{G[H]}(H, Q)$. Overall:
\begin{align*}
         & \: diam(G[H]) \leq 2dist_{G[H]}(H, Q) \leq 2 dist_{G[H^*]}(H^*, Q)
     \\ &\: \leq 2diam(G[H^*]).
\end{align*}
\end{proof}

\vspace{-3ex}
\head{Theorem~\ref{theorem:FTCS-local-complexity} (FTCS-Local Complexity)}
\begin{proof}
The proof is similar to analysis of time complexity of FTCS-Global algorithm.
\end{proof}

\head{Theorem~\ref{theorem:FirmTruss_Decomposition_Complexity} (FirmTruss Decomposition Complexity)}
\begin{proof}
By efficient computing of Top$-\lambda$ degree for all values of $\lambda$, lines 4-6 take $\mathcal{O}(|E||L|\log |L|))$ time. Given $\lambda$, we remove each edge schema exactly one time from the buckets, so lines 9-10 take $\mathcal{O}(|E||L|)$. For $\ell \in L$, let $nb^\ell_{\geq}(u) = \{ v | (v, u, \ell) \in E, deg_\ell(v) \geq deg_\ell(u) \}$. So for a given $\lambda$, lines 11-17 take \\ $\mathcal{O}(\sum_{(v, u) \in \mathcal{E}} \sum_{\ell \in L} deg_\ell(u) |nb^\ell_{\geq}(u) |)$. For each layer $\ell$, we know that $|nb^\ell_{\geq}(u)| \leq 2\sqrt{|E_\ell|}$~\cite{truss-algorithm}, so we can conclude that lines 11-17 take $\mathcal{O}(\sum_{\ell \in L} |E_\ell|^{1.5})$. Finally, to update the buckets for a given edge schema we need $\mathcal{O}(|L|)$, so line 20 takes $\mathcal{O}(|E| |L|^2)$ time.
\end{proof}

\head{Theorem~\ref{theorem:AFTCS-Hardness} (AFTCS Hardness)}

\begin{proof}
We reduce the problem of finding the densest subgraph with at least $k$ vertices in single-layer graphs. Here, for simplicity, we assume that $p=1$, we can easily extend this proof to any finite $p$. Given a single-layer graph $G = (V, E)$, construct an instance $G' = (V', E', L, \Psi)$, where $V' = V$, and for each $\ell \in L$, $G'_\ell$ is a complete graph over $V'$. Assume that $V' = \{u_1, u_2, \dots, u_n \}$, where $n = |V'|$. For each vertex $u_i \in V'$, we construct a vector of size $n^2 + n$ as their attribute vector as follows:
\begin{equation*}
    \Psi(u_i)_j = \begin{cases}
     1 & \text{if} \:\: 1 \leq j \leq n^2 \:\: \text{and} \:\: (u_i, u_{j - (i-1)\times n}) \in E \\
     & \text{or} \:\: (u_j, u_{i - (j-1)\times n}) \in E\\
     \sqrt{2n - 2\mathsf{Nb}(u_i)}  & \text{if} \:\: j = n^2 + i\\
     & \\
     0 & \text{otherwise}\\
    \end{cases},
\end{equation*}
where $\mathsf{Nb}(u)$ is the number of $u$'s neighbours in $G$. Now for each two vertices $u$ and $v$ that $(u, v) \in E : h(u,v) = \frac{1}{2n}$ and if $(u, v) \notin E : h(u,v) = 0$. Now, a solution to AFTCS on $G'$ is the densest subgraph with at least $k$ vertices in $G$. Also, each densest subgraph with at least $k$ vertices is a $(k, \lambda)$-FirmTruss with the maximum homophily score. 
\end{proof}

\head{Theorem~\ref{theorem:AFTCS-free-rider} (AFTCS Free-Rider Effect)}

\begin{proof}
Let $G[H]$ be the maximal subgraph among all optimal solutions to the AFTCS problem, and $G[H^*]$ be any query-independent optimal solution. We prove this theorem by contradiction. If $G[H^*] \:\char`\\ \:G[H]\ \neq \emptyset$, $G[H \cup H^*]$ is a connected $(k, \lambda)$-FirmTruss containing $Q$, and $\Gamma_p(G[H \cup H^*]) \geq \Gamma_p(G[H])$, then $G[H \cup H^*]$ is an optimal solution to the AFTCS problem while $G[H] \subseteq G[H \cup H^*]$, violating the maximality of $G[H]$.
\end{proof}




\head{Theorem~\ref{theorem:AFTCS-approx-complexity} (AFTCS-Approx Complexity)}

\begin{proof}
The time complexity of finding $G_0$ is $\mathcal{O}(\sum_{\ell \in L} |E_\ell|^{1.5} + |E||L| + |E|\lambda \log |L|)$ (Lemma~\ref{lemma:time_complexity_FirmTruss}). Let $t$ be the number of iterations, and $V_0$ and $E_0$ be the vertex set and edge set of $G_0$. The computation of homophily score for all vertices takes $\mathcal{O}(d|V_0|^2)$ time, where $d$ is the dimension of the attribute vector. The time complexity of maintenance algorithm is $\mathcal{O}(t|V_0| + |E_0|)$ in $t$ iterations, since each edge in $G_0$ is removed at most one time in all iteration and in each iteration we need to check the Top-$\lambda$ support of each edge schema. Since in each iteration we delete at least one node and its incident edges in all layers, which is at least $(k-1)
\times \lambda$, the total number of iterations $t$ is $\mathcal{O}(\min \{|V_0| - k, \frac{|E_0|}{\lambda . (k-1)} \} )$.
\end{proof}

\head{Theorem~\ref{theorem:AFTCS-approx-quality} (AFTCS-Approx Quality Approximation)}
We first mention two useful lemmas:

\begin{lemma}\label{lemma:power-p}
Given $a, b \in \mathbb{R}^+$ and $a \geq b$, and a real value $p \geq 1$:
\begin{equation*}
    a^p - (a-b)^p \leq p a^{p-1} b.
\end{equation*}
\end{lemma}

\begin{proof}
Given $p \geq 1$ and $a, b \in \mathbb{R}^+$, let $f(x) = x^p, \:\: \forall \: x \in \mathbb{R}$. Based on the mean value theorem, there is a $c \in [a - b, a]$ such that:
$$f'(c) = \frac{f(a) - f(a - b)}{b} \Rightarrow pa^{p - 1} \geq pc^{p-1} = f'(c) = \frac{a^p - (a - b)^p}{b}$$
\end{proof}

\begin{lemma}\label{lemma:supermodular}
Given $p \geq 1$, function $\Xi(S) = |S|\Gamma_p^p(S) = \sum_{v \in S} h_S(v)^p$ is supermodular.
\end{lemma}

\begin{proof}
We know that function $f(x) = x^p$ for $p \geq 1$ and $x \geq 0$ is increasing convex function. Since function $h_S(v)$ is a modular function, its combination with $f$ is supermodular. Accordingly, $p$-th power of $h_S(v)$ is a supermodular function. Finally, we know that the summation of supermodular functions is supermodular. 
\end{proof}

Based on lemmas~\ref{lemma:power-p} and \ref{lemma:supermodular}, we can prove the theorem:

\begin{proof}
Let $H^*$ be the optimal solution and $H$ be the output of Algorithm~\ref{alg:approx-AFTCS}. Since $H^*$ is the optimal solution, removing a node $u^*$ will produce a subgraph with homophily score at most $\Gamma_p^p(H^*)$, therefore, we have:
\begin{align*}
     & \frac{1}{|H^*| - 1} \left( |H^*| \Gamma^p_p(H^*) - \Delta_{u^*}(H^*) \right) \geq \Gamma^p_p(H^*) \\
    \Rightarrow  \:\: & \Gamma^p_p(H^*) \leq \Delta_{u^*}(H^*).
\end{align*}
Note that $u^*$ is an arbitrary vertex in $H^*$ and $H^*\char`\\\{u^*\}$ is not necessarily a $(k, \lambda)$-FirmTruss. Since $|S|\Gamma_p^p(S)$ is a supermodular function for $p \geq 1$, when $H^* \subseteq S$, $\Delta_{u^*}(H^*) \leq \Delta_{u^*}(S)$ (increasing differences). Let $S$ denote the subgraph maintained by the algorithm right before the first node $u^* \in H^*$ is removed. This node cannot be removed by maintaining the subgraph (line 10) as $u^* \in H^*$ and $H^* \subseteq S$ is a $(k, \lambda)$-FirmTruss. Since $u^*$ is the first node to be removed from $H^*$, $H^* \subseteq S$ and so $\Delta_{u^*}(H^*) \leq \Delta_{u^*}(S)$. Since $p \geq 1$, in each iteration we remove a node with minimum value of $\Delta$, $u^* = \arg\min_{u \in S} \Delta_{u}(S)$, so $\Delta_{u^*}(S)$ is smaller than the average value of $\Delta_{u}(S)$ across $u \in S$. Therefore, we have:
\begin{align*}
    &\Gamma_p^p(H^*) \leq \Delta_{u^*}(H^*) \leq \Delta_{u^*}(S) \leq \frac{1}{|S|}  \sum_{u \in S} \Delta_{u}(S)\\
    &= \frac{1}{|S|} \left( \sum_{u \in S} h_S(u)^p + \sum_{u\in S} \sum_{v \in S/\{u\}} h_S(v)^p - [h_S(v) - h(v, u)]^p\right)\\
    &\leq \frac{1}{|S|} \left( \sum_{u \in S} h_S(u)^p + \sum_{u\in S} \sum_{v \in S/\{u\}} p h_S(v)^{p-1}h(v, u)\right)\\
    &= \frac{1}{|S|} \left( \sum_{u \in S} h_S(u)^p + p \sum_{u\in S} h_S(v)^{p}\right) = (p+1) \Gamma_p^p(S)
\end{align*}
The one before last step follows from the inequality in Lemma~\ref{lemma:power-p}. Accordingly, $\Gamma_p(H^*) \leq (p+1)^{1/p} \Gamma_p(S)$. Since $S \in \{G_0, G_1, \dots, G_{i-1} \}$, and $H$ is the output of the algorithm, we have:
\begin{equation*}
    \Gamma_p(H^*) \leq (p+1)^{1/p} \Gamma_p(S) \leq  (p+1)^{1/p} \Gamma_p(H).
\end{equation*}
\end{proof}

\vspace{-2ex}
\head{Theorem~\ref{theorem:Correctness_maxmin_AFTCS} (Correctness of Exact-MaxMin AFTCS)}

\begin{proof}
Let $H$ be the found subgraph by Exact-MaxMin algorithm and $H^*$ be the optimal solution, we prove this theorem by contradiction. If $H \neq H^*$, let $s$ be the first time that a vertex $u^* \in H^*$ is removed. Accordingly, $H^* \subseteq G_s$ and since $u^*$ is removed, it must be removed when we remove a vertex with a minimum value of $h_S$ but $H^*$ is a subgraph with the maximum minimum homophily score. Thus, $\Gamma_{-\infty}(H^*) \leq h_{H^*}(u^*) \leq h_{G_s}(v)$ for any $v \in V[G_s]$. Therefore, $\Gamma_{-\infty}(H^*) \leq  \Gamma_{-\infty}(G_s)$, which is contradiction. 
\end{proof}

\section{Tightness Examples} \label{app:tightness_examples}

\head{Density Lower Bound}
For simplicity assume that $\beta=0$ (for other values of $\beta$ the example is similar), for a given $k$ and $\lambda$, we construct a multilayer graph $G = (V, E, L)$, such that $|L| = \lambda$ and for each $\ell \in L$, $G_\ell$ is a $k$-clique. In this case, $\rho_\beta(G) = \frac{k-1}{2}$ and lower bound for the density of $G$ is $\frac{k-1}{2\lambda}\times \lambda = \frac{k-1}{2}$.  Which shows that the lower bound is tight.

\head{Diameter Upper Bound}
First, the diameter of a connected $k$-truss with $n$ vertices in a single layer graph is no more than $\floor{\frac{2n - 2}{k}}$, and this is a tight upper bound. Therefore, for a given $k$, there exists a single layer $k$-truss, $H_k$, with diameter $\floor{\frac{2n - 2}{k}}$.
Now, for any arbitrary $k$ and $\lambda$, we construct a multilayer graph $G = (V, E, L)$, such that $|L| = \lambda$ and for each $\ell \in L$, $G_\ell = H_k$.   In this case $T=1$ and diameter of G is $\floor{\frac{2|G| - 2}{k}}$.

\head{Edge Connectivity}
For a given $k,\lambda$, we construct a multilayer graph $G = (V, E, L)$, such that $|L| = \lambda$ and for each $\ell \in L$, $G_\ell$ is a $k$-clique.  To make $G$ disconnected, we have to make $G_\ell$ for all $\ell \in L$ disconnected. Since the edge connectivity of a $k$-clique is $k-1$, we have to remove at least $\lambda \times (k-1)$ intra-layer edges to make $G$ disconnected.

\head{Approximation Factor of AFTCS-Approx}
In the proof of Theorem~\ref{theorem:AFTCS-Hardness}, we see the following corollary:

\begin{corollary} \label{cor:1}
Given a set of tuples $\mathcal{E} \subseteq V \times V$, we can construct attribute vectors for each $u \in V$ such that $\forall (u, v) \in \mathcal{E}: h(u, v) = c$, and $\forall (u, v) \not\in \mathcal{E}: h(u, v) = 0$, where $c$ is a constant number.
\end{corollary}

 Given a complete multilayer graph $G = (V, E, L)$,  we want to construct the set of tuples $\mathcal{E} \subseteq V \times V $ such that AFTCS-Approx returns an approximation solution that asymptotically approaches $(p + 1)^{1/p}$ of the exact solution of AFTCS problem.  

For a subgraph $S$, we let $\delta(S) = \min_{u \in S} \Delta_u(S)$. In order to construct the tightness example, we assume that $V = V_1 \cup V_2$, and so $\mathcal{E} = \mathcal{E}_1 \cup \mathcal{E}_2$ such that $\mathcal{E}_1 \subseteq V_1 \times V_1$ and $\mathcal{E}_2 \subseteq V_2 \times V_2$. Now we need to construct each $V_1$ and $V_2$, and so $\mathcal{E}_1$ and $\mathcal{E}_2$ accordingly. Given an arbitrary number $q \in \mathbb{N}$, we let $V_1$ be a set of nodes with size $q + 1$ and also $\mathcal{E}_1 = V_1 \times V_1$. Accordingly, based on Corollary~\ref{cor:1}, we have:
\begin{align}
    &\delta(V_1) = \min_{u \in V_1} \Delta_u(V_1) = (cq)^p + qc^p (q^p - (q - 1)^p), \\
    &\Gamma_p(V_1) = cq 
\end{align}

Now let $k$ and $t$ be two arbitrary integers such that $t < \frac{k}{2}$, we consider a set of nodes $V_2 = \{u_1, u_2, \dots, u_k\}$ such that $V_1 \cap V_2 = \emptyset$, and $|V_2| = k$. In order to construct $\mathcal{E}_2$, for each node $u_i \in V_2$, we add $(u_i, u_j)$ to $\mathcal{E}_2$ for each $i + 1 \leq j \leq \max\{ i + t, k \}$. Accordingly, we can see:
\begin{align}
\nonumber
    \delta(V_2) &= \Delta_{u_1}(V_2) = (ct)^p + c^p  \left(\sum_{i = 2}^{t + 1}  (t + i - 1)^p - (t + i - 2)^p \right) \\
    &= (ct)^p + c^p  \left( (2t)^p - t^p \right)= c^p(2t)^p,
\end{align}
\noindent
and
\begin{align}
    \Gamma_p(V_2) &= \frac{c}{\sqrt[p]{k}} \left(\sum_{i = 1}^{t} (t + i - 1)^p + \sum_{i = t + 1}^{k - t} (2t)^p + \sum_{i = k - t + 1}^{k} (t + i - 1)^p\right)^{1/p}
\end{align}

Note that $q$ is an arbitrary number. So, we let $q = \ceil{\frac{2t}{(p + 1)^{1/p}}}~+~1$, and accordingly, we have: 
\begin{align}
    \delta(V_1) =  (cq)^p + qc^p (q^p - (q - 1)^p) \geq c^p(2t)^p = \delta(V_2).
\end{align}

Now consider a complete attributed multilayer graph $G = (V_1 \cup V_2, E, L, \Psi)$. Let $\mathcal{E} = \mathcal{E}_1 \cup \mathcal{E}_2$, based on Corollary~\ref{cor:1}, we can construct $\Psi$ such that $\forall (u, v) \in \mathcal{E}: h(u, v) = c$, and $\forall (u, v) \not\in \mathcal{E}: h(u, v) = 0$, where $c$ is a constant number. Since $\delta(V_1) \geq \delta(V_2)$, AFTCS-Approx algorithm starts from removing all nodes from $V_2$, while based on the value of $\Gamma_p(V_1)$ and $\Gamma_p(V_2)$, we know that the exact solution is the induced subgraph by $V_2$. Therefore, the found solution by AFTCS-Approx algorithm is the induced subgraph by $V_1$. In order to see the approximation ratio for this example, we need to calculate $\frac{\Gamma_p(V_2)}{\Gamma_p(V_1)}$. To this end, we have:
\begin{align*}
    \frac{\Gamma_p(V_2)}{\Gamma_p(V_1)} &= \frac{ \sqrt[p]{\left(\sum_{i = 1}^{t} (t + i - 1)^p + \sum_{i = t + 1}^{k - t} (2t)^p + \sum_{i = k - t + 1}^{k} (t + i - 1)^p\right)}}{q \times  \sqrt[p]{k}} \\
    &\geq \frac{\sqrt[p]{(1 - \frac{2t}{k})} (2t)}{q} \geq \left(\sqrt[p]{1 - \frac{2t}{k}}\right) \left( \frac{\sqrt[p]{p + 1}}{1 + \frac{\sqrt[p]{p + 1}}{t}} \right) = M_p(k, t)
\end{align*}

Let $t \in \mathcal{O}(k)$ and fixed $p \geq 1$, we have $\lim_{k \rightarrow \infty} M_p(k, t) = \sqrt[p]{p + 1}$, which means the approximation factor converges to $\sqrt[p]{p + 1}$.

\section{Properties of Homophily Score Function}\label{app:homophily_property}

\head{Positive Influence of Similar Vertices}
The first property of homophily score is positive influence of similar nodes. As we discussed, $\Delta_{u}(S)$ is the value that droping $u \in S$ changes the numerator of the homophily score:

\begin{equation*}
    \Gamma_p^p(S/\{u\}) = \frac{1}{|S| - 1} \left( |S| \Gamma^p_p(S) - \Delta_u(S) \right),
\end{equation*}

\noindent
where 

\begin{equation}\label{eq:Xi2}
    \Delta_u(S) = h_S(u)^p + \left( \sum_{v \in S/\{u\}} h_S(v)^p - \left[ h_S(v) - h(v, u) \right]^p \right).
\end{equation}

Similarly, $\Delta_{u}(S \cup \{ u \})$ is the value that adding $u$ to $S$ channges the numerator of the homophily score. Next fact, illustrates when adding a vertex can increase the homophily score.

\begin{fact}\label{fact:positive_influence}
Adding a vertex $u$ to a subgraph $S$ increases the homophily score of subgraph $S$ iff $\Delta_{u}(S \cup \{ u \})^{1/p} > \Gamma_p(S)$.
\end{fact}

\begin{example}
Let $p = 1$, Fact~\ref{fact:positive_influence} states that adding a node $u$ with node-homophily score $h_{S}(u) > \frac{1}{2|S|} \sum_{v \in S} h_S(v)$ to subgraph $S$, increases the homophily score of the subgraph.
\end{example}


\head{Negative Influence of Dissimilar Vertices}
Adding dissimilar vertices to the community decreases the homophily score.

\begin{fact}\label{fact:negative_influence}
Adding a  vertex $u$ to a subgraph $S$ decreases the homophily score of subgraph $S$ iff $\:\Delta_{u}(S \cup \{ u \})^{1/p} < \Gamma_p(S)$.
\end{fact}

\begin{example}
Let $p = 1$, Fact~\ref{fact:negative_influence} states that adding a node $u$ with node-homophily score $h_{S}(u) < \frac{1}{2|S|} \sum_{v \in S} h_S(v)$ to subgraph $S$, decreases the homophily score of the subgraph.
\end{example}

\head{Non-monotone Property}
Non-monotonicity is a desirable property for the community model. Assume that if the attribute correlation measure for the community model is monotone and increasing, then always the maximal FirmTruss is the optimal solution. At the same time, we know that maximal FirmTruss may include some free riders. The introduced homophily score function is non-monotone: as we discussed in Facts~\ref{fact:positive_influence} and \ref{fact:negative_influence}, adding a vertex might increase or decrease the homophily score.

\head{Non-submodularity and Non-supermodularity}
Optimization problems over submodular or supermodular functions lend themselves to efficient approximation. We thus study whether our homophily score function is submodular or supermodular with respect to the set of vertices. By a contradictory example, we show that our homophily score function is neither submodular nor supermodular.

\begin{example}
Consider a graph $G$ with $V = \{u_1, u_2, v_1, v_2\}$. Let $p=1$. Assume that $h(v_1,u_2) = h (v_2,u_1) = 0.1$, $h(v_2,u_2) = 0.5$, $h(v_1,u_1) =0.2$, $h(u_1,u_2) = 0.3$ and $h(v_1,v_2) = 0$. Now consider the sets $S = \{u_1\}$ and $T =  \{u_1,u_2\}$. Let us compare the marginal gains $\Gamma_p(S \cup \{v^*\}) - \Gamma_p(S)$ and  $\Gamma_p(T \cup \{v^*\}) - \Gamma_p(T)$ , from adding the new vertex $v^* \notin T$ to $S$ and $T$. Suppose $v^*=v_1$, then we have $\Gamma_p(S \cup \{v_1\}) - \Gamma_p(S) = (2 \times 0.2)/2- 0 = 0.2 > \Gamma_p(T \cup \{v_1\}) - \Gamma_p(T) = 2(0.1+0.2+0.3)/3 - (2\times0.3)/2= 0.1$, violating supermodularity. On the other hand, suppose $v^*=v_2$. Then we have $\Gamma_p(S \cup \{v_2\}) - \Gamma_p(S) = (2 \times 0.1)/2- 0 = 0.1 < \Gamma_p(T \cup \{v_2\}) - \Gamma_p(T) = 2(0.1+0.5+0.3)/3 - (2\times0.3)/2= 0.3$, which violates submodularity.  Thus, this function is non-submodular and non-supermodular.
\end{example}

\section{Reproducibility}
All algorithm implementations and code for our experiments are available at Github.com/joint-em/FTCS.

\section{Datasets}\label{app:datasets}
We perform extensive experiments on thirteen real networks including social, genetic, co-authorship, financial, and co-purchasing networks, whose main characteristics are summarized in Table~\ref{tab:datastat}. Using the raw Noordin terrorist network data~\cite{Noordin-dataset}, we constructed the multilayer terrorist relationship network. Each node in this network represents a terrorist, and each layer represents a type of relationship (e.g., friendships, organizations, educational institutions, businesses, and religious institutions, etc.). RM~\cite{RM} has 10 networks, each with 91 nodes. Nodes represent phones and one edge exists if two phones detect each other under a mobile network. Each network describes connections between phones in a month. Phones are divided into two communities according to their owners’ affiliations. Brain is derived from the functional magnetic resonance imaging (fMRI) of 520 individuals with the same methodology used in~\cite{Brain_network_fmri}. It contains 190 individuals in the condition group, labelled as ADHD and 330 individuals in the control group, labelled as TD. Here, each layer is the brain network of an individual, nodes are brain regions, and an edge measures the statistical association between the functionality of nodes. Each community is a functional system in the brain. YEAST is a biological network, in which the layers correspond to interaction networks of genes in Saccharomyces cerevisiae (which was obtained via a synthetic genetic-array methodology) and correlation-based networks in which genes with similar interaction profiles are connected to each other~\cite{homo}. FAO represents different types of trade relationships among countries in which layers represent products, nodes are countries and edges at each layer represent import/export relationships among countries~\cite{FAO}. Obama network represent re-tweeting, mentioning, and replying among Twitter users, focusing on Barack Obama's visit in 2013~\cite{Twitter_datasets}. DBLP is a co-authorship network until 2014 that is used in~\cite{MLcore}. Amazon is a co-purchasing network, in which each layer corresponds to one of its four snapshots between March and June 2003~\cite{amazon_datset}. Higgs is a two-layer network where the first layer represents tracking the spread of news about the discovery of the Higgs boson on Twitter, and the second layer is for the following relation~\cite{Higgs}. Friendfeed contains commenting, liking, and following interactions among users of Friendfeed~\cite{Friendfeed}. StackOverflow represents user interactions from the StackExchange site where each layer corresponds to interactions in one hour of the day~\cite{KONECT}. Finally, Google+ is a billion-scale network consisting of 4 Google+ snapshots between July and October 2011 \cite{Google+}. 
\vspace{-2ex}
\section{Additional Experiments}
\label{app:additional_experiments}

\subsection{Scalability of Other Algorithms}
We test our algorithms using different versions of StackOverflow obtained by selecting a variable $\#$layers from 1 to 24 and also with different subsets of edges. Figure~\ref{fig:scalability_appendix} shows the results of the Global, and iLocal algorithms. The running time of these algorithms scales linearly in $\#$layers. By varying $\#$edges, all algorithms scale gracefully. As expected, the iLocal algorithm is less sensitive to varying $\#$edges than $\#$layers.

\begin{figure}[t]
    \centering
    \vspace{-2ex}
    \subfloat[\centering Global]{{\centering \includegraphics[width=0.4\linewidth]{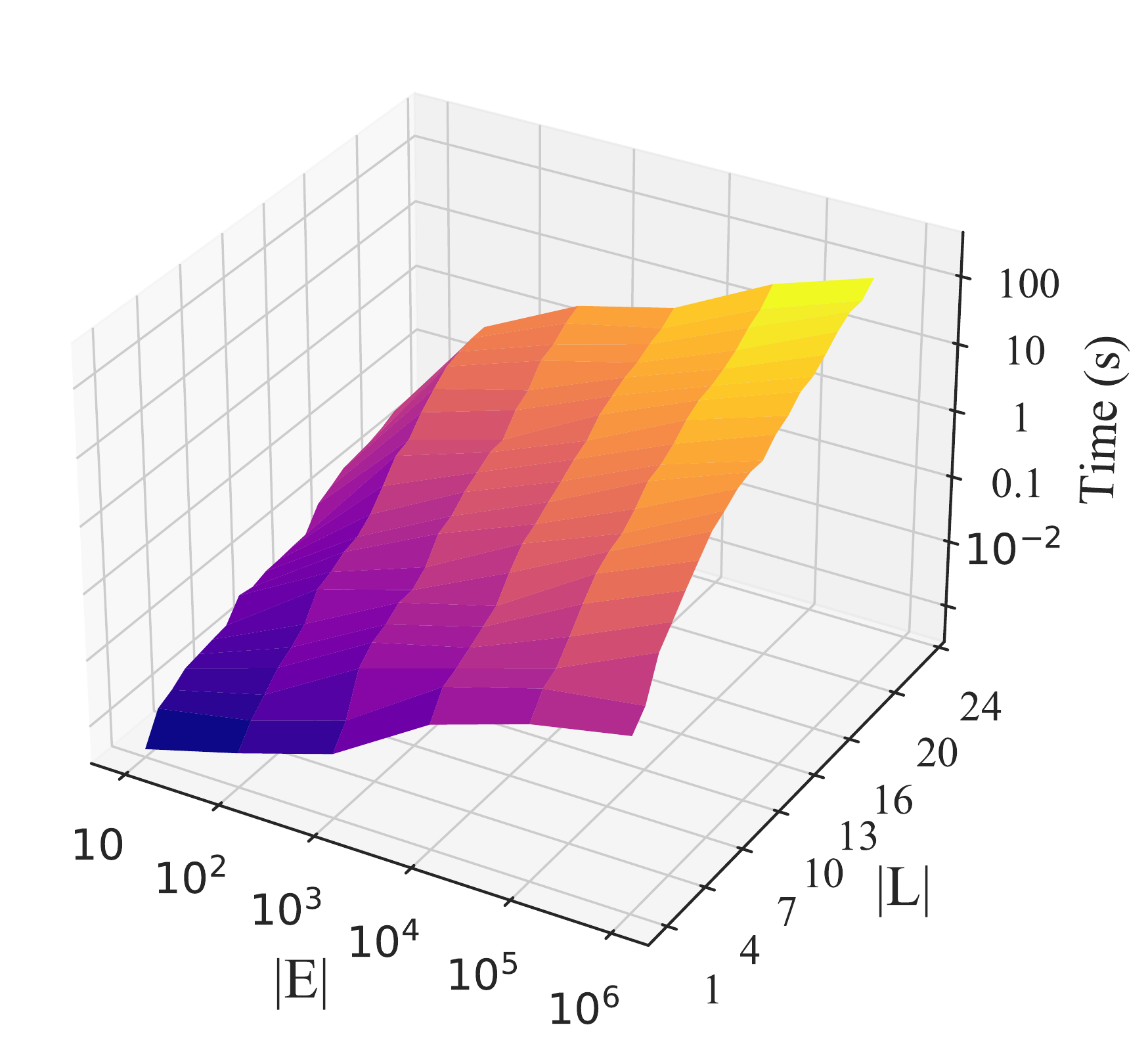}}}
    \hspace{5ex}
    \subfloat[\centering iLocal]{{\centering \includegraphics[width=0.4\linewidth]{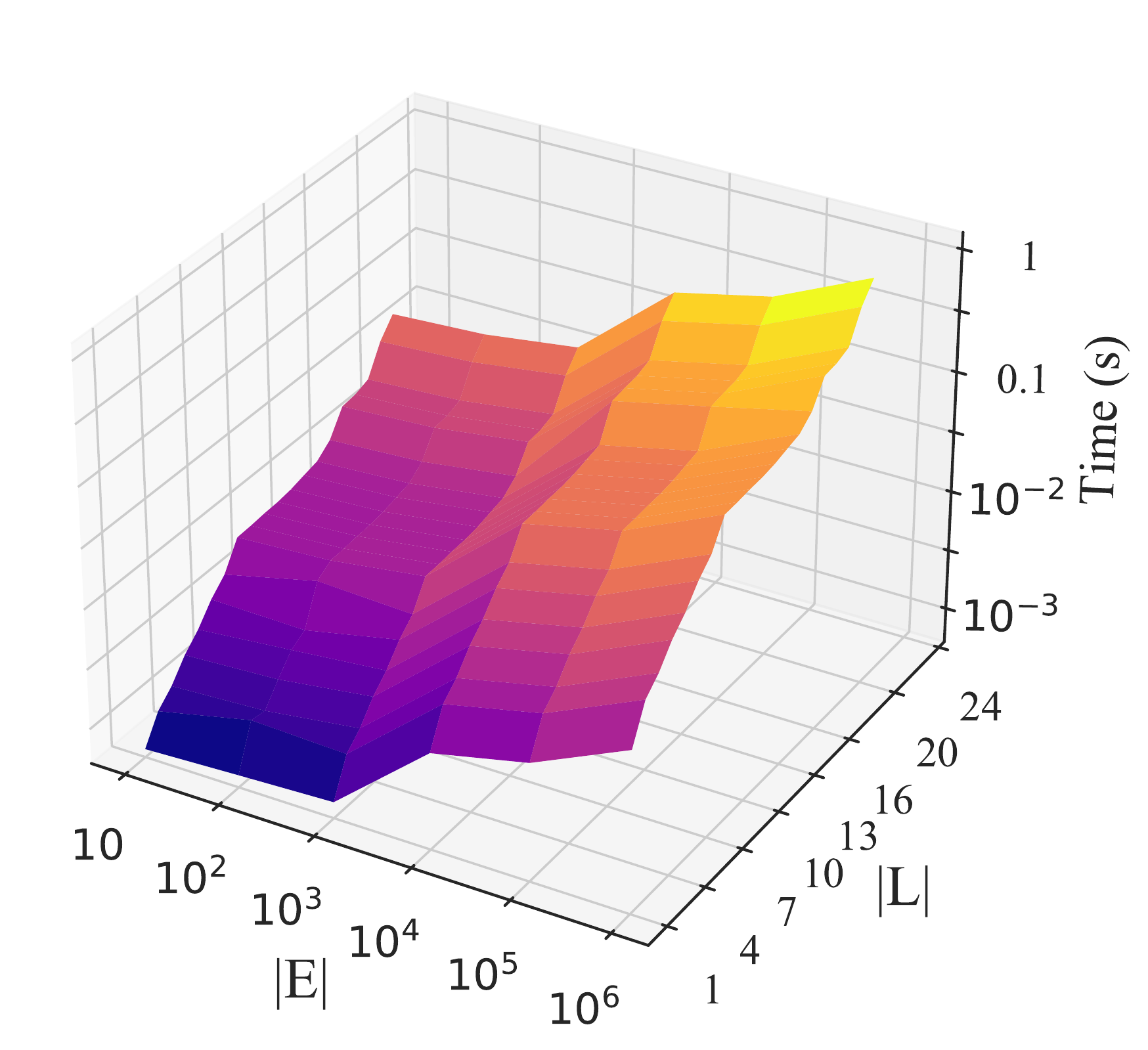}}}
    \caption{Scalability of Global and iLocal algorithms with varying the number of layers and the number of edges.}
    \label{fig:scalability_appendix}
    \vspace{-1ex}
\end{figure}

\begin{table*}
    \caption{Case study of DBLP: Effect of $k$ and $\lambda$ on Quality.}
    \vspace{-1.7ex}
   \resizebox{1\textwidth}{!} {
\begin{tabular}{l l| c c |c c |c c | cc |cc| cc | cc | cc}
 \toprule
  \multicolumn{2}{c|}{\multirow{2}{*}{Query node}} & \multicolumn{2}{c|}{$k = 2$} &  \multicolumn{2}{c|}{$k = 3$} & \multicolumn{2}{c|}{$k = 4$} & \multicolumn{2}{c|}{$k = 5$} & \multicolumn{2}{c|}{$k = 6$} & \multicolumn{2}{c|}{$k = 7$} & \multicolumn{2}{c|}{$k = 8$} & \multicolumn{2}{c}{$k = 9$}\\
  \cline{3-18}
  & & $\lambda = 1$ & $\lambda = 2$ &  $\lambda = 1$ & $\lambda = 2$ & $\lambda = 1$ & $\lambda = 2$ & $\lambda = 1$ & $\lambda = 2$ & $\lambda = 1$ & $\lambda = 2$ & $\lambda = 1$ & $\lambda = 2$ & $\lambda = 1$ & $\lambda = 2$ & $\lambda = 1$ & $\lambda = 2$\\
 \midrule 
 \midrule
    \multirow{4}{*}{Jiawei Han} & $G_0$ Size & 952305 & 952305 &   367471     & 308285  &    231953     & 104779 & 148625  &  29297  & 100171 & 18 & 70236 & 17 & 50111 & - & 37476 & -\\
                           & Community Size & \textbf{440} & \textbf{440} &   \textbf{284}   &  \textbf{264} &  \textbf{225}    & \textbf{102} & \textbf{169} & \textbf{34}  & \textbf{113} & \textbf{10} & \textbf{83} & \textbf{10} & \textbf{60} & - & \textbf{34} & -\\
                           & \#Free-riders & 951865 & 951865 &    367187    & 308021  &    231728     & 104677 & 148456 &  29233  & 100058 & 8 & 70153  & 7 & 50051 &-  & 37442 & - \\
                           & \new{F1-Score} & \new{0.27} & \new{0.27} &\new{0.39}&\new{0.41}&\new{0.47}&\new{0.78}&\new{0.57}&\new{0.63}&\new{0.75}&\new{0.26}&\new{0.85}&\new{0.26}&\new{0.94}&\new{-}&\new{0.63}&\new{-}\\
    \midrule
    \multirow{4}{*}{Samuel Madden} & $G_0$ Size & 952305 & 952305 &    367471    &  308285 &    231953     & 104779 & 148625 &  29297  & 100171 & 5107 & 70236 & 17 & 50111 & - & 37476 & -\\
                           & Community Size & \textbf{210} & \textbf{210} &    \textbf{173}    & \textbf{152}  &    \textbf{144}     & \textbf{99} & \textbf{122} &  \textbf{59} & \textbf{113} & \textbf{17} & \textbf{100} & \textbf{17} & \textbf{83} & - & \textbf{67} & -\\
                           & \#Free-riders & 952095 & 952095 &    367298    &  308133  &    231809    & 104680 & 148543 & 29238  & 100058 & 5090 & 70136 & 0 & 50028 & - & 37409 & -\\
                           & \new{F1-Score} & \new{0.44} & \new{0.44} &\new{0.52}&\new{0.57}&\new{0.59}&\new{0.75}&\new{0.66}&\new{0.97}&\new{0.69}&\new{0.44}&\new{0.75}&\new{0.44}&\new{0.84}&\new{-}&\new{0.93}&\new{-}\\
    \midrule
    \multirow{4}{*}{Yoshua Bengio} & $G_0$ Size & 952305 & 952305 &   367471     & 308285  &    231953     & 104779 & 148625 & 29297   & 100171 & 17 & 70236 & 17  & 50111 & 17 & 37476 & 17\\
                           & Community Size & \textbf{116} & \textbf{116} &    \textbf{109}    & \textbf{69}  &    \textbf{81}     & \textbf{34} & \textbf{58} &  \textbf{23} & \textbf{36} & \textbf{17} & \textbf{23} & \textbf{17} & \textbf{17} & \textbf{17}& \textbf{17} & \textbf{17}\\
                           & \#Free-riders & 952189 & 952189 &    367362    &  308216  &    231872    & 104745 & 148567 & 29274  & 100135 & 0 & 70213 & 0 & 50094 & 0 & 37459 & 0\\
                           & \new{F1-Score} & \new{0.32} & \new{0.32} &\new{0.34}&\new{0.48}&\new{0.43}&\new{0.75}&\new{0.55}&\new{0.98}&\new{0.72}&\new{0.89}&\new{0.98}&\new{0.89}&\new{0.89}&\new{0.89}&\new{0.89}&\new{0.89}\\
  \toprule
\end{tabular}
}
\vspace{-3ex}
 \label{tab:param_quality_dblp}
\end{table*}

\subsection{Effect of Query Set Size}
\head{Effectiveness}
We evaluate the effect of the query set size on the quality of found communities by FTCS and AFTCS. Figure~\ref{fig:query_set_size_F1}(a) and (b) show the F1-score as a function of query set size, $|Q|$, on Terrorist and RM datasets. We observe that increasing the size of the query set first results in increasing F1-score, but after it achieves its maximum, the F1-score decreases. \new{The reason is increasing $|Q|$ increases F1-score as long as $|Q|$ is less than the size of the ground-truth community. If $|Q|$ exceeds the size of the community, the F1-score starts to decrease.}

\begin{figure}
    \centering
    \hspace*{-2ex}
    \subfloat[\centering Terrorist dataset]{{\hspace*{-2ex}\includegraphics[width=0.20\textwidth]{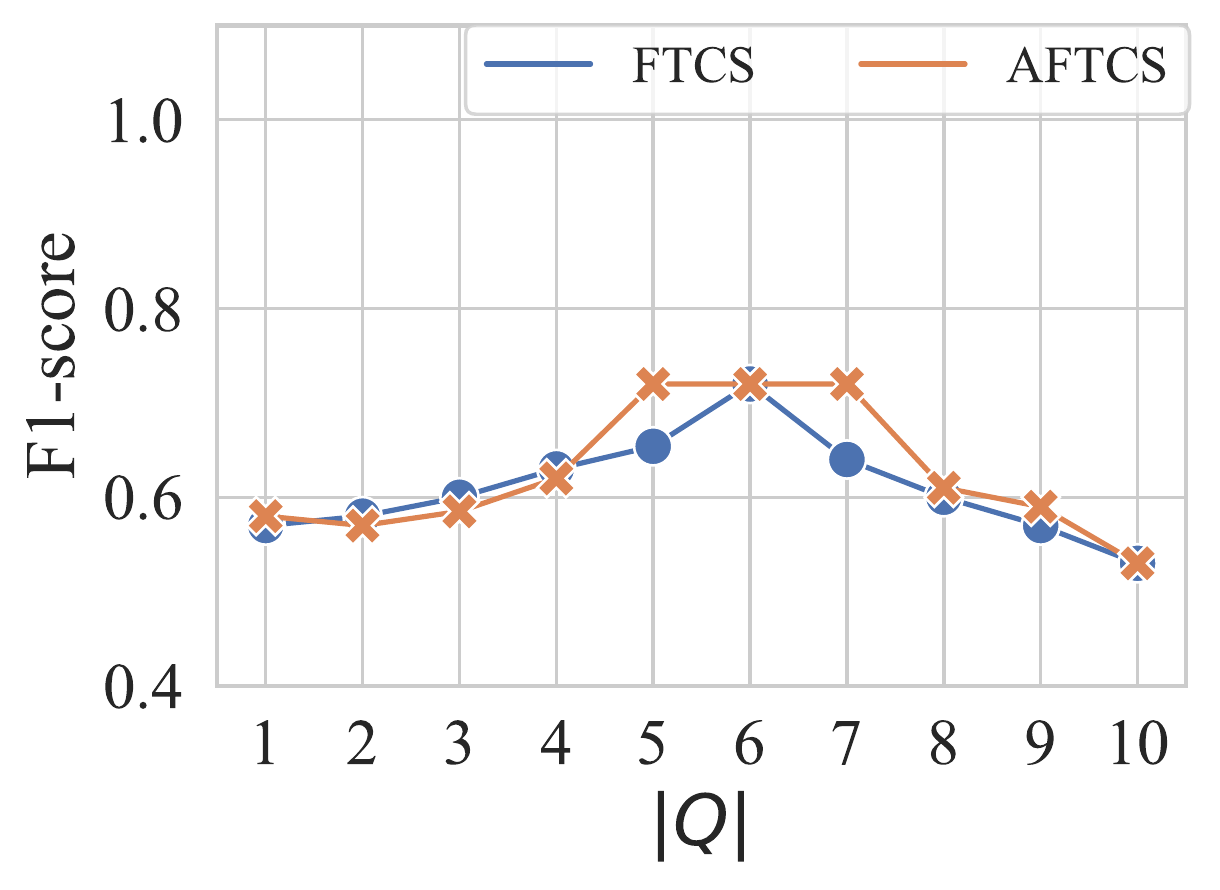} }}
    \hspace{4ex}
    \subfloat[\centering RM dataset]{{\hspace*{-1ex}\includegraphics[width=0.20\textwidth]{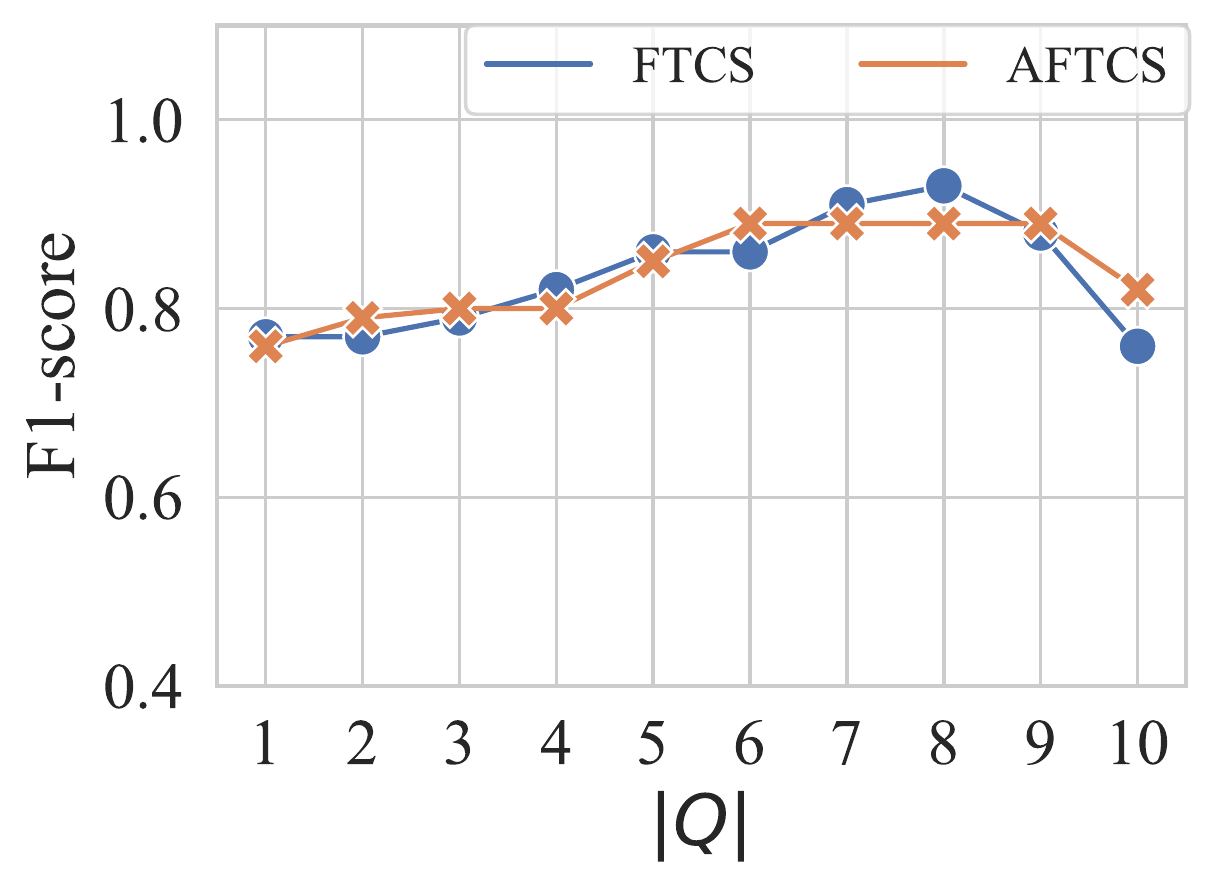} }}
    \vspace{-2ex}
    \caption{Effect of the query set size, |Q|, on F1-score.}
    \label{fig:query_set_size_F1}
        \vspace{-4ex}
\end{figure}

\head{Efficiency}
We evaluate the effect of the query set size on the efficiency of proposed algorithms. Figure~\ref{fig:query_set_size_running_time}(a) and (b) show the query time as a function of query set size, $|Q|$, on Terrorist and DBLP datasets. We observe that increasing the size of the query set results in increasing the running time. \new{Increasing the query set size results in increasing the running time as our algorithms need to explore more nodes in the neighborhood of query nodes.}

\begin{figure}
    \centering
    \hspace*{-2ex}
    \subfloat[\centering Terrorist dataset]{{\hspace*{0ex}\includegraphics[width=0.23\textwidth]{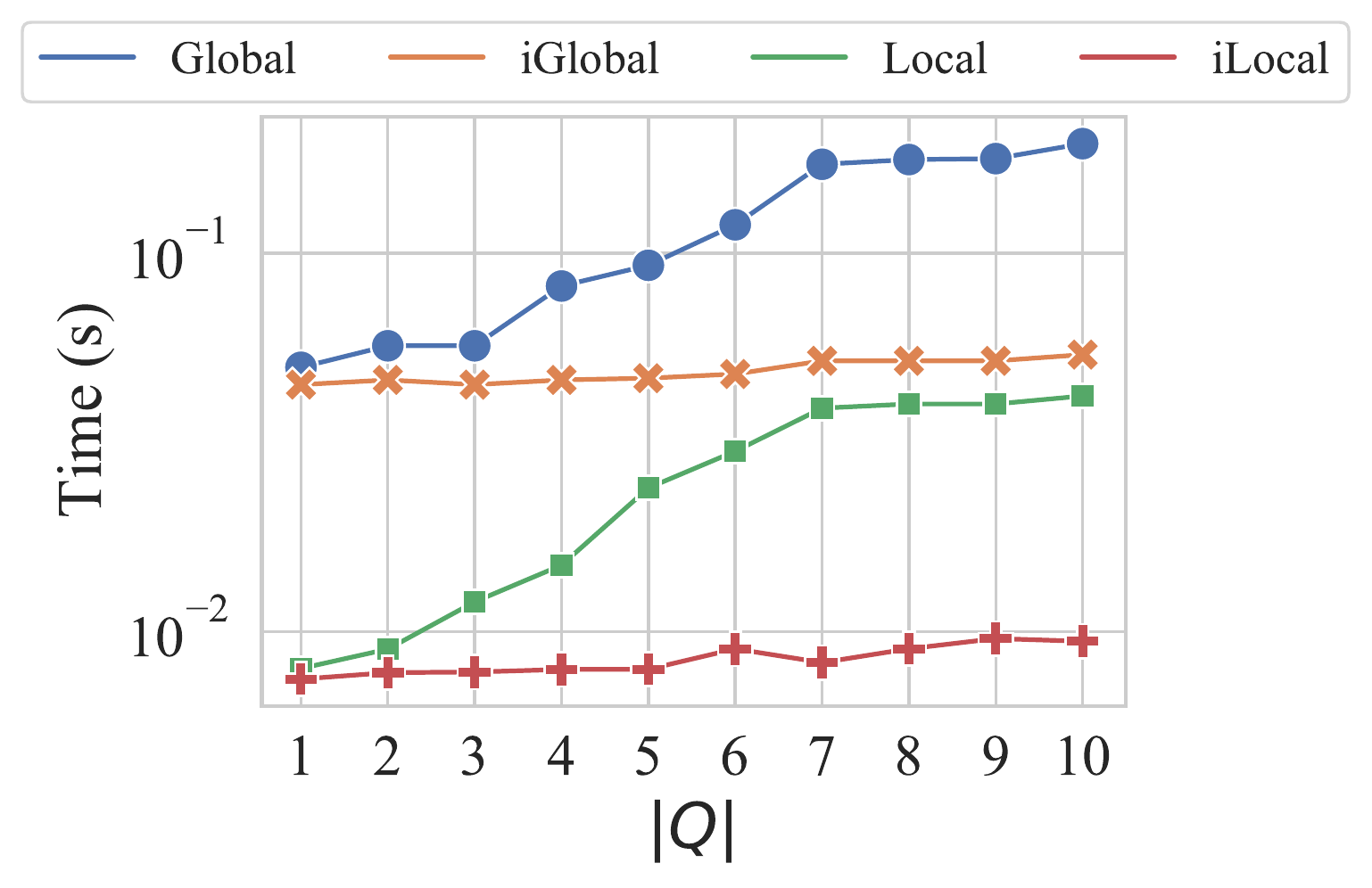} }}
    \hspace{2ex}
    \subfloat[\centering DBLP dataset]{{\hspace*{-1ex}\includegraphics[width=0.23\textwidth]{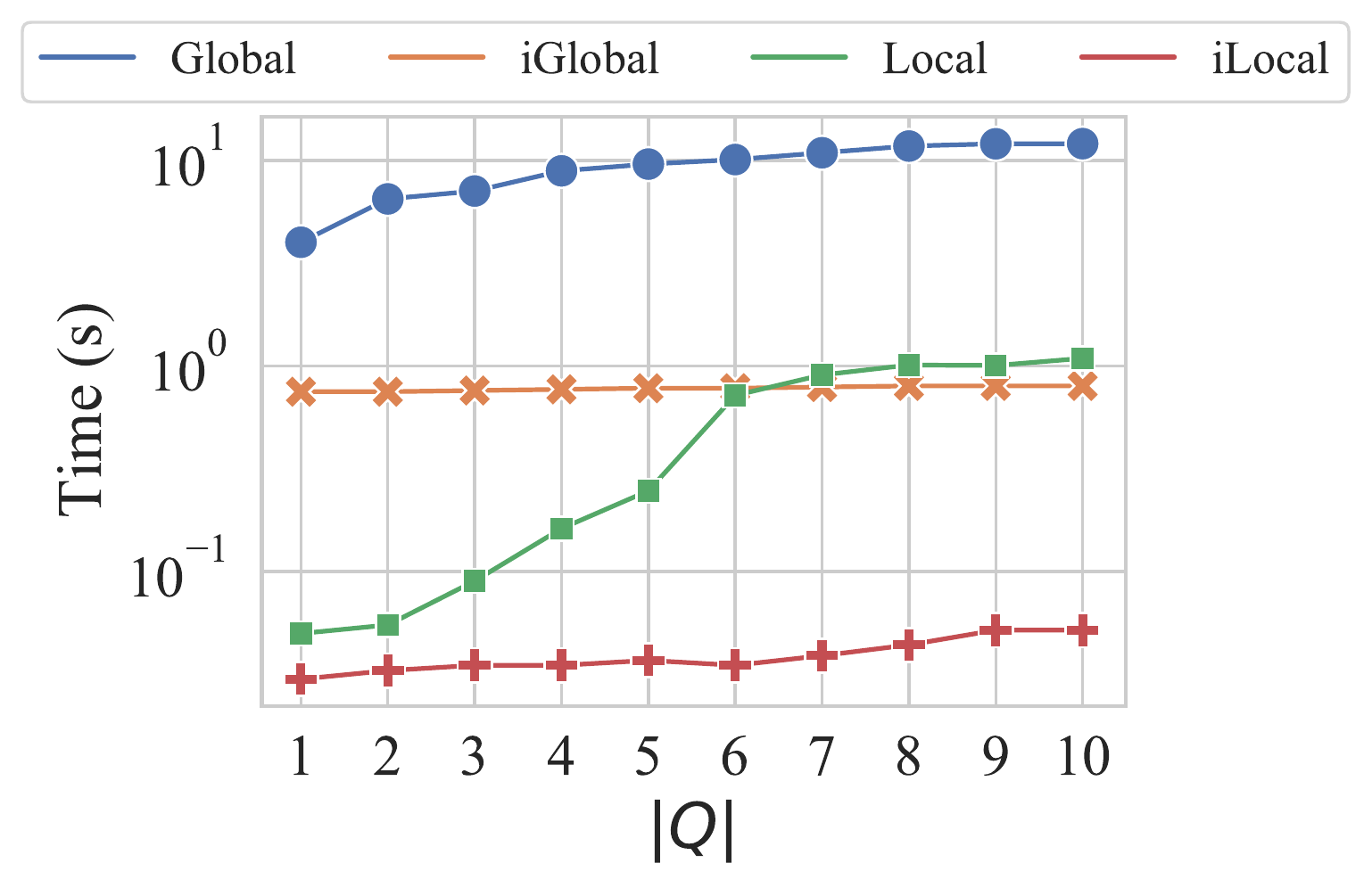} }}
    \vspace{-2ex}
    \caption{Effect of the query set size, |Q|, on running time.}
    \label{fig:query_set_size_running_time}
    \vspace{-5ex}
\end{figure}

\subsection{Effect of $k$ and $\lambda$ on Efficiency}
Figure~\ref{fig:param_run_brain} shows the result of the effect of varying $k$ and $\lambda$ on running time. As we discussed earlier, the larger $k$ and $\lambda$ for Global algorithms results in less running time since algorithms generate a smaller $G_0$. However, the larger $k$ and $\lambda$ increases the running time of Local algorithms since it needs to count more nodes in the neighbourhood of query nodes to find a $(k, \lambda)$-FirmTruss.

\begin{figure}
    \centering
    \hspace*{-2ex}
    \subfloat[\centering Effect of $k$]{{\hspace*{-2ex}\includegraphics[width=0.23\textwidth]{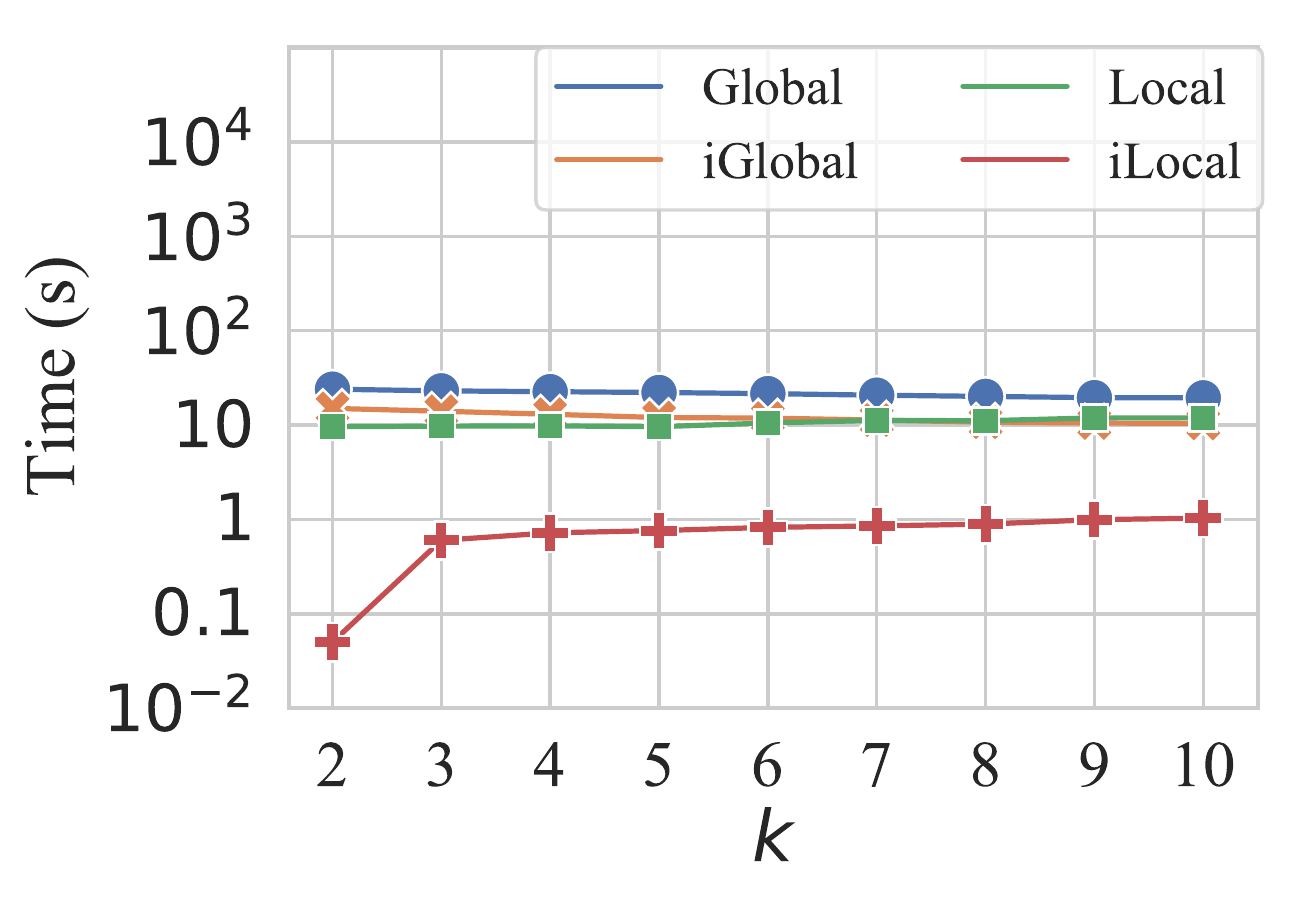} }}
    \hspace{4ex}
    \subfloat[\centering Effect of $\lambda$]{{\hspace*{-1ex}\includegraphics[width=0.23\textwidth]{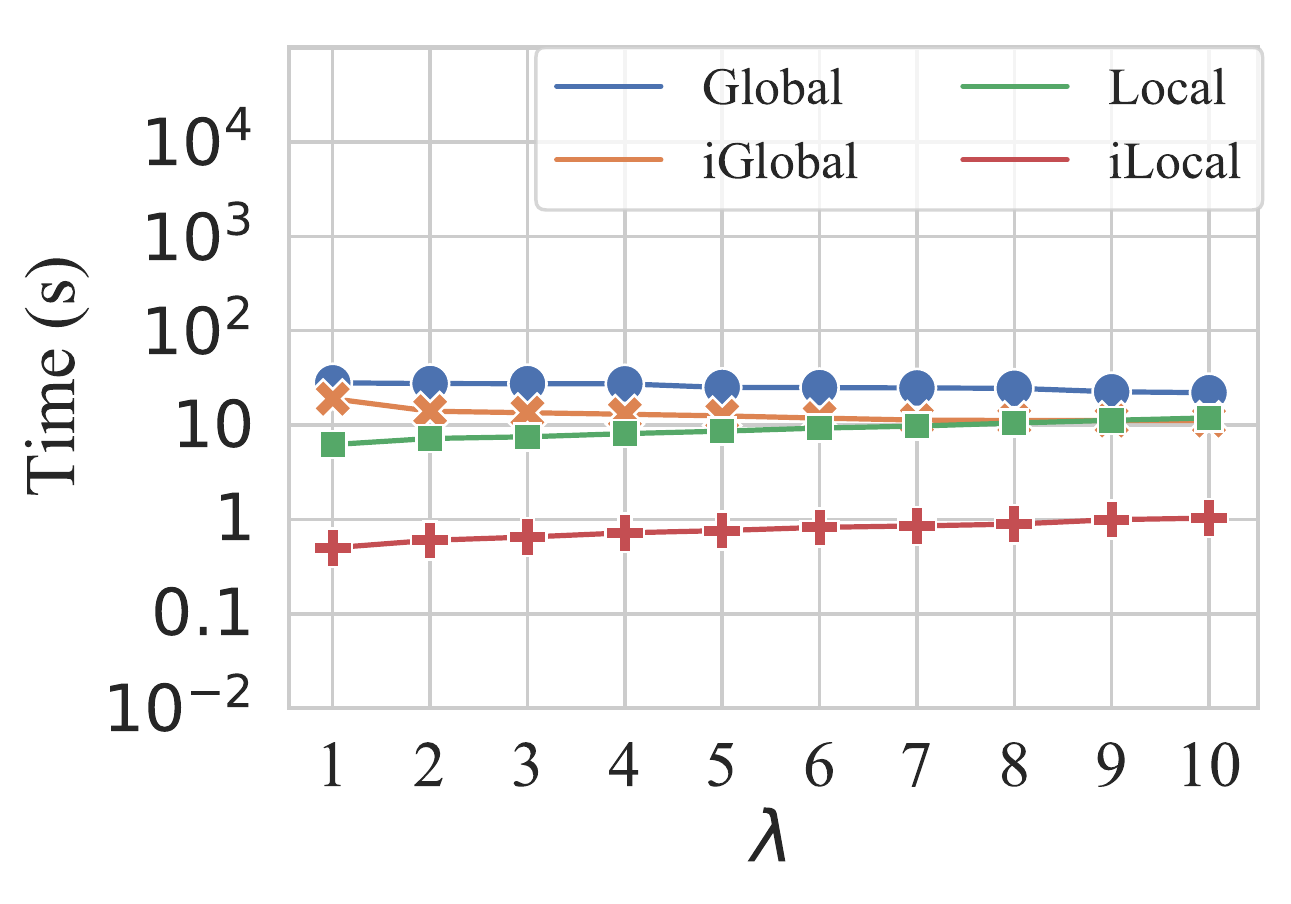} }}
    \vspace{-2ex}
    \caption{Effect of $k$ and $\lambda$ on running time (Brain).}
    \label{fig:param_run_brain}
\end{figure}

\subsection{Effect of $k$ and $\lambda$ on Found Communities}\label{sec:param_quality}
Since $k$ and $\lambda$ are given by the user, one might ask what happens if inappropriate values of $k$ and $\lambda$ are given? We conduct a case study on the DBLP dataset to evaluate the effect of $k$ and $\lambda$, and to show the effectiveness of the FTCS in removing free riders. Table~\ref{tab:param_quality_dblp} shows the $G_0$ size, the size of found community by FTCS, and the \# free-riders for different values of $k$ and $\lambda$, and for different query nodes, ``Jiawei Han'', ``Samuel Madden'', ``Yoshua Bengio''. These results show that giving inappropriate values of $k$ and $\lambda$ results in increasing the size of $G_0$. However, since FTCS minimizes the diameter, it is able to remove far nodes from the query node and always finds a community with appropriate size and quality. Accordingly, FTCS is robust to inappropriate values of $k$ and $\lambda$. 

 We know that given $\lambda~\in~\mathbb{N}^+$ and $k \geq 0$, the $(k + 1, \lambda)$-FT and $(k, \lambda + 1)$-FT of a multilayer graph are subgraphs of its $(k, \lambda)$-FT. Another property of FTCS that is shown in Table~\ref{tab:param_quality_dblp} is its hierarchical structure. This is an important property for a community model as by varying the value of $k$ and $\lambda$ it can find communities at different levels of granularity.

\begin{remark}
\new{
How is a user supposed to choose the parameters? The reported experiments on the effect of $k$, $\lambda$, and $p$ offer some guidance. Figure~\ref{fig:param_sensitivity} in the paper shows the effect of $k$, $\lambda$, and $p$ on the efficiency. Based on this experiment, the value of $p$ does not affect the efficiency much and the running times are comparable for different values of $p$. Also,  our iLocal algorithm shows a stable trend, whereby varying $k$ and $\lambda$ has a negligible impact on the running time. On the other hand, in the Global and iGlobal algorithms we see a decreasing trend and in the Local algorithm, we see an increasing trend, as $k, \lambda$ are increased. Thus, when an index is available, iLocal is clearly the algorithm of choice as it has a stable performance fairly independent of the parameters.  Otherwise,  one can use the Global algorithm for large values of $k$ and $\lambda$ and use Local algorithm for small values~of~$k$~and~$\lambda$. }
\new{
Table~\ref{tab:param_quality_dblp} shows the effect of $k$ and $\lambda$ on the  quality of the communities found from the DBLP dataset. We measured the size of the largest connected $(k,\lambda)$-truss found, the number of free-riders removed, the final community size, as well as the F1 score w.r.t. the ground truth. These results show the robustness of our approach even in many cases when different values of $k$ and $\lambda$ are chosen. }

\new{
In sum, a user can start with some initial values for the parameters $p, k, \lambda$, subsequently modify them, and explore the various extracted communities and choose the one that best matches their needs. Our experiments show that either using iLocal (if index is available) or using Global (if $k, \lambda$ are large), or Local (if $k, \lambda$ are small), each resulting community can be extracted in about a second, permitting efficient exploration. In terms of quality, our experiments show that even the worst choice of $p$ leads to a quality that is better than all the baselines. 
}
\end{remark}

\vspace{-3ex}
\subsection{Efficiency of FirmTruss Decomposition}
We compare the efficiency of FirmTruss decomposition algorithm with the state-of-the-art dense strcture mining algorithms, FirmCore~\cite{FirmCore}, ML \textbf{k}-core~\cite{MLcore}, and TrussCube~\cite{Truss_cube}, in multilayer graphs. Figure~\ref{fig:Runing_time_FirmTruss} shows the running time of these algorithms on different datasets. FirmTruss outperforms all other algorithms, except FirmCore, and can be scaled to graphs with hundred thousands of edges. As we previously discussed, although FirmCore is more efficient that FirmTruss, FirmTrusses are  denser than FirmCores. It is a trade-off between  density and efficiency. 

\begin{figure}
    \centering
    \hspace*{-3ex}
    \includegraphics[width=0.49\textwidth]{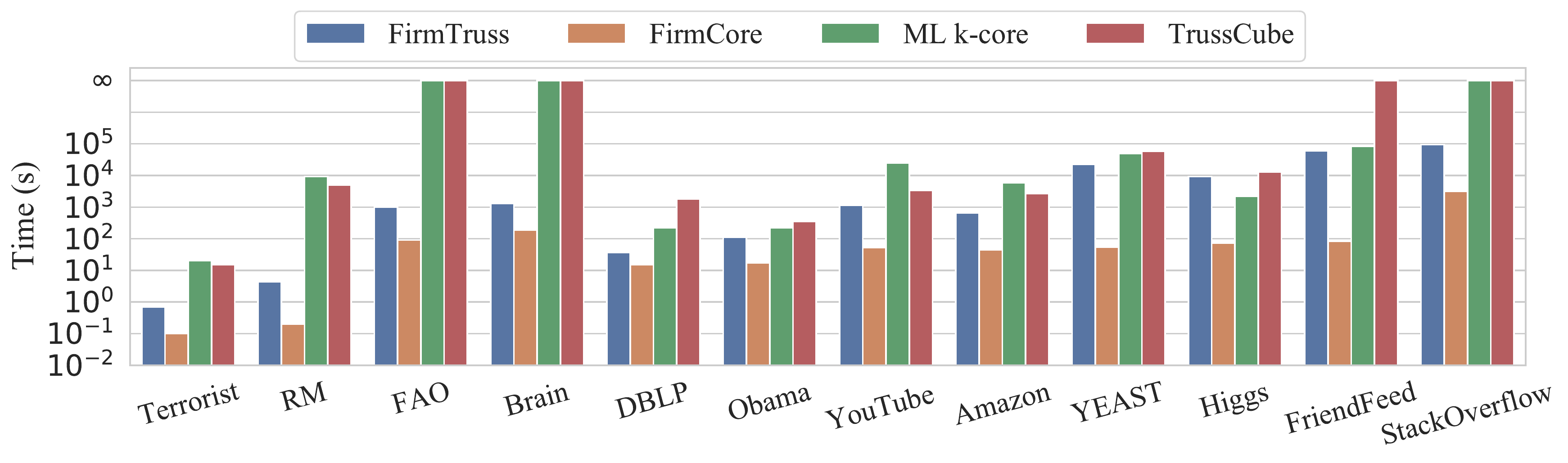}
    \vspace{-4ex}
    \caption{Efficiency Evaluation of FirmTruss.}
    \label{fig:Runing_time_FirmTruss}
\end{figure}

\section{Comparison of Different ML Community Search Models}

\begin{example}
\new{In Figure~\ref{fig:example}, the green nodes are densely connected and intuitively they form a community. Let $v_{10}$ be the query node, we expect a good CS model to return the green nodes as the corresponding community of $v_{10}$. Let $k = 4$ and $\lambda = 1$, we can see that FTCS returns the green nodes as a community. For TrussCube,\footnote{There is no prior community search model based on TrussCube. Here we consider a baseline that finds the maximal connected TrussCube containing query nodes.} when the number of selected layer is 2, or 3, it returns the entire graph as the community of $v_{10}$, which includes many false positive nodes. Note that, choosing only 1 layer is equivalent to a simple truss in a single layer graph. Let $\beta = 1$, then  $k$-core returns a subgraph including $\{v_1, v_2, v_3, v_8, v_9,v_{10}, v_{11}, v_{13}\}$  as the community, which misses several nodes in the expected community (i.e., the green nodes). Let $\beta = 3$, ML $k$-core returns a subgraph including \eat{$\{\}$}all nodes as the community, which includes many false positives. Finally, ML-LCD returns $\{v_9, v_{10}, v_{11}\}$ as the corresponding community, which misses several green nodes.}
\end{example}


\eat{




\subsection{Case Study of FAO}

\subsection{Case Study of Gene Network}

\subsection{Case Study of Airlines Network}

}

\end{document}